\documentclass{lmcs} 
\pdfoutput=1
\usepackage[utf8]{inputenc}

\usepackage{lastpage}
\lmcsdoi{21}{4}{18}
\lmcsheading{}{\pageref{LastPage}}{}{}%
{May~28,~2024}{Nov.~04,~2025}{}

\usepackage{cite}
\usepackage{amsmath,amssymb,amsfonts}
\usepackage{algorithmic}
\usepackage{graphicx}
\usepackage{textcomp}
\usepackage{xcolor}
\usepackage{float}
 \usepackage{cmll}
\def\BibTeX{{\rm B\kern-.05em{\sc i\kern-.025em b}\kern-.08em
    T\kern-.1667em\lower.7ex\hbox{E}\kern-.125emX}}
\usepackage{hyperref}
\usepackage{todonotes}
\setuptodonotes{inline}
\usepackage{cleveref} 
\usepackage{virginialake}
\usepackage{tikz}
\usetikzlibrary{positioning}
\usepackage{microtype}

\renewcommand{\emptyset}{\varnothing}
\renewcommand{\phi}{\varphi}
\newcommand{\dfn}{:=}
\newcommand{\bnf}{::=}
\newcommand{\oldtodo}[1]{\todo{#1}} 

\newcommand{\red}[1]{{\color{red}#1}}
\newcommand{\blue}[1]{{\color{blue}#1}}
\newcommand{\orange}[1]{{\color{orange}#1}}

\definecolor{mygreen}{rgb}{0, 0.5, 0}
\newcommand{\mygreen}[1]{{\color{mygreen}#1}}

\newcommand{\cyan}[1]{{\color{cyan}#1}}

\newcommand{\comment}[1]{\blue{#1}}
\newcommand{\anupam}[1]{\todo[color=green]{A: #1}}
\newcommand{\gianluca}[1]{\todo[color=purple]{G: #1}}

\renewcommand{\todo}[1]{}
\renewcommand{\comment}[1]{}
\renewcommand{\anupam}[1]{}
\renewcommand{\gianluca}[1]{}

\newcommand{\todonew}[1]{{\red{TodoNew: #1}}}

\renewcommand{\todonew}[1]{}

\theoremstyle{plain}
\newtheorem{theorem}[thm]{Theorem}
\newtheorem{proposition}[thm]{Proposition}
\newtheorem{lemma}[thm]{Lemma}
\newtheorem{corollary}[thm]{Corollary}

\newtheorem{example}[thm]{Example}

\theoremstyle{definition} 
\newtheorem{definition}[theorem]{Definition}
\newtheorem{defn}[thm]{Definition}
\newtheorem{remark}[thm]{Remark}
\newtheorem{convention}[thm]{Convention}
\newtheorem{question}[thm]{Question}

\renewcommand{\succ}{\mathsf s}
\newcommand{\pred}{\mathsf{p}}
\newcommand{\FV}{\mathrm{FV}}
\newcommand{\fv}[1]{\FV(#1)}

\newcommand{\numeral}[1]{\underline{#1}}
\newcommand{\cod}[1]{\numeral{#1}} 

\newcommand{\sem}[1]{\llbracket {#1} \rrbracket}

\newcommand{\Nat}{\mathbb{N}}
\newcommand{\Ord}{\mathrm{Ord}}

\newcommand{\Var}{\mathrm{Var}}
\newcommand{\var}{\mathrm{var}}
\newcommand{\Types}{\mathrm{Typ}}

\newcommand{\lambdaterms}{\Lambda}

\newcommand{\closure}[1]{\langle #1 \rangle}
\newcommand{\clomuLJneg}{\closure{\muLJneg}}
\newcommand{\clocmuLJneg}{\closure{\cmuLJneg}}

\newcommand{\typedclomuLJneg}{\clomuLJneg_{\N,\times,\arrow,\mu}}

\newcommand{\Sys}[1]{\mathcal E_{#1}}

\newcommand{\nat}{\mathit{N}}
\newcommand{\N}{\nat} 
\newcommand{\unit}{1}
\newcommand{\arrow}{\to}
\newcommand{\imp}{\arrow} 

\newcommand{\subform}{\subseteq}

\newcommand{\FL}{\mathrm{FL}}
\newcommand{\fl}[1]{\FL(#1)}
\newcommand{\flleq}{\preceq_\FL}
\newcommand{\fleq}{\approx_\FL}
\newcommand{\flnleq}{\prec_\FL}

\newcommand{\seqar}{\Rightarrow}

\newcommand{\mulunf}{\lr \mu'}
\newcommand{\nurunf}{\rr \nu'}

\newcommand{\iter}[1]{\mathsf{iter}_{#1}}
\newcommand{\inj}{\mathsf{in}}
\newcommand{\coiter}[1]{\mathsf{coiter}_{#1}}

\newcommand{\injX}[1]{\inj_{#1}}

\newcommand{\cnd}{\mathsf{cond}}

\newcommand{\Nzero}{\rr\N^0}
\newcommand{\Nsucc}{\rr\N^1}
\newcommand{\Ncnd}{\lr \N'}
\newcommand{\Niter}{\iter \N}

\newcommand{\infrule}{\mathsf{r}}
\newcommand{\rrule}{\infrule} 

\newcommand{\id}{\mathsf{id}}
\newcommand{\ax}{\id}
\newcommand{\ex}{\mathsf{e}}
\newcommand{\cut}{\mathsf{cut}}
\newcommand{\wk}{\mathsf{w}}
\newcommand{\contr}{\mathsf{c}}

\newcommand{\choice}[2]{#2^{#1}}

\newcommand{\pair}[2]{\langle #1,#2\rangle}

\newcommand{\lr}[1]{#1_l}
\newcommand{\rr}[1]{#1_r}

\newcommand{\cev}[1]{#1^0} 

\newcommand{\LJ}{\mathsf{LJ}}
\newcommand{\LJneg}{\LJ^-}
\newcommand{\T}{\mathsf{T}}
\newcommand{\F}{\mathsf{F}}
\newcommand{\muLJ}{\mu\LJ}
\newcommand{\muLJnorec}{\mu'\LJ}
\newcommand{\muLJneg}{\muLJ^-}
\newcommand{\muLJnegnorec}{\muLJnorec^-}

\newcommand{\circular}{\mathsf C}
\newcommand{\cmuLJ}{\circular\muLJ}
\newcommand{\cmuLJneg}{\cmuLJ^-}
\newcommand{\MALL}{\mathsf{MALL}}
\newcommand{\muMALL}{\mu\MALL}

\newcommand{\proj}[1]{\mathsf{p}_{#1}}
\newcommand{\projl}{\proj 0}
\newcommand{\projr}{\proj 1}

\newcommand{\enc}[1]{\underline{#1}}
\newcommand{\true}{\mathsf{tt}}
\newcommand{\false}{\mathsf{ff}}

\newcommand{\ifthen}[3]{\mathsf{if \ }{#1}\mathsf{\ then\ }{#2}\mathsf{\ else\ }{#3}}

\newcommand{\rewrite}{\rightsquigarrow}

\newcommand{\cutreduction}{\mathrm{cr}}
\newcommand{\cutred}{\rewrite_{\cutreduction}}

\newcommand{\cutreductionnorec}{\mathrm{cr}'}
\newcommand{\cutrednorec}{\rewrite_{\cutreductionnorec}}

\newcommand{\reduction}{\mathrm{r}}
\newcommand{\reduces}{\rewrite_{\reduction}}
\newcommand{\converts}{=_{\reduction}}

\newcommand{\reductionnorec}{\mathrm{r}'}
\newcommand{\reducesnorec}{\rewrite_{\reductionnorec}}
\newcommand{\convertsnorec}{=_{\reductionnorec}}

\newcommand{\reducesnoreceta}{\rewrite_{\reductionnorec}^\eta} 
\newcommand{\convertsnoreceta}{=_{\reductionnorec}^\eta} 
\newcommand{\extensional}{\eta} 

\newcommand{\betae}{=_{\beta}}

\newcommand{\betaeta}{=_{\beta, \eta}}

\renewcommand{\lor}{\vee}
\renewcommand{\land}{\wedge}
\newcommand{\limp}{\to}
\newcommand{\liff}{\leftrightarrow}
\newcommand{\mufo}[3]{\mu #1 \lambda #2 #3}
\newcommand{\nufo}[3]{\nu #1 \lambda #2 #3}

\newcommand{\Lang}{\mathcal L}
\newcommand{\lang}[1]{\Lang_{#1}}
\newcommand{\langarith}{\lang 1}
\newcommand{\langsoarith}{\lang 2}
\newcommand{\langmuarith}{\lang \mu}

\newcommand{\proves}{\vdash}

\newcommand{\der}{P}
\newcommand{\deri}{Q}
\newcommand{\derii}{R}
\newcommand{\derr}{\deri} 
\newcommand{\derrr}{\derii} 

\newcommand{\metader}{\mathcal{D}}

\newcommand{\axiom}[2]{  \vlin{\ax}{}{{#1} \seqar {#2}}{\vlhy{}}}
\newcommand{\haxiom}[2]{ \vlhy{{#1} \seqar {#2}}}
\newcommand{\uaxiom}{\vlin{\unit}{}{\seqar \unit}{\vlhy{}}}

\newcommand{\vldr}[2]{\vltr{#1}{#2}{\vlhy{\ \ \ }}{\vlhy{\ \ \ }}{\vlhy{\ \ \ }}}

\newcommand{\vldrs}[2]{\vltr{#1}{#2}{\vlhy{ \ }}{\vlhy{\ }}{\vlhy{\  }}}

\newcommand{\Sin}[2]{\Sigma^{#1}_{#2}}
\newcommand{\Pin}[2]{\Pi^{#1}_{#2}}
\newcommand{\Din}[2]{\Delta^{#1}_{#2}}

\newcommand{\posSin}[3]{\Sin{#1,+}{#2} (#3)}
\newcommand{\posPin}[3]{\Pin{#1,+}{#2} (#3)}
\newcommand{\posDin}[3]{\Din{#1,+}{#2} (#3)}

\newcommand{\posnegSin}[4]{\posSin{#1}{#2}{#3, \lnot #4}}
\newcommand{\posnegPin}[4]{\posPin{#1}{#2}{#3, \lnot #4}}
\newcommand{\posnegDin}[4]{\posDin{#1}{#2}{#3, \lnot #4}}

\newcommand{\PA}{\mathsf{PA}}
\newcommand{\HA}{\mathsf{HA}}
\newcommand{\muPA}{\mu\PA}
\newcommand{\muHA}{\mu\HA}
\newcommand{\SOPA}{\PA2}
\newcommand{\SOHA}{\HA2}

\newcommand{\preaxiom}{\mathsf{Pre}}
\newcommand{\nprezero}{\preaxiom^0_\pNat}
\newcommand{\npresucc}{\preaxiom^\succ_\pNat}

\newcommand{\indaxiom}{\mathsf{Ind}}
\newcommand{\nindaxiom}{\indaxiom_\pNat}

\newcommand{\Rel}{\mathrm{Rel}}
\newcommand{\CA}{\mathsf{CA}}
\newcommand{\AC}{\mathsf{AC}}
\newcommand{\ca}[1]{#1\text{-}\CA}
\newcommand{\ac}[1]{#1\text{-}\AC}
\newcommand{\RCA}{\mathsf R \CA_0}
\newcommand{\PCA}{\ca{\Pin 11}_0}
\newcommand{\SCA}{\ca{\Sin 1 1 }_0}
\newcommand{\PSCA}{\ca{\Pin 1 2 }_0}
\newcommand{\SPCA}{\ca{\Sin 1 2 }_0}
\newcommand{\ACA}{\mathsf{ACA}_0}
\newcommand{\ATR}{\mathsf{ATR}}

\newcommand{\WF}{\mathrm{WF}}
\newcommand{\PO}{\mathrm{PO}}

\newcommand{\TO}{\mathrm{TO}}
\newcommand{\WO}{\mathrm{WO}}

\newcommand{\ordleq}{\preceq}
\newcommand{\ordlneq}{\prec}
\newcommand{\ordgeq}{\succeq}

\newcommand{\ordeq}{\approx}

\newcommand{\upbnd}[1]{\lceil #1\rceil}

\newcommand{\oper}[3]{\Lambda #1 \lambda #2 \, #3}
\newcommand{\apprx}[4]{(\oper #1 #2 #3)^{#4}}

\newcommand{\hr}[1]{|#1|}
\newcommand{\HR}{\hr \cdot}
\newcommand{\nmod}{\mathfrak N}
\newcommand{\realises}{\mathbf{r}}
\newcommand{\type}[1]{\mathtt t (#1)}

\newcommand{\negtrans}[1]{{#1}^{\N}}
\newcommand{\unnegtrans}[1]{{#1}_{\N}}
\newcommand{\nega}[1]{\neg^{#1}}
\newcommand{\negn}{\neg}
\newcommand{\dercod}{{\der_{\mathsf{cod}}}}
\newcommand{\derdec}{{\der_{\mathsf{dec}}}}

\newcommand{\pNat}{\mathsf{N}}
\newcommand{\pnat}[1]{\pNat\hspace{.05em} #1 }

\newcommand{\nattrans}[1]{#1^{\pNat}}

\newcommand{\muHAneg}{\muHA^-}

\begin{document}
\title[Computational expressivity of (circular) proofs with fixed points]{On the computational expressivity of (circular) proofs with fixed points}

\author[G.~Curzi]{Gianluca Curzi\lmcsorcid{57216760760}}[a]
\address{University of Gothenburg and University of Birmingham}
\email{gianluca.curzi@gu.se}
\thanks{}

\author[A.~Das]{Anupam Das\lmcsorcid{0000-0002-0142-3676}}[b]
\address{University of Birmingham}
\email{a.das@bham.ac.uk}
\thanks{This work was supported by a UKRI Future Leaders Fellowship, `Structure vs Invariants in Proofs' (project reference MR/S035540/1), by the Wallenberg Academy Fellowship Prolongation 
project `Taming Jörmungandr: The Logical Foundations of Circularity' (project reference 251080003), and by the
VR starting grant ``Proofs with Cycles in Computation'' (project reference 251088801).
}

\maketitle

\begin{abstract}
We study the computational expressivity of proof systems with fixed point operators, within the `proofs-as-programs' paradigm. We start with a calculus $\muLJ$ (due to Clairambault) that extends intuitionistic propositional logic by least and greatest positive fixed points. Based in the sequent calculus, $\muLJ$ admits a standard extension to a `circular' calculus $\cmuLJ$.

Our main result is that, perhaps surprisingly, both $\muLJ$ and $\cmuLJ$ represent the same first-order functions: those provably recursive in $\PSCA$, a subsystem of second-order arithmetic beyond the `big five' of reverse mathematics and one of the strongest theories for which we have an ordinal analysis (due to Rathjen). This solves various questions in the literature on the computational strength of proof systems with fixed points.

For the lower bound we give a realisability interpretation from an extension of Peano Arithmetic by fixed points that has been shown to be arithmetically equivalent to $\PSCA$ (due to M\"ollerfeld). For the upper bound we construct a novel computability model to give a totality argument for circular proofs with fixed points. 
In fact we formalise this argument itself within $\PSCA$ in order to obtain the tight bounds we are after.
Along the way we develop some novel reverse mathematics for the Knaster-Tarski fixed point theorem.
\end{abstract}



\comment{Conventions/notations:
\begin{itemize}
\item not using todo-notes, instead using bespoke macros to easily distinguish different types of todos, and so they are not visually overbearing.
    \item Now using a macro for arrow-type (currently $\arrow$). \anupam{I sometimes forget and use $\to$, not a big deal anymore.} No dot after binders in grammars (it is a meta-level notation). $\lambda$ etc.\ binds as strongly as possible (unless there is a dot).
    \item application always without brackets, i.e.\ $st$ not $s(t)$. This is consistent with above convention.
    \item $\closure A$ is the closure of a set $A$ of (co)terms under untyped application. Just saying `typed' for typed version at the moment.
    \item Notation convention: everything to do with $\muLJ$ and the iteration rules have standard notations (e.g.\ $\muLJ$, $\lr \mu$, $\iter{}$); everything to do with $\cmuLJ$ and left-unfolding rules currently have `primed' notations (e.g.\ $\muLJnorec$, $\mulunf$, $\iter{}'$). \anupam{I am not necessarily sold on this notation, I welcome alternative suggestions. If we keep might be better to call $\cmuLJ$ something like $\muLJnorec_{\mathrm{reg}}$}
    \item do not say `number-theoretic function', better to say `function on natural numbers' or even `numerical function'.
    \item use $\Sigma, \Gamma$ for lists/cedents of types to avoid too many arrows.
    \item use $\seqar$ for sequent arrow, not $\vdash$ since that is reserved for provability (mathematical judgement, not syntactic delimiter).
    \item Using mathrm for meta-level sets (e.g. $\var, \Var, \Types$).
    \item using mathsf for term constants, proof rules etc.
    \item Do not use cref: does not work with this document class
\end{itemize}
}

\section{Introduction}
\label{sec:intro}
\emph{Fixed points} abound in mathematics and computer science.
In logic we may enrich languages by `positive' fixed points to perform (co)inductive reasoning, while in
programming languages positive fixed points in type systems are used to represent (co)datatypes and carry out (co)recursion.
In both settings the underlying systems may be construed as fragments of their second-order counterparts.

In this work we investigate the computational expressivity of type systems with least and greatest (positive) fixed points.
We pay particular attention to \emph{circular proof systems}, where typing derivations are possibly non-well-founded (but regular), equipped with an $\omega$-regular `correctness criterion' at the level of infinite branches.
Such systems have their origins in modal fixed point logics, notably the seminal work of Niwi\'nski and Walukiewicz \cite{NiwWal96:games-for-mu-calc}.
Viewed as type systems under the `Curry-Howard' correspondence, circular proofs have received significant attention in recent years, notably based in systems of \emph{linear logic} \cite{BaeldeDS16Infinitaryprooftheory,EhrJaf21:cat-models-muLL,EhrJafSau21:totality-CmuMALL,BDKS22:bouncing-threads,DeSau19:infinets1,DePelSau:infinets2,DeJafSau22:phase-semantics} after foundational work on related finitary systems in \cite{BaeldeMiller07,Baelde12}.
In these settings circular proofs are known to be (at least) as expressive as their finitary counterparts, but classifying the exact expressivity of both systems has remained an open problem.
This motivates the main question of the present work:

\smallskip

\begin{question}
\label{question:main}
What functions do (circular) proof systems with fixed points represent?
\end{question}

\medskip

Circular type systems with fixed points were arguably pre-empted by foundational work of Clairambault \cite{Clairambault09}, who introduced an extension $\muLJ$ of Gentzen's sequent calculus $\LJ$ for intuitionistic propositional logic by least and greatest positive fixed points.
$\muLJ$ forms the starting point of our work and, using standard methods, admits an extension into a circular calculus, here called $\cmuLJ$, whose computational content we also investigate.


In parallel lines of research, fixed points have historically received considerable attention within mathematical logic.
The ordinal analysis of extensions of Peano Arithmetic ($\PA$) by inductive definitions has played a crucial role in giving proof theoretic treatments to (impredicative) second-order theories (see, e.g., \cite{sep-proof-theory}).
More recently, inspired by Lubarsky's work on `$\mu$-definable sets' \cite{Lubarsky93}, M\"ollerfeld has notably classified the proof theoretic strength of extensions of $\PA$ by general inductive definitions in \cite{Moellerfeld02:phd-thesis}.

In this work we somewhat bridge these two traditions, in computational logic and in mathematical logic, in order to answer our main question.
In particular we apply proof theoretic and metamathematical techniques to show that both $\muLJ$ and $\cmuLJ$ represent precisely the functions provably recursive in the subsystem $\PSCA$ of second-order arithmetic.
This theory is far beyond the `big five' of reverse mathematics, and is among the strongest theories for which we have an effective ordinal analysis (see \cite{Rathjen95:recent-advances}). 
The best known lower bound for $\muLJ$ before was G\"odel's $\T$ (see, e.g., \cite{Clairambault09}), which has the same proof theoretic strength as $\PA$. 
The best known upper bound was Girard-Reynold's $\F$, thanks to its impredicative encodings of fixed points, which has the same proof theoretic strength as second-order arithmetic $\SOPA$.


\subsection{Outline and contribution}
The structure of our overall argument is visualised in Figure~\ref{fig:grand-tour}, outlining a cycle of inclusions of `representable functions'.
Here the upper row consists of theories of arithmetic, where the representable functions of an arithmetic theory $T$ are just its provably total recursive functions; i.e.\ those functions $ f : \Nat \to \cdots \to \Nat$ with graph computed by some $\Sigma^0_1$ formula $\phi_f(\vec x, y)$ such that $T \proves \forall \vec x \exists y \phi_f (\vec x,y) $.
The lower row consists of type systems whose representable functions are just those admitting a typing derivation with conclusion $\nat \to \cdots \to \nat$ computing the function under its operational semantics (as in, e.g., Definition~\ref{defn:numerals-representability}). 

(1) is a standard embedding of finitary proofs into circular proofs (Proposition~\ref{prop:cmulj-simulates-mulj}).
(2) reduces $\cmuLJ$ to its `negative fragment', in particular free of greatest fixed points ($\nu$), via a double negation translation (Proposition~\ref{prop:cmulj-into-cmuljminus}).

(3) is one of our main contributions: we build a higher-order computability model $\HR$ that interprets $\cmuLJneg$ (Theorem~\ref{thm:prog-implies-hr}), and moreover formalise this construction itself within $\PSCA$ to obtain our upper bound (Theorem~\ref{thm:cmuLJneg-to-psca}).
The domain of this model a priori is an (untyped) term extension of $\cmuLJneg$.
It is important for logical complexity that we interpret fixed points semantically as bona fide fixed points, rather than via encoding into a second-order system.
Along the way we must also establish some novel reverse mathematics of the Knaster-Tarski fixed point theorem (Theorem~\ref{thm:muphi-is-phiwo}).

(4) is an intricate and nontrivial result established by M\"ollerfeld in \cite{Moellerfeld02:phd-thesis}, which we use as a `black box'.
(5) is again a double negation translation, morally a specialisation of the $\Pin 0 2 $-conservativity of full second-order arithmetic $\SOPA$ over its intuitionistic counterpart $\SOHA$, composed with a relativisation of quantifiers to $\Nat$ (\Cref{thm:muPA-pi02-cons-muHA,prop:muHA-to-muHAneg}). 

(6) is our second main contribution: we provide a realisability interpretation from $\muHAneg$ into $\muLJ$ (Theorem~\ref{thm:realisability-soundness}), morally by considerable specialisation of the analogous interpretation from $\SOHA$ into Girard-Reynolds' system $\F$.
Our domain of realisers is a (typed) term extension of $\muLJneg$ (the negative fragment of $\muLJ$), which is itself interpretable within $\muLJ$ (Proposition~\ref{prop:cmulj-into-cmuljminus}).


\subsection{Related work}
Fixed points have been studied extensively in type systems for programming languages.
In particular foundational work by Mendler in the late '80s \cite{Mendler87:recursive-types,Mendler91:inductive-types} already cast inductive type systems as fragments of second-order ones such as Girard-Reynolds' $\F$ \cite{girard1972systemF,Reynolds74systemF}.
Aside from works we have already mentioned, (a variant of) (5) has already been obtained by Tupailo in \cite{Tupailo04doubleneg}.
Berger and Tsuiki have also obtained a similar result to (6) in a related setting \cite{BerTsu21:ifp}, for \emph{strictly} positive fixed points, where bound variables may never occur under the left of an arrow. 
Their interpretation of fixed points is more akin to that in our type structure $\HR$ than our realisability model.

Finally the structure of our argument, cf.~Figure~\ref{fig:grand-tour}, is inspired by recent works in cyclic proof theory, notably \cite{Simpson17:cyc-arith,Das20:ca-log-comp} for (cyclic) (fragments of) $\PA$ and \cite{Das21:CT-preprint,Das21:CT-fscd,KuperbergPP21systemT} for (circular) (fragments of) G\"odel's system $\T$.

\begin{figure}
    \begin{tikzpicture}
\node (muLJ) {$\muLJ^{(-)}$};
\node [right = of muLJ] (cmuLJ) {$\cmuLJ^{\phantom{-}}$};
\node [right = of cmuLJ] (cmuLJneg) {$\cmuLJneg$};
\node [above = of muLJ] (muHA) {$\muHAneg$};
\node [above = of cmuLJneg] (PSCA) {$\PSCA$};
\node [above = of cmuLJ] (muPA) {$\muPA$};
\draw[->] (muLJ) to node[below]{\small (1)} (cmuLJ);
\draw[->] (cmuLJ) to node[below]{\small (2)} (cmuLJneg);
\draw[->] (cmuLJneg) to node[right]{\small (3)} (PSCA);
\draw[->] (PSCA) to node[above]{\small (4)} (muPA);
\draw[->] (muPA) to node[above]{\small (5)} (muHA);
\draw[->] (muHA) to node[left]{\small (6)} (muLJ);
\node [left = .5cm of muLJ, align=center] (models) {\footnotesize Type \\ \footnotesize systems:};
\node [left = .5cm of muHA, align=center] (theories) {\footnotesize Arithmetic \\ \footnotesize theories:};
    \end{tikzpicture}
    \caption{Summary of the main `grand tour' of this work. All arrows indicate inclusions of representable functions.
    }
    \label{fig:grand-tour}
\end{figure}

\subsection{Comparison to preliminary version}
This paper is an expansion of the preliminary conference version \cite{CD23:muLJ-lics}.
In this version we additionally include full proofs of all our results, as well as further examples and narrative. 

We have reformulated our realisability argument in \Cref{sec:realisability} into a form of \emph{abstract} realisability, inspired by the approach of \cite{BerTsu21:ifp}. 
This factors the approach of the preliminary version by a more careful relativisation of quantifiers to deal with an inconvenient type mismatch when realising inductive predicates.

Finally, in the preliminary version we also showed equivalence of $\muLJ$ and $\cmuLJ$ with their counterparts in linear logic from (see, e.g.,~\cite{Baelde12,BaeldeDS16Infinitaryprooftheory}) via appropriate proof interpretations.
These results will be expanded upon in a separate self-contained paper.

\subsection{Notation}
%
Throughout this work we employ standard rewriting theoretic notation. 
Namely for a relation $\rightsquigarrow_a$, we denote by $\rightsquigarrow^*_a$ the reflexive and transitive closure of $\rightsquigarrow_a$, and by $=_{a}$ the relexive symmetric transitive closure of $\rightsquigarrow_a$.

We shall make use of \emph{(first-order) variables}, written $x,y$ etc., and \emph{(second-order) variables}, written $X,Y$ etc.\ throughout.
We shall use these both in the setting of type systems and arithmetic theories, as a convenient abuse of notation.

\section{Simple types with fixed points: system $\muLJ$}
\label{sec:muLJ-prelims}
In this section we recall the system $\muLJ$ from \cite{Clairambault09,Clairambault13interleaving}.
More precisely, we present the `strong' version of $\muLJ$ from \cite{Clairambault13interleaving}.

\subsection{The sequent calculus $\muLJ$}

\emph{Pretypes}, written $\sigma,\tau$ etc., are generated by the following grammar:
\[
\sigma,\tau \ \bnf\  X \ \vert \ \unit \ \vert \ \sigma + \tau \ \vert \ \sigma \times \tau \ \vert \ \sigma \arrow  \tau  \ \vert \ \mu X \sigma \ \vert \ \nu X \sigma      
\]

\emph{Free (second-order) variables} of a pretype are defined as expected, construing $\mu$ and $\nu$ as binders:
\begin{itemize}
    \item $\fv X \dfn \{X\}$
    \item $\fv \unit \dfn \emptyset$
    \item $\fv {\sigma \star \tau} \dfn \fv \sigma \cup \fv \tau$, for $\star \in \{+,\times,\arrow\}$
    \item $\fv{\kappa X \, \sigma} \dfn \fv \sigma \setminus \{X\}$, for $\kappa \in \{\mu,\nu\}$
\end{itemize}
A pretype is \emph{closed} if it has no free variables (otherwise it is \emph{open}).

Throughout this work we shall assume some standard conventions on variable binding, in particular that each occurrence of a binder $\mu$ and $\nu$ binds a variable distinct from all other binder occurrences in consideration.
This avoids having to deal with variable renaming explicitly.
We follow usual bracketing conventions, in particular writing, say, $\rho \arrow \sigma\arrow \tau$ for $(\rho \arrow (\sigma \arrow \tau))$.
Binders $\mu X $ and $\nu Y$ bind as strongly as possible but we may write, say, $\mu X, Y .\,  \sigma \to \tau$ for $\mu X \mu Y (\sigma \to \tau)$.

\begin{defn}
    [Types and polarity]
    \emph{Positive} and \emph{negative} variables in a pretype are defined as expected:
\begin{itemize}     
\item $X$ is positive in $X$.
        \item $1$ is positive and negative in $X$.
        \item if $\sigma,\tau$ are positive (negative) in $X$ then so is $\sigma \star \tau$, for $\star \in \{+,\times\}$.
        \item if $\sigma$ is negative (positive) in $X$ and $\tau $ is positive (resp., negative) in $X$, then $\sigma \to \tau$ is positive (resp., negative) in $X$.
        \item if $\sigma$ is positive (negative) in $X$ then so is $\kappa Y \sigma$ (resp.), for $\kappa \in \{\mu,\nu\}$, both when $Y=X$ and $Y \neq X$.
    \end{itemize}
A pretype is a \emph{type} (or even \emph{formula}) if, for any subexpression $\kappa X \sigma$, $\sigma$ is positive in $X$.
The notions of \emph{(type) context} and \emph{substitution} are defined as usual. 
\end{defn}

\begin{remark}
    [Positivity vs strict positivity]
    Many authors require variables bound by fixed point operators to appear in \emph{strictly} positive position, i.e.\ never under the left-scope of $\arrow$.
    Like Clairambault \cite{Clairambault09,Clairambault13interleaving} we do not impose this stronger requirement, requiring only positivity in the usual syntactic sense.
\end{remark}

\begin{defn}[System $\muLJ$]
A \emph{cedent}, written $\Sigma, \Gamma$ etc., is just a list of types. 
A \emph{sequent} is an expression $\Sigma \seqar \sigma$.
The symbol $\seqar$ is, formally, just a syntactic delimiter (but the arrow notation is suggestive).
The system $\muLJ$ is given by the rules of Figures~\ref{fig:sequent-calculus-mulj}, \ref{fig:unfolding-rules} and \ref{fig:co-iteration-rules} (colours may be ignored for now).
The notions of \emph{derivation} (or \emph{proof}) are defined as usual. 
We write $P: \Gamma \seqar \tau$ if $P$ is a derivation of the sequent $\Gamma \seqar \tau$. 
\end{defn}

\begin{remark}
    [General identity and substitutions]
    \label{rem:gen-identity+substitution}
    Note that $\muLJ$ is equipped with a general identity rule, not only for atomic types. 
    This has the apparently simple but useful consequence that typing derivations are closed under substitution of types for free variables, i.e.\ if $\der(X) : \Gamma(X) \seqar \sigma(X)$ in $\muLJ $ (with all occurrences of $X$ indicated), then also $\der(\tau) : \Gamma(\tau) \seqar \sigma(\tau)$ in $\muLJ$ for any type $\tau$.
    Later, this will allow us to derive inductively general functors for fixed points in $\muLJ$ rather than including them natively; this will in turn become important later for verifying our realisability model for $\muLJ$ \anupam{come back to this later, need to mention eta}.
\end{remark}

\begin{remark}
    [$\muLJ$ as a fragment of second-order logic]
    \label{rem:mulj-frag-F}
    We may regard $\muLJ$ properly as a fragment of Girard-Reynolds System $\F$~\cite{girard1972systemF,Reynolds74systemF}, an extension of simple types to a second-order setting. 
    In particular, (co)inductive types may be identified with second-order formulas by:
    \[
    \begin{array}{r@{\ = \ }l}
         \mu X \sigma & \forall X ((\sigma \arrow X) \arrow X) \\
         \nu X \sigma & \exists X (X \times (X \arrow \sigma))
    \end{array}
    \]
    The rules for fixed points in $\muLJ$ are essentially inherited from this encoding, modulo some constraints on proof search strategy.
%
    Later we shall use a different encoding of fixed point types into a second-order setting, namely in arithmetic, as bona fide fixed points, in order to better control logical complexity.\anupam{could say more}
\end{remark}

In proofs that follow, we shall frequently only consider the cases of \emph{least} fixed points ($\mu$-types) and not \emph{greatest} fixed points ($\nu$-types), 
appealing to `duality' for the latter.
The cases for $\nu$ should be deemed analogous.
As we shall soon see, in Subsection~\ref{subsec:reduction-to-neg-frag}, we can indeed reduce our consideration to $\nu$-free types, without loss of generality in terms of representable functions.

\begin{remark}
    [Why sequent calculus?]
    Using a sequent calculus as our underlying type system is by no means the only choice. 
    However, since we shall soon consider non-wellfounded and circular typing derivations, it is important to have access to a well behaved notion of \emph{formula ancestry}, in order to properly define the usual totality criterion that underlies them.
    This is why the sequent calculus is the formalism of choice in circular proof theory.
\end{remark}

\begin{remark}
    [Variations of the fixed point rules]
    \label{rem:iter-variations}
    It is common to consider context-free and `weak' specialisations of the fixed point rules, e.g.:
    \begin{equation}
    \small
        \label{eq:(co)iterators-as-rules}
            \vlinf{\iter{}}{}{\mu X \sigma(X)\seqar \tau}{\sigma(\tau)\seqar \tau}
    \quad
    \vlinf{\coiter}{}{\rho\seqar \nu X \sigma(X)}{\rho\seqar \sigma(\rho)}
    \end{equation}
    In the presence of cut the `(co)iterator' rules above are equivalent to those of $\muLJ$ (see \Cref{app:simple-types-with-fixed-points} for some further remarks).
    However since the computational model we presume is cut-reduction, as we shall soon see, it is not appropriate to take them as first-class citizens.
    When giving a semantics that interprets $\cut$ directly, e.g.\ as we do for the term calculi in Section~\ref{sec:term-calculi}, it is often simpler to work with the (co)iterators above.
    In the remainder of this work
    we shall freely use the versions above in proofs too.
\end{remark}

\begin{figure}
    \centering
    \[
    \small
    \begin{array}{c}
         \vlinf{\ax}{}{\sigma \seqar \sigma}{}
   \qquad  
   \vlinf{\ex}{}{\mygreen{\Gamma}, \red{\tau}, \blue{\sigma}, \orange{\Delta} \seqar \cyan{\gamma}}{\mygreen{\Gamma},\blue{\sigma}, \red{\tau}, \orange{\Delta} \seqar \cyan{\gamma}}
   \qquad
   \vliinf{\cut}{}{\mygreen{\Gamma}, \orange{\Delta} \seqar \cyan{\tau}}{\mygreen{\Gamma} \seqar \sigma}{\orange{\Delta}, \sigma \seqar \cyan{\tau}}
   \\ \\
   \vlinf{\wk}{}{\mygreen{\Gamma}, \sigma \seqar \cyan{\tau}}{\Gamma \seqar \cyan{\tau}}
   \qquad
   \vlinf{\contr}{}{\mygreen{\Gamma}, \red{\sigma} \seqar \cyan{\tau}}{\mygreen{\Gamma}, \red{\sigma}, \red{\sigma} \seqar \cyan{\tau}}
   \qquad 
   \vlinf{\rr \unit}{}{\seqar \unit}{}
   \qquad 
   \vlinf{\lr \unit}{}{\mygreen{\Gamma}, \unit \seqar \cyan{\sigma}}{\mygreen{\Gamma} \seqar \cyan{\sigma}}
   \\ \\
    \vlinf{\rr \arrow}{}{\mygreen{\Gamma} \seqar \red{\sigma \arrow \tau}}{\mygreen{\Gamma}, \red{\sigma} \seqar \red{\tau}}
   \qquad 
   \vliinf{\lr \arrow}{}{\mygreen{\Gamma},  \orange{\Delta}, \red{\sigma \arrow \tau} \seqar \cyan{\gamma}}{\mygreen{\Gamma} \seqar \red{\sigma}}{\orange{\Delta}, \red{\tau} \seqar \cyan{\gamma}}
   \\ \\
   \vliinf{\rr \times}{}{\mygreen{\Gamma} , \orange{\Delta} \seqar \red{\sigma \times \tau}}{\mygreen{\Gamma} \seqar \red{\sigma}}{\orange{\Delta} \seqar \red{\tau}}
   \qquad 
   \vlinf{\lr \times}{}{\mygreen{\Gamma}, \red{\sigma \times \tau} \seqar \cyan{\gamma}}{\mygreen{\Gamma},\red{\sigma},\red{\tau} \seqar \cyan{\gamma}}\\ \\
   \vlinf{\choice 0 {\rr +}}{}{\mygreen{\Gamma} \seqar \blue{\tau_0 + \tau_1}}{\mygreen{\Gamma} \seqar \blue{\tau_0}}\qquad
    \vlinf{\choice 1 {\rr +}}{}{\mygreen{\Gamma} \seqar \blue{\tau_0 + \tau_1}}{\mygreen{\Gamma} \seqar \blue{\tau_1}}
    \qquad 
    \vliinf{\lr +}{}{\mygreen{\Gamma}, \red{\sigma + \tau} \seqar \cyan{\gamma}}{\mygreen{\Gamma}, \red{\sigma} \seqar \cyan{\gamma}}{\mygreen{\Gamma}, \red{\tau} \seqar \cyan{\gamma}}
    \end{array}
    \]
    \caption{Sequent calculus rules for $\LJ$.}
    \label{fig:sequent-calculus-mulj}
\end{figure}

\anupam{27-02-23: I corrected the colouring here, it is still incorrect in the conference version.}

\begin{figure}
 \[
 \small
   \vlinf{\rr \mu}{}{\mygreen{\Gamma}\seqar  \red{\mu X \sigma(X)}}{\mygreen{\Gamma} \seqar \red{\sigma(\mu X \sigma(X))}}
   \qquad
   \vlinf{\lr \nu }{}{\mygreen{\Gamma},  \red{\nu X\sigma(X)} \seqar \cyan{\tau}}{\mygreen{\Gamma} , \red{\sigma(\nu X \sigma(X))}\seqar \cyan{\tau}}
   \]
    \caption{Some unfolding rules for $\mu$ and $\nu$.}
    \label{fig:unfolding-rules}
\end{figure}

\begin{figure}
    \[
    \small
    \vliinf{\lr \mu}{}{\Gamma,\Delta, \mu X \sigma(X)  \seqar \tau}{\Gamma, \sigma(\rho)\seqar \rho}{\Delta, \rho \seqar \tau}
  \qquad
   \vliinf{\rr \nu }{}{\Gamma , \Delta \seqar \nu X\sigma(X)}{\Gamma \seqar  \tau }{  \Delta, \tau \seqar \sigma(\tau)}
    \]
    \caption{`(Co)iteration' rules for $\mu$ and $\nu$.}
    \label{fig:co-iteration-rules}
\end{figure}

\begin{definition}
    [Functors]
    \label{def:functors-in-muLJ}
 Let $\sigma(X) $ and $\rho(X)$ be (possibly open) types that are positive (and negative, respectively) in $X$. 
    For a proof $P:\Gamma, \tau \seqar \tau'$ we define $\sigma(P): \Gamma, \sigma(\tau)\seqar \sigma(\tau')$ and $\rho(P): \Gamma, \rho(\tau')\seqar \rho(\tau)$ by simultaneous induction as follows:
    \begin{itemize}
    \item If $\sigma(X)=X$ then $\sigma(\der)$ is just $\der$. Notice that it is never the  case that $\rho=X$, as $X$ can only occur negatively in $\rho$.
    \item If $\sigma$ and  $\rho$ are $\unit$ or some $Y\neq X$ then $\sigma(\der)$ and ${\tau}(\der)$ are defined respectively as follows: 
    \[
    \small
    \vlderivation{
    \vliq{\wk}{}{\Gamma, \sigma \seqar \sigma}{\vlin{\ax}{}{\sigma \seqar \sigma}{\vlhy{}}}
    }
    \qquad 
    \vlderivation{
    \vliq{\wk}{}{\Gamma, \rho \seqar \rho}{\vlin{\ax}{}{\rho \seqar \rho}{\vlhy{}}}
    }
    \]
    \item If $\sigma=\sigma_1 \imp \sigma_2$ and $\rho= \rho \imp \rho$ then we define $\sigma(\der)$ and $\rho(\der)$ respectively as follows:
    \[
    \footnotesize
    \vlderivation{
    \vlin{\rr \imp}{}{\Gamma, \sigma(\tau)\seqar \sigma(\tau')  }
     {
     \vliq{\contr}{}{\Gamma, \sigma(\tau) , \sigma_1(\tau')\seqar   \sigma_2(\tau')}{
     \vliin{\lr \imp}{}{\Gamma ,\Gamma, \sigma(\tau) , \sigma_1(\tau') \seqar   \sigma_2(\tau')}
       {
      \vldr{\sigma_1(\der)}{\Gamma, \sigma_1(\tau') \seqar \sigma_1(\tau)}
       }
       {
        \vldr{\sigma_2(\der)}{\Gamma, \sigma_2(\tau) \seqar \sigma_2(\tau')}
       }
       }
     }
    } 
     \qquad 
   \vlderivation{
    \vlin{\rr \imp}{}{\Gamma, \rho(\tau')\seqar \rho(\tau)  }
     {
     \vliq{\contr}{}{\Gamma, \rho(\tau'), \rho_1(\tau)\seqar   \rho_2(\tau)}{
     \vliin{\lr \imp}{}{\Gamma, \Gamma, \rho(\tau') , \rho_1(\tau) \seqar   \rho_2(\tau)}
       {
       \vldr{\rho_1(\der)}{\Gamma, \rho_1(\tau) \seqar \rho_1(\tau')}
       }
       {
        \vldr{\rho_2(\der)}{\Gamma, \rho_2(\tau') \seqar \rho_2(\tau)}
       }
       }
     }
    }
    \]
    \item If $\sigma=\sigma_1 \times \sigma_2$ and $\rho= \rho_1 \times \rho_2$  then we define $\sigma(\der)$ and $\rho(\der)$ respectively as follows:
     \[
     \footnotesize
   \vlderivation{
     \vliq{\contr}{}{\Gamma, \sigma(\tau)\seqar \sigma(\tau')  }
      {
      \vliin{\rr \times}{}{\Gamma, \Gamma, \sigma(\tau)\seqar\sigma(\tau)  }
       {
       \vlin{\lr \times }{}{\Gamma,\sigma(\tau)\seqar \sigma_1(\tau')}
         {
         \vldr{\sigma_1(\der)}{\Gamma, \sigma_1(\tau) \seqar \sigma_1(\tau')}
         }
       }
       {
         \vlin{\lr \times }{}{\Gamma,\sigma(\tau)\seqar \sigma_2(\tau')}
          {
             \vldr{\sigma_2(\der)}{\Gamma,  \sigma_2(\tau)\seqar \sigma_2(\tau')}
          }
       }
      }
    } 
 \qquad 
    \vlderivation{
     \vliq{\contr}{}{\Gamma, \rho(\tau')\seqar \rho(\tau)  }
      {
      \vliin{\rr \times}{}{\Gamma, \Gamma, \rho(\tau') \seqar \rho(\tau)   }
       {
       \vlin{\lr \times }{}{\Gamma, \rho(\tau') \seqar \rho_1(\tau)}
         {
         \vldr{\rho_1(\der)}{\Gamma, \rho_1(\tau') \seqar \rho_1(\tau)}
         }
       }
       {
         \vlin{\lr \times }{}{\Gamma, \rho(\tau') \seqar \rho_2(\tau)}
          {
             \vldr{\rho_2(\der)}{\Gamma,  \rho_2(\tau')\seqar \rho_2(\tau)}
          }
       }
      }
    } 
    \]
    \item If $\sigma=\sigma_1 + \sigma_2$ and $\rho= \rho_1 + \rho_2$ then we define $\sigma(\der)$ and $\rho(\der)$ respectively as follows:
    \[
    \footnotesize
   \vlderivation{
    \vliin{\lr +}{}{\Gamma, \sigma(\tau)\seqar  \sigma(\tau') }
     {
     \vlin{\rr +}{}{\Gamma, \sigma_1(\tau) \seqar  \sigma(\tau')}
       {
       \vldr{\sigma_1(\der)}{\Gamma, \sigma_1(\tau) \seqar \sigma_1(\tau')}
       }
     }
     {
      \vlin{\rr +}{}{\Gamma, \sigma_2(\tau) \seqar  \sigma(\tau')}
       {
        \vldr{\sigma_2(\der)}{\Gamma, \sigma_2(\tau) \seqar \sigma_2(\tau')}
       }
     }
    }
    \qquad 
     \vlderivation{
    \vliin{\lr +}{}{\Gamma, \rho(\tau') \seqar \rho(\tau)  }
     {
     \vlin{\rr +}{}{\Gamma, \rho_1(\tau') \seqar \rho(\tau) }
       {
       \vldr{\rho_1(\der)}{\Gamma, \rho_1(\tau') \seqar \rho_1(\tau)}
       }
     }
     {
      \vlin{\rr +}{}{\Gamma, \rho_2(\tau') \seqar \rho(\tau) }
       {
        \vldr{\rho_2(\der)}{\Gamma, \rho_2(\tau') \seqar \rho_2(\tau)}
       }
     }
    }
    \]
\item if $\sigma(X) = \mu Y \sigma'(X,Y)$ and $\rho(X) = \mu Y \rho'(X,Y)$  then we define $\sigma(\der)$ and $\rho(\der)$ respectively as follows:
        \[
        \small
        \vlderivation{
        \vlin{\lr \mu}{}{\Gamma, \sigma(\tau) \seqar \sigma(\tau')}{
        \vlin{\rr \mu}{}{\Gamma, \sigma'(\tau, \sigma(\tau'))\seqar \sigma(\tau')}{
        \vldr{\colorbox{white}{$\scriptstyle\sigma'(P,\sigma(\tau'))$}}{\Gamma, \sigma'(\tau, \sigma(\tau')) \seqar \sigma'(\tau', \sigma(\tau'))}
        }
        }
        }
        \qquad         \vlderivation{
        \vlin{\lr \mu}{}{\Gamma, \rho(\tau') \seqar \rho(\tau)}{
        \vlin{\rr \mu}{}{\Gamma, \rho'(\tau', \rho(\tau))\seqar \rho(\tau)}{
        \vldr{\colorbox{white}{$\scriptstyle\rho'(P,\rho(\tau))$}}{\Gamma, \rho'(\tau', \rho(\tau)) \seqar \rho'(\tau, \rho(\tau))}
        }
        }
        }
        \]
        where $\sigma'(P,\sigma(\tau'))$ (resp., $\rho'(P,\rho(\tau))$) are obtained from the IH for $\sigma'(P,Y)$ (resp., $\rho'(P,Y)$) under substitution of $\sigma(\tau')$ for $Y$ (resp., $\rho(\tau)$), cf.~Remark~\ref{rem:gen-identity+substitution}.
        \item   if $\sigma(X) = \nu Y \sigma'(X,Y)$ and $\rho(X) = \nu Y \rho'(X,Y)$  then we define $\sigma(\der)$ and $\rho(\der)$ respectively as follows: 
            \[
        \small
        \vlderivation{
        \vlin{\rr \nu}{}{\Gamma, \sigma(\tau) \seqar \sigma(\tau')}{
        \vlin{\lr \nu}{}{\Gamma, \sigma(\tau) \seqar \sigma'(\tau', \sigma(\tau))}{
        \vldr{\colorbox{white}{$\scriptstyle\sigma'(P,\sigma(\tau))$}}{\Gamma, \sigma'(\tau, \sigma(\tau)) \seqar \sigma'(\tau', \sigma(\tau))}
        }
        }
        }
        \qquad 
        \vlderivation{
        \vlin{\rr \nu}{}{\Gamma, \rho(\tau') \seqar \rho(\tau)}{
        \vlin{\lr \nu}{}{\Gamma, \rho(\tau')\seqar \rho'(\tau, \rho(\tau'))}{
        \vldr{\colorbox{white}{$\scriptstyle\rho'(P,\rho(\tau'))$}}{\Gamma, \rho'(\tau', \rho(\tau')) \seqar \rho'(\tau, \rho(\tau'))}
        }
        }
        }
        \]
        where $\sigma'(P,\sigma(\tau))$ (resp., $\rho'(P,\rho(\tau'))$) are obtained from the IH for $\sigma'(P,Y)$ (resp., $\rho'(P,Y)$) under substitution of $\sigma(\tau)$ for $Y$ (resp., $\rho(\tau')$), cf.~Remark~\ref{rem:gen-identity+substitution}.
        \end{itemize}
\end{definition}

\begin{example}
[Post-fixed point]
\label{ex:post-fixed-in-muLJ}
It is implicit in the rules of $\muLJ$ that $\mu X \sigma(X)$ may be seen as the \emph{least} fixed point of $\sigma(\cdot)$, under a suitable semantics (e.g.\ later in Section~\ref{sec:totality}).
The $\rr \mu $ rule indicates that it is a pre-fixed point, while the $\lr \mu$ rule indicates that it is least among them. 
To see that it is also a post-fixed point $\small\vlinf{\mulunf}{}{\Gamma, \mu X \sigma(X)\seqar \tau}{\Gamma, \sigma(\mu X \sigma (X))\seqar \tau}$ we may use a derivation that mimics standard textbook-style proofs of Knaster-Tarski:
\[
\small
\vlderivation{
\vliin{\lr \mu}{}{\Gamma, \mu X \sigma(X) \seqar \tau}{
    \vliq{\sigma}{}{ \sigma(\sigma(\mu X \sigma(X)))\seqar \sigma(\mu X \sigma(X))}{
    \vlin{\rr \mu}{}{ \sigma(\mu X\sigma(X))\seqar \mu X \sigma(X)}{
    \vlin{\id}{}{ \sigma(\mu X \sigma(X))\seqar \sigma(\mu X \sigma(X))}{\vlhy{}}
    }
    }
}{
    \vlhy{\Gamma, \sigma(\mu X \sigma(X))\seqar \tau}
}
}
\]
Dually, we can derive  $\vlinf{\nurunf}{}{\Gamma \seqar \nu X \sigma(X)}{\Gamma \seqar \sigma(\nu X \sigma(X))}$ as follows: 
\[
\small
\vlderivation{
\vliin{\rr \nu}{}{\Gamma\seqar  \nu X \sigma(X) }
{
    \vlhy{\Gamma \seqar  \sigma(\nu X \sigma(X))}
}
{
    \vliq{\sigma}{}{\sigma(\nu X \sigma(X)) \seqar \sigma(\sigma(\nu X \sigma(X))) }{
    \vlin{\lr \nu}{}{ \nu X \sigma(X)\seqar \sigma(\nu X\sigma(X)) }{
    \vlin{\id}{}{ \sigma(\nu X \sigma(X))\seqar \sigma(\nu X \sigma(X))}{\vlhy{}}
    }
    }
}
}
\]
\end{example}

\subsection{Computing with derivations}
The underlying computational model for sequent calculi, with respect to the `proofs-as-programs' paradigm, is \emph{cut-reduction}.
In our case this follows a standard set of cut-reduction rules for the calculus $\LJ$.
For the fixed points, cut-reduction is inherited directly from the encoding of fixed points in system $\F$ that induces our rules, cf.~Remark~\ref{rem:mulj-frag-F}.
Following Baelde and Miller~\cite{BaeldeMiller07}, we give self-contained cut-reductions here:
\gianluca{I have put my/Farzad's definition of cut-elimination for lrmu vs rr mu. I commented the notion of application. Maybe some of the discussion above is now obsolete?}
\anupam{16-10-23: it might be worth giving the application-version of cut-reduction here too, at least as an example: it exemplifies how cut-reduction in $\muLJ$ specialises that of second-order logic.}

 \gianluca{Also, most of the following discussion above can be now found in the appendix. Anupam, please check if you are fine with it: NB: contractions above occur because we have $\Delta$ context in the step function. This seems to break cut-elimination in linear logic, hence why Baelde and Miller have context-free step function. 
        Note that it is on this point that the format of the rules must (crucially) differ from that given by F, where the presence of context in the step function does not compromise linearity.
        Under an appropriate semantics, e.g. in term calculi section, or even embedding into circular system, it makes no difference, but we might want to change the rule presentation globally. Update: in fact perhaps this is not best at all, since without a context in the step function it is not clear how to cut-free derive the iterator! it would be nice for the iterator to be cut-free derivable, not least for embedding target of realisability. Under suitable semantics, e.g. term calculi, all choices are equivalent (using cuts).}

\begin{defn}
    [Cut-reduction for fixed points]
    \label{def:appl-and-cut-red-muLJ}
\emph{Cut-reduction} on $\muLJ$-derivations, written $\cutred$, is the smallest relation on derivations including all the usual cut-reductions of $\LJ$ and the reductions in~\Cref{fig:cut-reduction-mu-nu-muLJ-paper}. {As usual, we allow these reductions to be performed on sub-derivations. I.e.\ they are `context-closed'.} 
%
\end{defn}

\begin{figure}
    \centering
    \centering
    \[
    \footnotesize
    \vlderivation{
    \vliin{\cut}{}{\Sigma, \Gamma,\Delta\seqar \tau}
    {
     \vlin{\rr \mu}{}{\Sigma \seqar \mu X \sigma(X)}{
     \vlhy{\Sigma \seqar \sigma(\mu X \sigma(X))}}
     }
     {
  \vliin{\lr \mu}{}{\Gamma,\Delta, \mu X \sigma(X)  \seqar \tau}
  {\vlhy{\Gamma, \sigma(\rho)\seqar \rho}}{\vlhy{\Delta, \rho \seqar \tau}}    
  }
    }
    \qquad
    \quad
    \rightsquigarrow
\]
\vspace{0.2cm}
\[
\footnotesize
\vlderivation{
\vliq{\contr}{}{\Sigma, \Delta,  \Gamma\seqar \tau}
{
\vliin{\cut}{}{\Sigma,  \Delta, \Gamma, \Gamma \seqar \tau}
{
\vliin{\cut}{}{\Sigma, \Gamma, \Gamma \seqar \rho}
{
\vliin{\cut}{}{\Sigma, \Gamma\seqar \sigma(\rho)}
{
\vlhy{\Sigma \seqar \sigma(\mu X \sigma(X))}
}
{
\vliq{\sigma}{}{\Gamma, \sigma(\mu X \sigma(X) ) \seqar \sigma(\rho)}
{
 \vliin{\lr \mu}{}{\Gamma, \mu X \sigma(X)  \seqar \rho}
  {\vlhy{\Gamma, \sigma(\rho)\seqar \rho}}{\vlin{\id}{}{ \rho \seqar \rho}{\vlhy{}}}    
}
}
}
{\vlhy{\Gamma, \sigma(\rho)\seqar \rho}}
}
{
\vlhy{\Delta, \rho \seqar \tau}
}
}
}
\]
\vspace{1cm}
\[
 \footnotesize
\vlderivation{
\vliin{\cut}{}{\Delta, \Gamma , \Sigma \seqar \tau }
{
 \vliin{\rr \nu }{}{\Delta,  \Gamma \seqar \nu X\sigma(X)}{\vlhy{\Delta \seqar  \rho }}{\vlhy{  \Gamma, \rho \seqar \sigma(\rho)}}
 }
 {
 \vlin{\lr \nu }{}{\Sigma,  \nu X\sigma(X) \seqar \tau}{\vlhy{\Sigma , \sigma(\nu X \sigma(X))\seqar \tau}}
 }
 }
 \qquad
 \quad
 \rightsquigarrow
\]
\vspace{0.2cm}
\[
\footnotesize
\vlderivation{
\vliq{\contr}{}{\Delta , \Gamma, \Sigma \seqar \tau}
{
\vliin{\cut}{}{\Delta , \Gamma, \Gamma, \Sigma , \Delta \seqar \tau}
{\vlhy{\Delta \seqar \rho}}
{
\vliin{\cut}{}{\Gamma, \Gamma, \Sigma , \rho \seqar \tau}
{
\vlhy{  \Gamma, \rho \seqar \sigma(\rho)}
}
{
\vliin{\cut}{}{\Gamma ,  \Sigma, \sigma(\rho) \seqar \tau}
{
\vliq{\sigma }{}{\Gamma ,  \sigma(\rho) \seqar \sigma(\nu X\sigma(X))}
{
\vliin{\rr \nu }{}{\Gamma , \rho \seqar \nu X\sigma(X)}{\vlin{\id}{}{\rho \seqar  \rho}{\vlhy{} }}{\vlhy{  \Gamma, \rho \seqar \sigma(\rho)}}
}
}
{\vlhy{\Sigma , \sigma(\nu X \sigma(X))\seqar \tau}}
}
}
}
}
\]
    \caption{Cut-reduction rules for $\mu$ and $\nu$ in $\muLJ$.}
    \label{fig:cut-reduction-mu-nu-muLJ}
    \label{fig:cut-reduction-mu-nu-muLJ-paper}
\end{figure}


When speaking of (subsets of) $\muLJ$ as a computational model, we always mean with respect to the relation $\cutred^*$ unless otherwise stated.
More precisely:

\begin{defn}[Representability in $\muLJ$]\label{defn:numerals-representability}
We define the \emph{type of natural numbers} as $\N\dfn \mu X(1+ X)$. 
We also define the \emph{numeral} $\cod{n}: \N$  by induction on $n \in \Nat$: 
\[
 \small
\begin{array}{rclrcl}
 \cod 0    & \dfn&
     \vlderivation{
     \vlin{\rr \mu}{}{\seqar \N}{\vlin{ {\rr +}}{}{\seqar \unit + \N}{\uaxiom }}
     }
     &\qquad 
    \cod {n+1} & \dfn & 
    \vlderivation{
    \vlin{\rr \mu }{}{\seqar\N}{\vlin{{\rr +}}{}{\seqar \unit + \N}{
    \vliq{\numeral n}{}{\seqar \N}{\vlhy{}}
    }
    }
    }
\end{array}
\]    
We say that a {(possibly partial)} function $f: \Nat \times \overset{k}{\ldots} \times \Nat \to \Nat$ is \emph{representable} in $\muLJ$ \todonew{or any other system? maybe mention that muLJ is known to be total} if there is  a $\muLJ$-derivation $\der_f: \N , \overset{k}{\ldots} ,    \N \seqar \N  $ s.t., for any $n_1, \ldots, n_k \in \Nat$, the derivation,
\[\small
\vlderivation{
\vliin{\cut}{}{\seqar \N}{\vliq{\numeral n_k}{}{\seqar \N}{\vlhy{}}}{
\vliin{\cut}{}{\vdots}{ \vliq{\numeral n_2}{}{\seqar \N}{\vlhy{}} }{
\vliin{\cut}{}{\N,\overset{k-1}{\ldots},  \N\seqar \N}{ \vliq{\numeral n_k}{}{\seqar \N}{\vlhy{}} }{\vltr{\der_f}{\N,\overset{k}{\ldots},  \N\seqar \N}{\vlhy{\ }}{\vlhy{\ }}{\vlhy{\ }} }
}
}
}
\]
reduces under $\cutred^*$ to the numeral $\numeral{f(n_1, \dots, n_k)}$, whenever it is defined (otherwise it reduces to no numeral).
In this case we say that $\der_f$ \emph{represents} $f$ in $\muLJ$.
\end{defn}

\begin{figure}
    \[
    \vlinf{\Nzero}{}{\seqar \N}{}
    \qquad 
    \vlinf{\Nsucc}{}{\Gamma \seqar \N}{\Gamma \seqar \N}
    \qquad
    \vliiinf{\lr \N}{}{\Gamma,\Delta, \N \seqar \tau}{\Gamma \seqar \sigma}{\Gamma, \sigma \seqar \sigma}{\Delta, \sigma \seqar \tau}
    \]
    \caption{Native rules for $\N$ in $\muLJ$.}
    \label{fig:native-n-rules-mulj}
\end{figure}
\begin{example}
    [Native rules for natural number computation]
    `Native' rules for type $\N$ in $\muLJ$ are given in Figure~\ref{fig:native-n-rules-mulj}, all routinely derivable in $\muLJ$, as in Figure~\ref{fig:deriving-nat-rules-in-muLJ}.
    The corresponding `native' cut-reductions, derivable using $\cutred$ cf.~\Cref{fig:cut-reduction-mu-nu-muLJ}, are also routine.
    We shall examine this further in \Cref{sec:term-calculi}.
    \todonew{give cut reductions?}
     \gianluca{Say that recursor is iteration: we can get proper recursion using pairing.}
    Note that, from here we can recover the usual recursor of system $\T$, as shown formally by Clairambault~\cite{Clairambault13interleaving}.
\end{example}


\begin{example}\label{exmp:examples-of-derivations}
The least and greatest fixed point operators $\mu$ and $\nu$ allow us to encode inductive data (natural numbers, lists, etc) and coinductive data (streams, infinite trees, etc). We have already seen the encoding of natural numbers. The type of lists and streams (both over natural numbers) can be represented by, respectively,  $L\dfn \mu X (\unit + (\N \times X))$ and $S= \nu X (\N \times X)$.
Figure~\ref{fig:examples-derivations}, left-to-right, shows the encoding of  $\varepsilon$ and $n::l$ (i.e., the empty list and the operation appending  a natural number to a list). 
Figure~\ref{fig:list+stream-to-stream-muLJ} shows the encoding of a concatenation of a list and a stream into a stream (by recursion over the list with the invariant $S \to S$).\footnote{Both examples were originally given for $\muMALL$ in \cite{Doumane17thesis}.}
\end{example}

\begin{figure}
    \centering
     \[
\Nzero\ \dfn \ 
\vlderivation{
\vlin{\rr \mu}{}{\seqar \nat}{
\vlin{\rr +^0}{}{\seqar 1 + \nat}{
\vlin{\rr 1 }{}{\seqar 1}{\vlhy{}}
}
}
}
\qquad
 \Nsucc\ \dfn \ 
\vlderivation{
\vlin{\rr \mu}{}{\Gamma \seqar N}{
\vlin{\rr {+^1}}{}{\Gamma\seqar \unit + \N}
{
\vlhy{\Gamma \seqar \N}
}
}
}
\qquad 
\lr N \ \dfn \ \vlderivation{
\vliin{\lr\mu}{}{\Gamma, \Delta, \N \seqar \tau}
{
\vliin{\lr +}{}{\Gamma, 1 + \sigma \seqar \sigma }
{
\vlin{\lr \unit}{}{\Gamma, \unit \seqar \sigma }{\vlhy{\Gamma \seqar \sigma}}
}
{
\vlhy{\Gamma, \sigma \seqar \sigma}
}
}
{
\vlhy{\Delta, \sigma \seqar  \tau}
}
}
    \]
    \caption{Constructing and destructing natural numbers in $\muLJ$.}
    \label{fig:deriving-nat-rules-in-muLJ}
\end{figure}

\begin{figure}
\[
\small
\vlderivation{
\vlin{\rr \mu}{}{\seqar L}
{
\vlin{\rr {+^0}}{}{\seqar\unit + (\N \times L)}
{
\vlin{\rr \unit}{}{\seqar\unit}{\vlhy{}}
}
}
}
\qquad 
\vlderivation{
\vlin{\rr \mu}{}{\seqar L}
{
\vlin{\rr{+^1}}{}{\seqar \unit +(\N \times L) }
{
\vliin{\times}{}{\seqar \N \times L}
{
\vltr{\cod{n}}{\seqar \N}{\vlhy{}}{\vlhy{}}{\vlhy{}}
}
{
\vltr{\underline{l}}{\seqar L}{\vlhy{}}{\vlhy{}}{\vlhy{}}
}
}
}
}
\]
        \caption{Constructing lists in $\muLJ$.}
    \label{fig:examples-derivations}
\end{figure}

\begin{figure}
    \centering
\[
\small
\vlderivation{
\vlin{\lr \nu}{}{L \seqar S \to S}
{
\vliin{\lr + }{}{\unit +(N \times (S \to S))  \seqar S \to S}
{
\vlin{\lr \unit}{}{\unit \seqar S \to S}
{
\vlin{\rr \to }{}{ \seqar S \to  S}
{
\vlin{\id }{}{S \seqar S}{\vlhy{}}
}
}
}
{
\vlin{\lr \times}{}{N \times (S \to S) \seqar  S \to S}
{
\vlin{\rr \to}{}{N, S \to S \seqar  S \to S}
{
\vlin{\rr \nu}{}{N, S \to S, S \seqar S}
{
\vliin{\rr \times}{}{N, S \to S, S \seqar N \times S}
{
\vliin{\lr \to}{}{S \to S, S \seqar S}
{
\vlin{\id}{}{S \seqar S}{\vlhy{}}
}
{
\vlin{\id}{}{S \seqar S}{\vlhy{}}
}
}
{
\vlin{\id}{}{N \seqar N}{\vlhy{}}
}
}
}
}
}
}
}
\]    \caption{Concatenation of a list and a stream in $\muLJ$.}    \label{fig:list+stream-to-stream-muLJ}
\end{figure}

 \section{A circular version of $\muLJ$}
\label{sec:cmuLJ-prelims}

In this section we shall develop a variation of $\muLJ$ that does not have rules for (co)iteration, but rather devolves such work to the proof structure.
First, let us set up the basic system of rules we will work with:
\begin{definition}
[$\muLJ$ `without (co)iteration']
Write $\muLJnorec$ for the system of all rules in Figures~\ref{fig:sequent-calculus-mulj}, \ref{fig:unfolding-rules} and \ref{fig:unfolding-further-rules} (but not \ref{fig:co-iteration-rules}).
\end{definition}


\subsection{`Non-wellfounded' proofs over $\muLJnorec$}

`Coderivations' are generated \emph{coinductively} by the rules of a system, dually to derivations that are generated inductively.
I.e.\ they are possibly infinite proof trees generated by the rules of a system.

\begin{defn}
[Coderivations]
A ($\muLJnorec$-)\emph{coderivation} $\der$ is a possibly infinite {rooted} tree (of height $\leq \omega$) generated by the rules of $\muLJnorec$.
Formally, we identify $\der$ with a prefix-closed subset of $ \{0,1\}^*$ (i.e.\ a binary tree) where each node  is labelled by an inference step from $\muLJnorec$ such that, whenever $\alpha\in \{0,1\}^*$ is labelled by a step $\vliiinf{}{}{S}{S_1}{\cdots}{S_n}$, for $n\leq 2$, $\alpha$ has $n$ children in $\der$ labelled by steps with conclusions $S_1, \dots, S_n$ respectively. 

We say that a coderivation is \emph{regular} (or \emph{circular}) if it has only finitely many distinct sub-coderivations.
\end{defn}

A regular coderivation can be represented as a finite labelled graph (possibly with cycles) in the natural way.

\begin{figure}
     \centering
 \[
 \small
   \vlinf{\mulunf}{}{\mygreen{\Gamma}, \red{\mu X. \sigma} \seqar \cyan{\tau}}{\mygreen{\Gamma}, \red{\sigma(\mu X.\sigma)} \seqar \cyan{\tau}}
   \qquad 
   \vlinf{\nurunf}{}{\mygreen{\Gamma} \seqar \red{\nu X.\sigma}}{\mygreen{\Gamma} \seqar \red{\sigma(\nu X.\sigma)}}
   \]
    \caption{Further unfolding rules for $\mu$ and $\nu$.}
    \label{fig:unfolding-further-rules}
\end{figure}

\subsection{Computing with coderivations}

Just like for usual derivations, the underlying notion of computation for coderivations is cut-reduction, and the notion of representability remains the same.
However we must also adapt the theory of cut-reduction to the different fixed point rules of $\muLJnorec$.

\begin{definition}
    [Cut-reduction on coderivations]\label{defn:rrmu-vs-llmu-cmull}
    $\cutrednorec$ is the smallest relation on $\muLJnorec$-coderivations including all the usual cut-reductions of $\LJ$ and the cut-reductions in Figure~\ref{fig:cut-reduction-cmulj-least-greatest},\footnote{Again, we allow these reductions to be applied on sub-coderivations.}.
%
When speaking of (subsets of) coderivations as computational models, we typically mean with respect to $\cutrednorec^*$.
\end{definition}

\begin{figure}
    \centering
\[
    \footnotesize
           \vlderivation{
\vliin{\cut}{}{\Gamma, \Delta \seqar \tau}{
    \vlin{\rr \mu}{}{\Gamma \seqar \mu X \sigma(X) }{\vlhy{\Gamma \seqar \sigma (\mu X \sigma(X))}}
}{
    \vlin{\mulunf }{}{\Delta , \mu X \sigma (X)\seqar \tau}{\vlhy{\Delta, \sigma(\mu X \sigma(X)) \seqar \tau}}
}
}
  \quad \rewrite \quad      
  \vliinf{\cut}{}{\Gamma, \Delta \seqar \tau}{\Gamma \seqar \sigma(\mu X \sigma(X))}{\Delta, \sigma(\mu X \sigma(X)) \seqar \tau}
  \]
  \[
  \footnotesize
   \vlderivation{
\vliin{\cut}{}{\Gamma, \Delta \seqar \tau}{
    \vlin{\nurunf }{}{\Gamma \seqar \nu X \sigma(X) }{\vlhy{\Gamma \seqar \sigma (\nu X \sigma(X))}}
}{
    \vlin{\lr \nu }{}{\Delta , \nu X \sigma (X)\seqar \tau}{\vlhy{\Delta, \sigma(\nu X \sigma(X))}}
}
}
 \quad \rewrite \quad 
 \vliinf{\cut}{}{\Gamma, \Delta \seqar \tau}{\Gamma \seqar \sigma(\nu X \sigma(X))}{\Delta, \sigma(\nu X \sigma(X)) \seqar \tau}
\]    \caption{Cut-reduction for least and greatest fixed points in $\cmuLJ$.}
    \label{fig:cut-reduction-cmulj-least-greatest}
\end{figure}


\begin{example}
    [Decomposing the (co)iterators]
    \label{ex:(co)iter-to-cycles}
    The `(co)iterator' rules of Figure~\ref{fig:co-iteration-rules} can be expressed by regular coderivations using only the unfolding rules for fixed points as follows:
    \begin{equation}
    \label{eq:mul-as-circ-coder}
    \small
       \vlderivation{
    \vliin{\cut}{}{\Gamma, \Delta, \mu X \sigma(X) \seqar \tau}{
    \vlin{\mulunf}{\bullet}{\Gamma, \mu X \sigma(X) \seqar \rho}{
    \vlin{\contr}{}{\Gamma,  \sigma(\mu X \sigma(X))\seqar \rho}
    {
    \vliin{\cut}{}{\Gamma,\Gamma,  \sigma(\mu X \sigma(X))\seqar \rho}{
    \vliq{\sigma}{}{\Gamma, \sigma(\mu X \sigma(X)) \seqar \sigma(\rho)}{
    \vlin{\mulunf}{\bullet}{\Gamma, \mu X \sigma(X) \seqar \rho}{\vlhy{\vdots}}
    }
    }{
    \vlhy{\Gamma, \sigma(\rho)\seqar \rho}
    }
    }
    }
    }{
    \vlhy{\Delta, \rho \seqar \tau}
    }
    }
    \end{equation}
    Here we mark with $\bullet$ roots of identical coderivations, a convention that we shall continue to use throughout this work.
    Dually for the coiterator:
    \[
\small
       \vlderivation{
    \vliin{\cut}{}{\Gamma, \Delta \seqar \nu X \sigma(X)}
    {
    \vlhy{\Delta \seqar \tau}
    }
    {
    \vlin{\nurunf}{\bullet}{\Gamma, \tau \seqar \nu X \sigma(X)}{
    \vliq{\contr}{}{\Gamma,  \tau\seqar \sigma(\nu X \sigma(X))}
    {
    \vliin{\cut}{}{\Gamma,\Gamma,  \tau\seqar \sigma(\nu X \sigma(X)}
    {
    \vlhy{\Gamma, \tau\seqar \sigma(\tau)}
    }
    {
    \vliq{\sigma}{}{\Gamma, \sigma(\tau) \seqar \sigma(\nu X \sigma(X))}{
    \vlin{\nurunf}{\bullet}{\Gamma, \tau \seqar \nu X \sigma(X)}{\vlhy{\vdots}}
    }
    }
    }
    }
    }
    }
\]

    Moreover, one can verify that this embedding gives rise to a bona fide simulation of $\cutred^*$ by $\cutrednorec^*$. 
    We do not cover the details at this point, but make a stronger statement later in Proposition~\ref{prop:cmulj-simulates-mulj}.
\end{example}

\begin{example}
    [Functors and $\eta$-expansion of identity]
    Thanks to the decomposition of (co)iterators above, we can derive `functors' in $\cmuLJ$, cf.~Definition~\ref{def:functors-in-muLJ}.
    This gives rise to an `$\eta$-expansion' of identity steps, reducing them to atomic form. 
The critical fixed point cases are:
    \[
    \small
    \vlderivation{
    \vlin{\mulunf}{\bullet}{\mu X \sigma(X)\seqar \mu X \sigma(X)}{
    \vlin{\rr \mu}{}{\sigma(\mu X(\sigma(X))\seqar \mu X \sigma(X)}{
    \vliq{\sigma}{}{\sigma(\mu X \sigma(X))\seqar \sigma(\mu X \sigma(X))}{
    \vlin{\lr \mu}{\bullet}{\mu X \sigma(X)\seqar \mu X \sigma(X)}{\vlhy{\vdots}}
    }
    }
    }
    }
    \qquad
    \small
    \vlderivation{
    \vlin{\lr \nu}{\bullet}{\nu X \sigma(X)\seqar \nu X \sigma(X)}{
    \vlin{\rr {\nu'}}{}{\sigma(\nu X(\sigma(X))\seqar \nu X \sigma(X)}{
    \vliq{\sigma}{}{\sigma(\nu X \sigma(X))\seqar \sigma(\nu X \sigma(X))}{
    \vlin{\lr \nu}{\bullet}{\nu X \sigma(X)\seqar \nu X \sigma(X)}{\vlhy{\vdots}}
    }
    }
    }
    }
    \]
Notice that the functors $\sigma$ indicated above will depend on smaller identities, cf.~\Cref{def:functors-in-muLJ}, calling the inductive hypothesis.
    Note that the coderivations above are `logic-independent', and indeed this reduction is common in other circular systems for fixed point logics, such as the modal $\mu$-calculus and $\muMALL$ (see, e.g.,~\cite{BaeldeDS16Infinitaryprooftheory}).
\end{example}

\subsection{A totality criterion}
We shall adapt to our setting a well-known `termination criterion' from non-wellfounded proof theory. First, let us recall some standard  proof theoretic concepts about (co)derivations, similar to those in \cite{BaeldeDS16Infinitaryprooftheory,KuperbergPP21systemT,Das21:CT-preprint,Das21:CT-fscd}.

\begin{defn}
[Ancestry]
\label{defn:ancestry}
Fix a $\muLJnorec$-coderivation $\der$. We say that a type occurrence $\sigma$ is an \emph{immediate ancestor} of a type occurrence $\tau$ in $\der$ if they are types in a premiss and conclusion (respectively) of an inference step and, as typeset in~\Cref{fig:sequent-calculus-mulj},~\Cref{fig:unfolding-rules} and~\Cref{fig:unfolding-further-rules}, have the same colour.
If $\sigma$ and $\tau$ are in some $\mygreen \Gamma$ or $\orange{ \Delta}$, then furthermore they must be in the same position in the list.
\end{defn}

Being a binary relation, immediate ancestry forms a directed graph upon which our correctness criterion is built. 
Our criterion is essentially the same as that from \cite{BaeldeDS16Infinitaryprooftheory}, only for $\muLJ$ instead of $\muMALL$.

\begin{defn}
    [Threads and progress]\label{defn:progressiveness}
    A \emph{thread} along (a branch of) $\der$ is a maximal path in $\der$'s graph of immediate ancestry.
We say a thread is \emph{progressing} if it is infinitely often principal and has a smallest infinitely often principal formula that is either a $\mu$-formula on the LHS or a $\nu$-formula on the RHS.
A coderivation $\der$ is \emph{progressing} if each of its infinite branches has a progressing thread.
\end{defn}

We shall use several properties of (progressing) threads in Section~\ref{sec:totality} 
which are relatively standard, e.g.~\cite{Koz83:results-on-mu,Studer08:mu-calc,KupMarVen22:graph-reps-mu-forms}.

\begin{defn}
    [Circular system]
    $\cmuLJ$ is the class of regular progressing $\muLJnorec$-coderivations.
\end{defn}

Referencing Example~\ref{ex:(co)iter-to-cycles}, and for later use, we shall appeal to the notion of \emph{simulation} for comparing models of computation in this work. 
Recalling that we construe $\muLJ$ as a model of computation under $\cutred^*$ and $\cmuLJ$ as a model of computation under $\cutrednorec^*$, we have:

\begin{proposition}[Simulation]
\label{prop:cmulj-simulates-mulj}
$\cmuLJ$ simulates $\muLJ$.
\end{proposition}

\begin{proof}[Proof sketch]
    Replace each instance of a (co)iterator by the corresponding regular coderivation in Example~\ref{ex:(co)iter-to-cycles}.
    Note that those coderivations are indeed progressing due to the progressing thread on $\mu X \sigma(X)$ along the unique infinite branch in the case of $\lr \mu$ (dually for $\rr \nu$).
    The statement follows by closure of $\cmuLJ$ under its rules.
\end{proof}


\begin{figure}
    \centering
      \[\small
  \vlinf{\Nzero}{}{\seqar \N}{}
  \qquad 
  \vlinf{\Nsucc}{}{\mygreen \Gamma \seqar \cyan \N}{\mygreen \Gamma \seqar \cyan \N} 
  \qquad 
  \vliinf{\Ncnd}{}{\mygreen{\Gamma}, \red{\N} \seqar \cyan \sigma}{\mygreen\Gamma \seqar \cyan \sigma}{\mygreen \Gamma, \red \N \seqar \cyan \sigma}
    \]
    \[
    \small
      \vlderivation{
      \vliin{\cut}{}{\Gamma \seqar \sigma}
       {
       \vlin{\rr\N^0}{}{\seqar \N}{\vlhy{}}
       }
      {
 \vliin{\Ncnd}{}{\Gamma, \N \seqar \sigma}{
    \vltr{\derr}{\Gamma \seqar \sigma}{\vlhy{\ }}{\vlhy{}}{\vlhy{\ }}
    }{
    \vltr{\derrr}{\Gamma, \N \seqar \sigma}{\vlhy{\ }}{\vlhy{}}{\vlhy{\ }}
    }
      }
      }
      \quad \cutrednorec \  
      \vlderivation{\vltr{\derr}{\Gamma \seqar \sigma}{\vlhy{\ }}{\vlhy{}}{\vlhy{\ }}} 
      \]
      \[
          \hspace{-1em}
      \small
      \vlderivation{
      \vliin{\cut}{}{\Gamma \seqar \sigma}
       {
       \vlin{\Nsucc}{}{\Gamma \seqar \N}{\vltr{\der}{\Gamma \seqar \N}{\vlhy{\ }}{\vlhy{}}{\vlhy{\ }} }
       }
      {
 \vliin{\Ncnd}{}{\Gamma, \N \seqar \sigma}{\vltr{\derr}{\Gamma \seqar \sigma}{\vlhy{\ }}{\vlhy{}}{\vlhy{\ }}}{\vltr{\derrr}{\Gamma, \N \seqar \sigma}{\vlhy{\ }}{\vlhy{}}{\vlhy{\ }}}
      }
      }   
      \ \cutrednorec
      \vlderivation{
      \vliin{\cut}{}{\Gamma \seqar \sigma}
        {
        \vltr{\der}{\Gamma \seqar \N}{\vlhy{\ }}{\vlhy{}}{\vlhy{\ }}
        }{
      \vltr{\derrr}{\Gamma, \N \seqar \sigma}{\vlhy{\ }}{\vlhy{}}{\vlhy{\ }}        
        }
      }
    \]
    \caption{Native inference rules and cut-reduction steps for $\N$ in $\muLJnorec$.}
    \label{fig:native-n-rules+cutreds-muljnorec}
\end{figure}

\begin{example}
    [Revisiting natural number computation]
    Just like for $\muLJ$, we give native rules for $\N$ in $\muLJnorec$, along with corresponding cut-reductions in Figure~\ref{fig:native-n-rules+cutreds-muljnorec}. We just show how to derive the conditional:
 \[
\vlderivation
{
\vlin{\lr {\mu'}}{}{\Gamma, \N \seqar \sigma}
{
\vliin{\lr{+}}{}{\Gamma, \unit + \N \seqar \sigma}
{
\vlin{\lr \unit}{}{\Gamma, \unit \seqar \sigma}{\vlhy{\Gamma \seqar \sigma}}
}
{
\vlhy{\Gamma , \N \seqar \sigma }
}
}
}
\]
    
    As before, it is routine to show that these reductions are derivable using $\cutrednorec$.
    
Now, specialising our simulation result to recursion on $\N$, we have  the  following regular coderivation for the recursor of system $\T$ (at type $\sigma$):
    \[
           \vlderivation{
            \vliin{\Ncnd }{}{\Gamma, \N \seqar \sigma}
            {
            \vlhy{\Gamma \seqar \sigma}
            }{
            \vliin{\cut}{}{\Gamma, \N \seqar \sigma}
              {
              \vlin{\Ncnd}{\bullet}{\Gamma, \N \seqar \sigma}{\vlhy{\vdots}}
              }{
              \vlhy{\Gamma, \N,\sigma  \seqar \sigma}
              }
            } 
           }
    \]
Indeed it is immediate that $\cmuLJ$ contains circular versions of system $\T$ from \cite{Das21:CT-preprint,Das21:CT-fscd,KuperbergPP21systemT}.
\end{example}

\subsection{Reduction to the negative fragment}
\label{subsec:reduction-to-neg-frag}
It is folklore that coinductive types can be eliminated using inductive types (possibly at the loss of strict positivity) using, say, a version of the G\"odel-Gentzen negative translation, without affecting the class of representable functions (as long as $\N$ is included as a primitive data type) (see, e.g.,~\cite{avigad1998godel}).
Indeed this translation can be designed to eliminate other `positive' connectives too, in particular $+$.\footnote{Note that the attribution of `positive' or `negative' to a connective is unrelated to that of positive or negative context.}

The same trick does not quite work for coderivations since it introduces cuts globally that may break the progressing criterion in the limit of the translation.
\todonew{could elaborate here}
However a version of the Kolmogorov translation, more well behaved at the level of cut-free proof theory, is well suited for this purpose.
In this section we establish such a reduction from $\cmuLJ$ to its `negative' fragment.
Not only is this of self-contained interest, being more subtle than the analogous argument for $\muLJ$ (and type systems with (co)inductive data types), this will also greatly simplify reasoning about the representable functions of $\cmuLJ$ in what follows, in particular requiring fewer cases in arguments therein.

\begin{defn}[Negative fragments] \label{defn:negative-fragments}
We define $\muLJneg$ as the subsystem of $\muLJ$ using only rules and cut-reductions over $\N, \times, \arrow, \mu$.
In particular we insist on the native rules and cut-reductions for $\N$ from \Cref{fig:native-n-rules+cutreds-muljnorec} to avoid extraneous occurrences of $+$ from $\N$ and remain internal to the fragment.
We define $\muLJnegnorec$ and $\cmuLJneg$ similarly, only as subsystems of $\muLJnorec$-coderivations and their cut-reductions.
\end{defn}


The main result of this subsection is:
\begin{proposition}\label{prop:cmulj-into-cmuljminus}
    Any function on natural numbers representable in $\cmuLJ$ is also representable in $\cmuLJneg$.
\end{proposition}
\begin{proof}
[Proof idea]
We give a bespoke combination of a Kolmogorov negative translation and a Friedman-Dragalin `$A$-translation' (setting $A= \N$). We define the translations $\negtrans \cdot$ and $\unnegtrans \cdot$ from arbitrary types to types over $\{\N,\times, \to , \mu\}$ as follows, where $\lnot \sigma \dfn \sigma \arrow \N$:
\[
\begin{array}{r@{\ := \ }l}
\negtrans \sigma & \neg \unnegtrans \sigma\\ 
\unnegtrans X & \negn X \\
    \unnegtrans 1 & \N \\
    \unnegtrans{(\sigma \times \tau)} & \negn (\negtrans \sigma \times \negtrans \tau) \\
    \unnegtrans{(\sigma \to \tau)} & \neg( \negtrans\sigma \to \negtrans \tau) \\
    \unnegtrans{(\sigma+\tau)} & \neg   \negtrans \sigma \times \neg \negtrans \tau \\
    \unnegtrans{(\nu X \sigma)} & \neg \neg \mu X \negn \negtrans\sigma[\negn X / X] \\
    \unnegtrans{(\mu X \sigma)} &  \neg \mu X \negtrans{\sigma}
\end{array}
\]
The translation can be extended to coderivations by mapping every inference rule $\rrule$ to a gadget $\negtrans \rrule$  preserving threads.
Further details are given in~\Cref{subsec:proofs-of-cmulj-into-cmuljminus}.
\end{proof}


\begin{example}\label{exmp:examples-of-coderivations}
The left coderivation of Figure~\ref{fig:examples-coderivations} shows the encoding of a  stream $n_0 :: n_1 :: n_2 \ldots$ by a (not necessarily regular) coderivation. 
Note that this coderivation is regular just if the stream is ultimately periodic.
The right coderivation shows the circular presentation of the  concatenation of a list and a stream into a stream  discussed  in Example~\ref{exmp:examples-of-derivations}. 
Note that, compared to the inductive encoding of this function, the circular one has an arguably more `explicit' computational meaning. 
Both coderivations are progressing, by the  \red{red} progressing threads in their  only infinite branches.

It is worth discussing  how   computation  over streams is simulated in $\cmuLJneg$ via the double negation translation illustrated in Proposition~\ref{prop:cmulj-into-cmuljminus}. 
The type $S$ of streams is translated into $\neg \neg \neg \mu X \neg \neg \neg  (\negtrans \N \times \neg \neg \neg X)$, for some appropriate translation $\negtrans \N$ of the type for natural numbers. Hence,  computation over streams is simulated by computation over a type of the form  $( \sigma\to\N)\to \N$.
Note that this resembles (and embeds) the type $\N\arrow \N$ for representing streams in system $\T$, so in some sense we can see $\negtrans\cdot$-translation as extending/adapting the embedding of $S$ into $\N \arrow \N$.
\end{example}

\begin{figure}
    \centering
    \[
    \small
    \vlderivation{
\vlin{\rr \nu}{}{\seqar \red{S}}
{
\vliin{\rr \times}{}{\seqar \red{N \times S}}
{
\vltr{\cod {n_0}}{\seqar N}{\vlhy{\ }}{\vlhy{\ }}{\vlhy{\ }}
}
{
\vlin{\rr \nu}{}{\seqar \red{S}}
{
\vliin{\rr \times}{}{\seqar \red{N \times S}}
{
\vltr{\cod {n_1}}{\seqar N}{\vlhy{\ }}{\vlhy{\ }}{\vlhy{\ }}
}
{
\vlin{\rr \nu}{}{\seqar \red{S}}
{
\vlhy{\vdots}
}
}
}
}
}
}
\qquad
\vlderivation{
\vlin{\rr \to}{}{\red{L} \seqar S \to S}
{
\vlin{\lr \mu}{\bullet}{\red{L}, S \seqar S}
{
\vliin{\lr +}{}{\red{\unit +(N \times L}), S \seqar S}
{
\vlin{\lr \unit}{}{\unit, S \seqar S}{\vlin{\id}{}{S \seqar S}{\vlhy{}}}
}
{
\vlin{\rr \nu}{}{\red{N \times L}, S \seqar S}
{
\vlin{\lr \times}{}{\red{N \times L}, S \seqar N \times S}
{
\vliin{\rr \times}{}{N, \red{L}, S \seqar N \times S}
{
\vlin{\id}{}{N \seqar N}{\vlhy{}}
}
{
\vlin{\lr \mu  }{\bullet}{\red{L}, S \seqar S}
{
\vlhy{\vdots}
}
}
}
}
}
}
}
}
    \]
    \caption{Left: pointwise computation of a stream in $\muLJnorec$. Right: concatenation of a list and a stream in $\cmuLJ$.}
    \label{fig:list+stream-to-stream-cmuLJ}
     \label{fig:examples-coderivations}
\end{figure}
 \section{Extensions to (un)typed term calculi}
\label{sec:term-calculi}
In light of the reduction to the negative fragment at the end of the previous section, we shall only consider types formed from $\N, \times, \arrow, \mu$ henceforth.

\subsection{From (co)derivations to (co)terms: rules as combinators}
It will be convenient for us to extend our computational model from just (co)derivations to a larger class of untyped (co)terms.
The main technical reason behind this is to allow the definition of a higher-order computability model necessary for our ultimate totality argument for $\cmuLJ$.
At the same time, we obtain a compressed notation for (co)derivations for notational convenience, and indeed carve out typed (conservative) extensions of the proof calculi thusfar considered.

In what follows, we use the metavariables $\rrule $ etc.\ to vary over inference steps of $\muLJneg$ and/or $\muLJnegnorec$, i.e.\ instances of any inference rules in these systems.
\begin{definition}
[(Co)terms~\cite{Das21:CT-fscd}]
 A \emph{coterm}, written $s,t$ etc., is generated coinductively by the grammar:
 \[
 s,t \ \bnf \ \rrule  \ \vert \ s\, t
 \]
I.e.\ coterms are possibly infinite expressions (of depth $\leq \omega$) generated by the grammar above.
A coterm is a \emph{term} if it is a finite expression, i.e.\ generated inductively from the grammar above.
 If all steps in a (co)term are from a system $\mathsf R$, we may refer to it as a $\mathsf R$-(co)term.
%
\end{definition}

Our notion of (co)term is untyped, in that an application $st$ may be formed regardless of any underlying typing. 
(Co)terms will be equipped with a theory that (a) subsumes cut-reduction on (co)derivations; and (b) results in a computational model that is Turing complete.
Before that, however, let us see how (co)derivations can be seen as (co)terms.

\begin{definition}
    [(Co)derivations as (co)terms]
    We construe each $\cmuLJneg$ coderivation as a coterm (and each $\muLJneg$ derivation as a term) by identifying rule application with term application:
        if $\der$ ends with an inference step $\rrule$ with immediate sub-coderivations $\der_1 , \dots, \der_n$ 
        then $\der$ is $\rrule\, \der_1\, \dots \, \der_n$.
    
    Given a set $A$ of (co)terms, the \emph{closure} of $A$, written $\closure A$, is the smallest set of coterms containing $A$ and closed under application, i.e.\ if $s,t \in \closure A$ then also $st \in \closure A$.
\end{definition}

Of course if $\der$ is a derivation, then it is also a term.
Of particular interest to us in this work will be the class $\closure \cmuLJneg$, essentially finitary applications of progressing regular $\muLJnegnorec$-coderivations.

\begin{example}
    [Iterator coderivation as a regular coterm]
    \label{ex:iter-coder-as-coterm}
    Recalling the decomposition of the iterator as a circular coderivation in Example~\ref{ex:(co)iter-to-cycles}, let us specialise to the variation $    \vlinf{\iter{}}{}{\mu X \sigma(X)\seqar \tau}{\sigma(\tau)\seqar \tau}$.
 By following Example~\ref{ex:(co)iter-to-cycles}, we can express $\iter{}\, \der$
    by a regular coderivation, say $\iter{}'(\der)$, that, viewed as a coterm, satisfies the (syntactic) equation,
    \begin{equation}
        \label{eq:iter-coder-syntactic-eqn}
        \iter{}'(\der)\ = \ \mulunf\, (\cut\, \sigma(\iter{}'(\der))\, \der)
    \end{equation}
    Note that $\iter{}'(\der)$ above is indeed a \emph{regular} coterm: it has only finitely many distinct sub-coterms.
\end{example}

\anupam{commented `regular coterms' below}





\subsection{Computational models: theories of (co)terms}
\anupam{need to make clear that we include `open' derivations among (co)terms too? Should check implications of this, probably just a matter of convenience.}

Let us henceforth make the following abbreviations:
\[
\pair \cdot \cdot : 
\vlderivation{
\vliin{\rr \times}{}{\seqar \sigma \times \tau}{
    \vlhy{\seqar \sigma}
}{
    \vlhy{\seqar \tau}
}
}
\qquad
\proj i:
\vlderivation{
\vlin{\lr \times}{}{\sigma_0 \times \sigma_1 \seqar \sigma_i}{
\vlin{\wk}{}{\sigma_0, \sigma_1 \seqar \sigma_i}{
\vlin{\id}{}{\sigma_i \seqar \sigma_i}{\vlhy{}}
}
}
}
\]
\[
\injX\sigma : \vlinf{\rr\mu}{}{\seqar \mu X\sigma(X)}{\seqar \sigma(\mu X \sigma(X))}  
\qquad
\iter{\sigma} :  \vlinf{\lr \mu}{}{\mu X \sigma(X)\seqar \tau}{\sigma(\tau)\seqar \tau}
\]
\[
\numeral 0 : \Nzero \qquad \succ : \vlinf{\Nsucc}{}{\seqar \N}{\seqar \N}\qquad 
\Niter : 
   \vliinf{\lr\N}{}{\N\seqar \sigma}{\seqar\sigma}{\sigma\seqar \sigma}
\]
for $i \in \{0,1\}$.
We may omit types from subscripts when unimportant.
When writing (co)terms using the derivations above, we employ the convention that $\proj i$ and $\succ$ bind stronger than other applications, i.e.\ we write simply $t\, \proj i u$ for $t\, (\proj i u)$ and $t\, \succ u$ for $t\, (\succ\, u)$.

\todo{explain why we factor the $\times$ interactions (and even $\mu$ too) through other constants/gadgets: $\times$ to facilitate totality argument, and $\mu$ for similarity/to slightly facilitate realisability (though not really since we need bespoke iterators there).}
\todo{explain that we don't include $\inj$ for $\N$ to avoid $+$, instead this is hardcoded into the rules}
\anupam{recall that, apart from in, the actual types are not important due to untypedness.}

When referring to an arbitrary instance of a rule, the specification should be understood to be as originally typeset, unless otherwise indicated.
In particular, we follow this convention to define our notion of reduction on coterms:

\begin{figure}
    \centering
\[
\begin{array}{rcl}
  \ax \  x   &\rewrite  & x \\
\ex \ t \ \vec x \ x \ y \ \vec y&\rewrite& t\  \vec x  \ y \ x \ \vec y  \\
\wk \ t \  \vec x\  x     & \rewrite & t\  \vec x\\
\contr \ t \ \vec x\ x &\rewrite&  t \ \vec x \ x \ x \\
\cut \ s \ t \ \vec x\ \vec y &\rewrite& t \ \vec y \ (s \ \vec x)
\end{array}
\qquad
\begin{array}{rcl}
     {\rr \times \ s \ t \ \vec x \ \vec y }&\rewrite& \pair {  s\, \vec x}{t\,  \vec  y}
\\
 \lr \times \ t \, \vec x  \,y & \rewrite &  t\ \vec x \ \proj 0 y\ \proj 1 y 
 \\
 \proj i \pair {x_0} {x_1} & \rewrite & x_i \\
 \rr \imp \, t  \ \vec x \ x  &\rewrite&  t  \ \vec x \ x \\ 
\lr \imp \, s \ t  \ \vec x\ \vec y \ z  &\rewrite& t \ \vec y \ (z  \, (s \ \vec x )) 
\end{array}
\]
    \caption{Reduction for $\LJneg$ (both $\reduces$ and $\reducesnorec$). }
    \label{fig:reduction-ljneg}
\end{figure}

\begin{figure}
    \centering
    \[
\begin{array}{rcl}
 \rrule \ t \ \vec x &\reducesnorec&  t \ \vec x \qquad  \rrule\in \{\rr \mu, \mulunf \} 
\end{array}
\]
    \caption{Reduction for least fixed point rules in $\muLJnorec$. }
    \label{fig:reduction-mu-norec}
\end{figure}

\begin{figure}
    \[
    \begin{array}{r@{\ \reduces\ }l}
         \lr \mu \, s \, t\, \vec x\, \vec y\, z & t\, \vec y\, (\iter{}\, (s\, \vec x)\, z) \\
         \rr \mu\, t\, \vec x & \injX\sigma (t\, \vec x) \\
         \iter{}\, t\, x & x\, t \\
         \injX \sigma\,  x\, t & t\, (\sigma(\iter\, t)\, x)
    \end{array}
    \]
    \caption{Reduction for least fixed points in $\muLJ$.}
    \label{fig:reduction-mu-iter}
\end{figure}

\begin{figure}
  \[
   \begin{array}{r@{\ \reduces\ }l}
         \lr \N\, s\, t\, u\, \vec x\, \vec y\, z & u\, \vec y\, (\Niter\, (s\, \vec x)\, (t\, \vec x)\, z) \\
        { \Nsucc \, t \, \vec x  }&{ \succ \, (t\, \vec x) } \\
         \Niter\, s\, t\, \numeral 0 & s \\
         \Niter\, s\, t\, \succ x & t\, (\Niter \, s\, t\, x)
    \end{array}
  \]
    \caption{Reduction for $\N$ in $\muLJneg$}
    \label{fig:reduction-N-muLJ}
\end{figure}

\begin{figure}
   \[
   \begin{array}{r@{\ \reducesnorec\ }l}
        \Ncnd \, s\, t\, \vec x\, \numeral 0 & s\, \vec x \\
        \Ncnd \, s\, t\, \vec x\, \succ y & t\, \vec x\, y
   \end{array}
   \]
    \caption{Reduction for $\N$ in $\muLJnegnorec$.}
    \label{fig:reduction-N-muLJnorec}
\end{figure}

\begin{definition}
    [Theories]
    We define two (context-closed) reduction relations on (co)terms: 
    \begin{itemize}
        \item $\reduces$ is generated by the clauses in Figures~\ref{fig:reduction-ljneg}, \ref{fig:reduction-mu-iter} and \ref{fig:reduction-N-muLJ}.
        \item $\reducesnorec$ is generated by the rules in Figures~\ref{fig:reduction-ljneg}, \ref{fig:reduction-mu-norec} and \ref{fig:reduction-N-muLJnorec}. 
    \end{itemize}
    In all cases the lengths of vectors $\vec x$ and $\vec y$ match those of the relevant contexts $\Gamma$ and $\Delta$ from the original typesetting of the rules, i.e.\ from \Cref{fig:sequent-calculus-mulj,fig:co-iteration-rules,fig:unfolding-rules,fig:native-n-rules-mulj,fig:native-n-rules+cutreds-muljnorec}.
    Note that the use of both variables and term metavariables in these reduction rules is purely to aid parsing. All reduction rules are closed under substitution.
    
    When referring to (fragments of) $\closure \muLJneg$ as a computational model, we typically mean with respect to $\reduces$, and when referring to (fragments of) $\closure \cmuLJneg$ as a computational model, we typically mean with respect to $\reducesnorec$. {However we also consider a (weakly) extensional version of  $\convertsnorec$:}
    \begin{itemize}
        \item $\convertsnoreceta$ is the closure of $\convertsnorec$ under the rule   $\vlinf{\extensional}{}{ t\convertsnoreceta t'}{t\, x \convertsnoreceta  t'\, x}$.
    \end{itemize}
    Above $x$ must be a fresh variable, not a general (co)term.
\end{definition}

Admitting some extensionality is not necessary to reason about representability, since extensionality can be eliminated for low type levels, but simplifies some of the theorem statements.

\begin{example}
[Iteration equations]
The fundamental equation for iteration is indeed derivable by $\reduces$:
\[
\begin{array}{r@{\ \reduces\ }l}
     \iter{}\, P\, (\injX\sigma\, x) & \injX \sigma \, x\, P  \\
     & P\, (\sigma(\iter\, P)\, x)
\end{array}
\]
For $\reducesnorec$, recalling Example~\ref{ex:iter-coder-as-coterm} and using its notation,
we can simulate the iteration equation for $\iter\, P$ with $\iter{}'(P)$:
\[
\begin{array}{r@{\ \reducesnorec\ }ll}
     \iter{}'(P)\, (\injX\sigma\, x) & \iter{}'(P)\, x & \text{by $\injX\sigma$ reduction} \\
     & \cut\, \sigma(\iter{}'(P))\, P\, x & \text{by $\mulunf$ reduction} \\
     & P\, (\sigma(\iter{}'(P))\, x) & \text{by $\cut$ reduction}
\end{array}
\]
\anupam{point out that, at this juncture, $\reduces$ and $\reducesnorec$ do not agree/are not compatible, unlike $\cutred^* \subseteq \cutrednorec^*$? check this.}
\end{example}

\todo{need to mention clashing reductions and problem of confluence. we address this shortle by giving an interpretation in lambda terms.}

More importantly for us, our notion of extensional reduction on coterms subsumes that of cut-reduction on coderivations.
Since we have identified coderivations as coterms, we may state this rather succinctly, constituting the main result of this subsection:
\begin{theorem}
[Extensional reduction includes cut-reduction]
\label{thm:reduction-inludes-cut-elimination}
${\cutrednorec}\subseteq {\convertsnoreceta}$.
\end{theorem}
\begin{proof}
    We show that, if $\der \cutrednorec \der'$ then, for some $n \geq 0$ sufficiently large, and for any $x_1, \ldots, x_n$, $\der\, x_1 \, \ldots \,  x_n\convertsnorec \der'\, x_1 \, \ldots \,  x_n $. We then conclude by repeatedly applying the rule $(\extensional)$. \anupam{would be easier to not have two distinct steps and just proceed directly.}


Suppose $\der \cutrednorec  \der'$. It suffices to consider the case where  the last rule of $\der$ is the cut rule rewritten by the  cut-reduction step. Indeed, the latter implies the general statement by appealing to the context closure of $ \reducesnorec$. We only consider some relevant cases. 

If $\der$ has the form,
\[
\vlderivation{
\vliin{\cut}{}{\Gamma, \Delta, \Sigma \seqar \gamma}
   {\vliin{\rr \times}{}{\Gamma,\Delta \seqar \sigma \times \tau}{\vldrs{\der_1}{\Gamma \seqar \sigma }}{\vldrs{\der_2}{\Delta \seqar \tau}}
   }
   {
\vlin{\lr \times}{}{\Sigma , \sigma \times \tau \seqar \gamma}{\vldrs{\der_3}{\Sigma, \sigma, \tau \seqar \gamma}}
   }
   }
\]
we have:
\[
\arraycolsep=2pt
\begin{array}{rcl}
\cut \, (\rr \times \, \der_1\, \der_2)\, (\lr \times\, \der_3) \, \vec x\, \vec y\, \vec z&\reducesnorec&
\lr \times \, \der_3 \, \vec z\, (\rr \times \, \der_1 \, \der _2\, \vec x\, \vec y)\\
& \reducesnorec & \lr \times \, \der_3 \, \vec z\, \pair{\der_1 \, \vec x}{\der_2\, \vec y}\\
& \reducesnorec^*  &   \der_3 \, \vec z\, (\der_1\, \vec x)\, (\der_2\, \vec y)\\ \\ 
\cut \, \der_2 \, (\cut \, \der_1\, \der_3)\, \vec x\, \vec y\, \vec z  &\reducesnorec & 
   (\cut \, \der_1 \, \der_3)\, \vec x\, \vec z\, (\der_2\, \vec y)\\
  & \reducesnorec  &    \der_3\, \vec z\, (\der_1 \, \vec x)\, (\der_2\, \vec y)
\end{array}
\]

If $\der$ has the form,
\[
\vlderivation{
\vliin{\cut}{}{\Gamma, \Delta, \Sigma \seqar \gamma}
   {
   \vlin{\rr \to}{}{\Gamma \seqar \sigma \to \tau}{\vldrs{\der_1}{\Gamma, \sigma \seqar \tau}}
   }{
   \vliin{\lr \to}{}{\Delta, \sigma \to \tau, \Sigma\seqar \gamma}{\vldrs{\der_2}{\Delta \seqar \sigma}}{\vldrs{\der_3}{\Sigma, \tau \seqar \gamma}}
   }
}
\]
we have:
\[\arraycolsep=2pt
\begin{array}{rcl}
\cut\, (\rr \to \der_1)\, (\lr \to  \der_2\, \der_3)\, \vec x\, \vec y\, \vec z & \reducesnorec & \lr \to  \der_2\, \der_3\, \vec y\, (\rr \to \der_1\, \vec x)\, \vec z \\
    &\reducesnorec & \der_3 \, \vec z \, (\rr \to \der_1 \, \vec x \, (\der_2\, \vec y)) \\
    &\reducesnorec & \der_3 \, \vec z \, ( \der_1 \, \vec x \, (\der_2\, \vec y)) 
    \\
    \\
\cut \, (\cut \, \der_2\, \der_1)\, \der_3\, \vec x\, \vec y \, \vec z &\reducesnorec & \der_3\, \vec z\, (\cut \, \der_2\, \der_1\, \vec x\, \vec y)\\
&\reducesnorec & 
\der_3 \, \vec z\, (\der_1 \, \vec x\, (\der_2 \, \vec y))
\end{array}
\]

If $\der$ has the form,
\[
\vlderivation{
\vliin{\cut}{}{\Gamma, \Delta \seqar \gamma}
    {
    \vlin{\rr \mu}{}{\Gamma \seqar \mu X \sigma(X)}{\vldrs{\der_1}{\Gamma \seqar \sigma (\mu X \sigma(X))}}
    }{
    \vlin{\lr \mu ' }{}{\Delta, \mu X \sigma (X) \seqar \gamma}{\vldrs{\der_2}{\Delta, \sigma (\mu X \sigma(X))\seqar \gamma}}
    }
}
\]
we have:
\[
\arraycolsep=2pt
\begin{array}[b]{rcl}
\cut \, (\rr \mu \, \der_1)\, \lr \mu' \, \der_2\, \vec x\, \vec y&\reducesnorec &  (\lr \mu' \, \der_2)\, \vec y \, (\rr \mu \, \der_1\, \vec x)\\
&\reducesnorec^* &   \der_2 \, \vec y \, ( \der_1\, \vec x)\\ 
\noalign{\medskip}
\cut \, \der_1 \, \der_2 \, \vec x\, \vec y&\reducesnorec &   \der_2 \, \vec y \, ( \der_1\, \vec x) 
\end{array}
\qedhere
\]
\end{proof}

 \subsection{An embedding into $\lambda$-terms}
 While the significant technical development of this work involves `totality' arguments, e.g.\ in Section~\ref{sec:totality} showing that the representable partial functions of $\cmuLJneg$ (under $\convertsnoreceta$) are total, we better address \emph{determinism} too.
 
 As it stands, $\convertsnorec$ and $\convertsnoreceta$ may fire distinct reductions on the same (co)terms, so there is a priori no guarantee that the output of representable functions is unique.
However
this can  be shown by  defining a straightfoward  interpretation of coterms of $\closure \cmuLJneg$ into the (untyped) $\lambda$-calculus. We shall prove that   the (untyped) $\lambda$-calculus, under $\beta\eta$-reduction, simulates $\closure\cmuLJneg$ under $\reducesnoreceta$. This simulation relies on the fact that we can express regular coterms as a finite system of equations, which are known to always have solutions in the (untyped) $\lambda$-calculus.
As the techniques are rather standard, we shall be quite succinct in the exposition.



\anupam{I commented a remark below that seems completely out of place.}


\emph{$\lambda$-terms}, written $s,t$ etc., are generated as usual by:
\[
s,t \quad ::= \quad 
x \quad | \quad (t\, s) \quad | \quad \lambda x \,  t 
\]
We write $\Lambda$ for the set of all terms.
The notion of `free variable' is defined as expected, and we write $\fv t $ for the set of free variables of the term $t$.

\newcommand{\doublepair}[2]{\langle \! \langle {#1}, {#2}\rangle \! \rangle}
\newcommand{\doubleproj}[2]{\pi_{#1}\,  {#2}}
\newcommand{\godel}[1]{\lceil {#1} \rceil}

\begin{figure}[t]
    \centering
 \[
 \small
\begin{array}{rcl}
 \true&\dfn& \lambda x,y. x\\
      \false &\dfn& \lambda x,y. y\\
 \ifthen s t u   & \dfn & s \, t \, u \\
      \doublepair s t & \dfn & \lambda z.z\, s\, t\\
      \doubleproj{0}{} &\dfn&  \lambda z. z \, \true \\
         \doubleproj{1}{} &\dfn&  \lambda z. z \, \false \\
  \godel {\numeral 0} & \dfn & \doublepair \false  {\lambda x.x}   \\
    \godel {\numeral{n+1}} &\dfn  & \doublepair \true {\godel n} \\
    \succ &\dfn & \lambda x. \doublepair \true   x \\
     \pred  &\dfn & \lambda x. \doubleproj{1}{x}  \\
    \cnd &\dfn &\lambda x,y,z. \ifthen{(\doubleproj 0 z)}{x}{y\, ({\pred \, z})}  
      \end{array}
      \begin{array}{rcl}
\ifthen {\enc{\true}} s t & \betaeta & s\\
\ifthen {\enc{\false}} s t & \betaeta & t\\
 \doubleproj i {\doublepair {s_0} {s_1}} & \betaeta & s_i\\
 \pred \, {\godel {\numeral 0}} &\betaeta & \godel 0\\
 \pred\,  {(\succ \, {\godel {\numeral n}})} &\betaeta & \godel {\numeral n}\\
  \cnd \, {\godel {\numeral 0}} \, s \, t &\betaeta & s\\
 \cnd \, {(\succ \, {\godel {\numeral n}})}\, s \, t &\betaeta & t \, \godel{n}\\
 \end{array}
\]

    \caption{Macros for $\lambda$-terms  and corresponding reductions.}
    \label{fig:abbreviation}
\end{figure}

\begin{figure}[t]
    \centering
   \[
   \small
   \begin{array}{rcl}
       \sem{\ax}& \dfn& \lambda x.x \\
        \sem{\ex}&\dfn &\lambda t. \lambda \vec u. \lambda x. \lambda y. \lambda \vec v. t\, \vec u \,y\,x\, \vec v \\
  \sem{\cut}&\dfn &\lambda t. \lambda s. \lambda \vec{u}.\lambda \vec v. s\, \vec v \, (t \,\vec u)  \\ 
  \sem{\contr}&\dfn &\lambda t. \lambda \vec u. \lambda x. t\,\vec u \,x\,x\\ 
  \sem{\wk}&\dfn &\lambda t. \lambda \vec u. \lambda x. t \,\vec u  \\
   \sem{\rr \arrow}&\dfn& \lambda t. \lambda \vec u. \lambda x. t \,\vec u\, x\\  
   \sem{\lr \arrow}&\dfn& \lambda t. \lambda s. \lambda \vec u. \lambda \vec v. \lambda x.  s\,\vec v\,(x\,(t\,\vec u))\\
   \end{array}
   \qquad
   \begin{array}{rcl}
         \sem{\rr \times }& \dfn&  \lambda t. \lambda s. \lambda \vec{u}.\lambda \vec v . \doublepair {t\, \vec u}{s \, \vec v} \\
  \sem{\lr \times}&\dfn &\lambda t. \lambda \vec u. \lambda x. {t\,  \vec u \, (\doubleproj0 x) \, (\doubleproj1 x)}  
\\
\sem{\rr \mu}&\dfn &\lambda t. \lambda \vec{u}. t\, \vec u \\
  \sem{\lr \mu}&\dfn& \lambda t. \lambda \vec{u}. \lambda z. t \, \vec u \, z  \\ 
  \sem{\Nzero}&\dfn & \godel{0}\\
  \sem{\Nsucc}&\dfn& \lambda t.\lambda \vec x. \succ \, {(t\, \vec x)}\\
  \sem{\Ncnd} &\dfn & \lambda t. \lambda s. \lambda \vec u. \lambda x. \cnd\, {x}\, {(t \, \vec u)}\, {(s\, \vec u)}
   \end{array}
   \]
    \caption{Translation of  $\cmuLJneg$ into $\lambda$-terms.}
    \label{fig:translation-of-rules}
\end{figure}

We work with a standard equational theory on $\lambda$-terms. 
$\betaeta$ is the smallest congruence on $\lambdaterms$ satisfying:
\[
\lambda x t\,  s \betae  t[s/x] \qquad \qquad \lambda x ( tx) =_\eta t \quad (x \not \in  \FV(t))
\]
\Cref{fig:abbreviation} displays some macros for $\lambda$-terms (and the corresponding reduction rules) we will adopt in this subsection.
 We define an interpretation of coterms in $\closure \cmuLJneg$ into $\Lambda$.  
 %
 %
 We start with an interpretation of the basic inference rules: 

\todonew{A: I think we need to interpret regular coderivations first, and then extend to terms of $\closure \cmuLJneg$}

 \begin{defn}[Interpreting rules] To each inference step $\rrule$ of $\cmuLJneg$ we associate a $\lambda$-term  $\sem \rrule$ as shown in~\Cref{fig:translation-of-rules}. 
 \end{defn}
 
It is easy to see that this interpretation preserves equations from~\Cref{fig:reduction-ljneg},~\ref{fig:reduction-mu-norec}, and~\ref{fig:reduction-N-muLJnorec}:

 \begin{lemma} \label{lem:rewriting-into-beta}The following equations hold in $\Lambda$:
    \[
\begin{array}{rcl}
 \sem \ax \  x   &\betaeta  & x \\
 \sem\ex \ t \ \vec x \ x \ y \ \vec y&\betaeta& t\  \vec x  \ y \ x \ \vec y  \\
 \sem\wk \ t \  \vec x\  x     & \betaeta & t\  \vec x\\
 \sem\contr \ t \ \vec x\ x &\betaeta&  t \ \vec x \ x \ x \\
 \sem\cut \ s \ t \ \vec x \ \vec y &\betaeta& t \ \vec y \ (s \ \vec x)\\
{ \sem{ \rr \times} \ s \ t \ \vec x \ \vec y }&\betaeta& \doublepair {  s\, \vec x}{t\,  \vec  y}
\\
  \sem{\lr \times} \ t \, \vec x  \,y & \betaeta &  t\ \vec x \ (\sem{\proj 0} y)\ (\sem{\proj 1} y )
\end{array}
\qquad
\begin{array}{rcl}
 \sem{\proj i }\doublepair {x_0} {x_1} & \betaeta & x_i \\
      \sem{ \rr \imp} \ t  \ \vec x \ x  &\betaeta&  t  \ \vec x \ x \\ 
 \sem{\lr \imp} \ s \ t  \ \vec x\ \vec y \ z  &\betaeta& t \ \vec y \ (z  \, (s \ \vec x )) \\
 \sem{ \rr \mu} \ t \ \vec x &\betaeta&  t \ \vec x \\
  \sem{ \lr \mu} \ t \ \vec x &\betaeta&  t \ \vec x \\
  \sem{ \Ncnd}\, s\, t\, \vec x\, \godel{0} &\betaeta& s\, \vec x \\
      \sem{   \Ncnd} \, s\, t\, \vec x\, \succ y &\betaeta&  t\, \vec x\, y
\end{array}
\]
Also, if $\sem{t}\, x \betaeta \sem{t'}\, x$, for $x\notin \fv t$, then $\sem{t}\betaeta \sem{t'}$ by $\eta$-reduction in $\Lambda$.
\end{lemma}

 We now show how to extend  $\sem{\_}$ to regular coderivations, by noting that any such coterm can be described by a finite system of simultaneous equations. 
 From here we rely on a well-known result that such finite systems of equations always admit solutions in the untyped lambda calculus, with respect to $\betaeta$ (see, e.g., \cite{Hindley}).

 \begin{definition}
 [Interpreting regular coderivations]
     Consider a $\cmuLJneg$ coderivation $P$ with subcoderivations $\vec P = P_1, \dots , P_n$ where $P = P_1$.
 Suppose each $P_i$ is concluded by an inference step $\rrule_i$ with immediate subcoderivations $\vec P_i$ (a list of regular coderivations among $\vec P$).
 Write $\Sys \der$ for the system of equations $\{x_i = \sem{\rrule_i} \vec x_i\}_{i=1}^n$, where we set $\vec x_i \dfn x_{i_1}, \dots, x_{i_k}$ when $\vec P_i = P_{i_1},\dots, P_{i_k}$.
We define $\sem P$ to be some/any solution to $x_1$ of $\Sys P$ in $\Lambda$, with respect to $\betaeta$.

 From here we extend the definition of $\sem\cdot $ to all coterms in $\closure \cmuLJneg$ inductively as expected, setting $\sem{ts}\dfn \sem t \sem s$.
 \end{definition}

\todonew{A: I just inlined the actual equations into a definition above and skipped the intermediate lemma here. I've glossed over the narrative too.}

Now, immediately from the definition of $\sem\cdot$ and \Cref{lem:rewriting-into-beta} we arrive at our intended interpretation:

\begin{proposition}\label{prop:cut-elimination-reduction} If $t \in \closure \cmuLJneg$ and $t \reducesnoreceta s$ then $\sem{t}\betaeta \sem{s}$. \anupam{added eta to the reduction (and removed $*$ for typesetting reasons). this is needed in the next corollary and follows from the results referenced in the proof.}
\end{proposition}

From here, by confluence of $\beta\eta$-reduction on $\lambda$-terms, we immediately have:
\begin{corollary}[Uniqueness]\label{thm:uniqueness} Let $t \in \closure \cmuLJneg$. If $\cod m \convertsnoreceta t \convertsnoreceta \cod{n}$ then $n=m$.
\end{corollary}
\begin{proof}
    Clearly, by Proposition~\ref{prop:cut-elimination-reduction} we have $\sem{t} \betaeta \sem{\cod{n}}$ and $\sem{t} \betaeta \sem{\cod{m}}$. Since $\sem{\cod{n}}\betaeta \godel{\cod n}$, we have $\sem{t} \betaeta \godel{\cod n}$ and $\sem{t} \betaeta \godel{\cod m}$. Since $ \godel{\cod n}$ and $\godel{\cod m}$ are normal forms, by confluence of  $\Lambda$ it must be that  $ \godel{\cod n}=\godel{\cod m}$, and hence $n=m$.
\end{proof}

\subsection{From typed terms back to proofs}\label{subsec:from-typed-terms-back-to-proofs.}
Let us restrict our attention to $\muLJneg$-terms in this subsection. 
{In what follows, for a list of types $\vec \sigma=(\sigma_1, \ldots, \sigma_n)$,  we write  $\vec \sigma \to \tau$ for $\sigma_1 \to \ldots \to \sigma_n \to \tau$}.
In order to more easily carry out our realisability argument in Section~\ref{sec:realisability}, it will be convenient to work with a typed version of $\clomuLJneg$:

\begin{definition}
    [Type assignment]
    \emph{Type assignment} is the smallest (infix) relation `$:$' from terms to types satisfying:
    \begin{itemize}
        \item for each step $\vliiinf{\rrule}{}{\vec \sigma \seqar \tau}{\vec \sigma_1\seqar \tau_1}{\cdots}{\vec \sigma_n\seqar \tau_n}$ we have $\rrule : (\vec \sigma_1 \arrow \tau_1) \arrow \cdots \arrow (\vec \sigma_n \arrow \tau_n)\arrow \vec \sigma \arrow \tau$.
        \item if $t:\sigma\arrow \tau$ and $s:\sigma$ then $ts:\tau$.
        \item if $t:\mu X \sigma(X)$ and $s:\sigma(\tau)\arrow \tau$ then $ts:\tau$.
    \end{itemize}
  We write $\typedclomuLJneg$ for the class of typed $\clomuLJneg$-terms.
\end{definition}

The main result of this subsection is:
\begin{theorem}
    [Terms to derivations]\label{thm:terms-to-derivations}
  The natural number functions represented by $\typedclomuLJneg$ are already representable by $\muLJneg$.
\end{theorem}
\begin{proof}
[Proof sketch]
First, given a derivation $\der$ of ${\seqar \sigma \to \tau}$ in $\muLJneg$, we define the derivation  $\mathsf{uncurry}(\der)$ of $\sigma  \seqar \tau$ as  follows:
\[
\vlderivation{
\vliin{\cut}{}{\sigma \seqar  \tau}
{
\vldr{\der}{\seqar \sigma \to \tau}
}
{
\vliin{\lr \to}{}{\sigma, \sigma \to \tau \seqar \tau}{\vlin{\id}{}{\sigma \seqar \sigma}{\vlhy{}}}{\vlin{\id}{}{\tau \seqar \tau}{\vlhy{}}}
}
}
\]
We now define an interpretation of type assignments $t:\tau$ into derivations $\der_{\metader} $ of $ \seqar \tau$ by induction on $t:\tau$ as follows:
\begin{itemize}
     \item For each step $\rrule : (\vec \sigma_1 \arrow \tau_1) \arrow \cdots \arrow (\vec \sigma_n \arrow \tau_n)\arrow \vec \sigma \arrow \tau$, $\der_\rrule$ is the derivation of $\rrule$ in $\muLJneg$:
     \[
     \vlderivation{
     \vliq{\rr \to}{}{\seqar (\vec \sigma_1 \arrow \tau_1) \arrow \cdots \arrow (\vec \sigma_n \arrow \tau_n)\arrow \vec \sigma \arrow \tau}{
     \vliiin{\rrule}{}{\vec \sigma_1 \to \tau_1, \dots, \vec \sigma_n \to \tau_n , \vec \sigma \seqar \tau}{
        \vliq{}{}{\vec \sigma_1 \to \tau_1, \vec \sigma_1 \seqar \tau_1}{\vlhy{}}
     }{
        \vlhy{\cdots}
     }{
        \vliq{}{}{\vec \sigma_n \to \tau_n, \vec \sigma_n \seqar \tau_n}{\vlhy{}}
     }
     }
     }
     \]
    
    \item If $t:\sigma \to \tau$ and $s: \sigma$ then $\der_{ts}$ is defined as: 
    \[
    \vlderivation{
    \vliin{\cut}{}{\seqar \tau}{
        \vltr{\der_{s}}{\seqar \sigma}{\vlhy{\quad }}{\vlhy{}}{\vlhy{\quad }}
    }{
        \vltr{\colorbox{white}{$\scriptstyle\mathsf{uncurry}(\der_{t})$}}{\sigma \seqar \tau}{\vlhy{\quad }}{\vlhy{}}{\vlhy{\quad }}
    }
    }
    \]
    \item If $t: \mu X \sigma(X)$ and $s:\sigma(\tau) \to \tau$ then $\der_{ts}$ is defined as:
    \[
\vlderivation{
\vliin{\cut}{}{\seqar \tau}
{
\vldr{\der_{t}}{\seqar \mu X \sigma (X)}
}
{
\vliin{\lr \mu}{}{ \mu X \sigma (X)\seqar \tau}
{
    \vldr{\colorbox{white}{$\scriptstyle\mathsf{uncurry}(s)$}}{\sigma(\tau)\seqar  \tau}
    }
    {
    \vlin{\id}{}{\tau \seqar \tau}{\vlhy{}}
    }
    }
    }
    \]
\end{itemize}

\todonew{A: now that type assignment is deterministic, no longer need dependency on derivation. I've changed the proof reflecting that, old cases commented below.}

Now, to show that  typed-$\closure \muLJneg$  and $\muLJneg$ represent the same functions on natural numbers, it suffices to prove that if $ t: \tau$  and $t \reduces t'$ then $\der_{t}\cutred^* \der_{t'}$.
For this let us observe:
\begin{itemize}
    \item Whenever $t$ is typed, so are all its subterms by definition of type assignment.
    \item Any derivation $\der_{t(s)}$ in the form $\der_t(\der_s)$, for appropriate $\der_t(\cdot)$ (with leaf $\cdot$).
\end{itemize}
Thus the simulation of any reduction step $t(s) \reduces t(s')$, with redex $s$, is reduced to showing $s\reduces s' \implies \der_s \cutred^* \der_{s'}$.
%
This boils down to checking that the reductions in~\Cref{fig:reduction-ljneg},~\Cref{fig:reduction-mu-iter} and~\Cref{fig:reduction-N-muLJ} are simulated by a series of cut-reduction rules on $\muLJneg$, which  is routine. 
\end{proof}


 \section{Totality of circular proofs}
\label{sec:totality}

\comment{
Minor things:
\begin{itemize}
    \item explain difference with interpretation of fixed points in models of F
    \item mention, maybe in previous section, that untyped terms are Turing-complete.
    \item need to mention that similar semantic techniques in literature, e.g. Mendler, Barendregt... this also highlights differences with F.
    \item make a comment about (weak) extensionality: we do not take full extensionality to keep equational theory semi-recursive, avoiding extensionality-elimination techniques à la Luckhardt.
    \item notation: $\clomuLJneg$ should have a default equational theory, maybe put $\eta$ in sub/superscript
\end{itemize}
}

In this section we provide a semantics for (circular) proofs, using computability theoretic tools.
Our aim is to show that $\cmuLJ$ represents only total functions on $\Nat$ (Corollary~\ref{cor:cmuLJ-represents-only-total-fns}), by carefully extending circular proof theoretic techniques to our semantics.

Throughout this section we shall only consider types formed from $\N, \times, \arrow, \mu$, unless otherwise indicated.

\subsection{A type structure of regular coterms}
We shall define a type structure whose domain will be contained within $\closure \cmuLJneg$. 
Before that, it will be convenient to have access to a notion of a `good' set of terms.

A \emph{(totality) candidate} is some $A\subseteq \closure \cmuLJneg$ that is closed under $\convertsnoreceta$.
We henceforth expand our language of types by including each candidate $A$ as a type constant. 
%
An immediate albeit powerful observation is that the class of candidates forms a complete lattice under set inclusion.
This justifies the following definition of our type structure:
\begin{definition}
    [Type structure]
    \label{def:hr}
    For each type $\sigma$ we define $\hr\sigma \subseteq \closure \cmuLJneg$ by:
    \[
    \begin{array}{r@{\ := \ }l}
         \hr A & A \\
         \hr \N & \{ t  \ \vert \ \exists n \in \Nat .\, t\convertsnoreceta \numeral n\} \\
         \hr {\sigma \times \tau} & \{ t \ \vert\ \proj 0 t \in \hr \sigma \text{ and }  \proj 1 t \in \hr \tau\} \\
         \hr {\sigma \to \tau} & \{ t \ \vert\ \forall s \in \hr \sigma. \, ts \in \hr\tau\} \\
         \hr {\mu X\sigma(X)} & \bigcap \{ A \text{ a candidate} \ \vert\ \hr{\sigma(A)}\subseteq A\}
    \end{array}
    \]
    We write $\hr{\sigma(\cdot)}$ for the function on candidates $A\mapsto \hr{\sigma(A)}$.
\end{definition}
As we shall see, the interpretation of $\mu$-types above is indeed a least fixed point of the corresponding operation on candidates.

\begin{remark}
    [Alternative SO interpretation]
    Recalling the second-order interpretation of $\mu$-types, \Cref{rem:mulj-frag-F}, an alternative definition of $\hr {\mu X \sigma(X)}$ could be $\bigcap_A (\sigma(A)\to A) \to A$.
    This gives rise to a different type structure, indeed similar to the realisability model we give later in \Cref{sec:realisability}.
    However such a choice does not allow us to readily interpret $\clocmuLJneg$: the totality argument in this section, \Cref{thm:prog-implies-hr}, crucially exploits the fact that $\mu$-types are interpreted as bona-fide least fixed points.
\end{remark}

A routine but important property is:
\begin{proposition}
    [Closure under conversion]
    \label{prop:hr-closed-under-conversion}
    If $t \in \hr \tau$ and $t \convertsnoreceta t'$ then $t'\in \hr\tau$.
\end{proposition}
\begin{proof}
    By induction on the structure of $\tau$.
    The base cases when $\tau$ is a candidate $A$ or the type $\nat$ follow immediately from the definitions. For the remaining cases:
    \begin{itemize}
        \item If $\tau = \tau_0 \times \tau_1$, then $\proj i t' \convertsnoreceta \proj i t \in \hr{\tau_i}$ by IH, for $i=0,1$, so indeed $t'\in \hr \tau$.
        \item If $\tau = \tau_0 \to \tau_1$ and $s \in \tau_0$, then $t's \convertsnoreceta ts \in \hr{\tau_1}$ by IH, so indeed $t'\in \hr\tau$.
        \item If $\tau = \mu X \tau'(X)$ and $A$ is a candidate with $\hr{\sigma(A)} \subseteq A$, then $t' \in A$ by IH, so indeed $t'\in \hr\tau$.\qedhere
    \end{itemize}
\end{proof}

Let us point out that this immediately entails, by contraposition and symmetry of $\convertsnoreceta$, closure of \emph{non-elementhood} of the type structure under conversion: if $t\notin \hr \tau$ and $t\convertsnoreceta t'$ then also $t'\notin \hr\tau$.

The main result of this section is:
\begin{theorem}
[Interpretation]
\label{thm:prog-implies-hr}
For any $\cmuLJneg$-coderivation $\der:\Sigma \seqar \tau$ and $\vec s \in \hr\Sigma$ we have $\der\vec s \, \in  \hr{ \tau}$.
\end{theorem}

\anupam{I moved the proof of this to the end of section, as it is way too early here. We need the notions built up over the next subsections for it to make sense. Added sentence below.}

The rest of this section is devoted to proving this result, but before that let us state our desired consequence:




\begin{corollary}
\label{cor:cmuLJ-represents-only-total-fns}
$\cmuLJneg$ represents only total functions on $\Nat$ with respect to $\convertsnoreceta$.
\end{corollary}
\begin{proof}
[Proof idea]
Consider a $\cmuLJneg$-coderivation $P: \vec \N \seqar \N$.
    By \Cref{thm:prog-implies-hr} and closure under conversion, Proposition~\ref{prop:hr-closed-under-conversion} (namely applying $\cut$-reductions), we have for all $\vec m \in \Nat$ there is $n\in \Nat$ with  $P\numeral{\vec m} \convertsnoreceta \numeral n$.
\end{proof}

\subsection{Montonicity and transfinite types}
\label{subsec:monotonicity-and-transfinite-types}
To prove our main Interpretation Theorem, we shall need to appeal to a lot of background theory on fixed point theorems, ordinals and approximants, fixed point formulas, and cyclic proof theory. 
In fact we will go on to formalise this argument within fragments of second-order arithmetic.

Since the class of candidates forms a complete lattice under set inclusion, we can specialise the well-known Knaster-Tarski fixed point theorem:
\begin{proposition}
    [Knaster-Tarski for candidates]
    \label{prop:kt-on-candidates}
    Let $F$ be a monotone operation on candidates, with respect to $\subseteq$.
    $F$ has a least fixed point
     $\mu F = \bigcap \{A \text{ a candidate}\ \vert\ F(A)\subseteq A\}$.
\end{proposition}

At this point it is pertinent to observe that the positivity constraint we impose for fixed point types indeed corresponds to monotonicity of the induced operation on candidates with respect to our type structure:
\begin{lemma}
[Monotonicity]
\label{lem:mon-pos-types-in-hr}
Let $A\subseteq B$ be candidates.
\begin{enumerate}
    \item 
If $\sigma(X)$ is positive in $X$ then $\hr{\sigma(A)}\subseteq \hr{\sigma(B)}$; 
\item If $\sigma(X)$ is negative in $X$ then $\hr{\sigma(B)}\subseteq \hr{\sigma(A)}$.
\end{enumerate}
\end{lemma}
These properties are proved (simultaneously) by a straightforward induction on the structure of $\sigma(X)$. \todo{give proof?}
Note that, by the Knaster-Tarski fixed point theorem this yields:
\begin{proposition}
[``Fixed points'' are fixed points]
\label{prop:hr-mu-is-lfp}
    $\hr{\mu X \sigma(X)}$ is the least fixed point of $\hr{\sigma(\cdot)}$ on candidates.
\end{proposition}

We will need to appeal to an alternative characterisation of fixed points via an inflationary construction, yielding a notion of `approximant' that:
\begin{enumerate}[(a)]
    \item allows us to prove the Interpretation Theorem by reduction to well-foundedness of approximants (or, rather, the ordinals that index them); and 
    \item allows a logically simpler formalisation within second-order arithmetic, cf.~Section~\ref{sec:rev-math-kt}, crucial for obtaining a tight bound on representable functions,
\end{enumerate}

\begin{definition}
    [Approximants]
    Let $F$ be a monotone operation on candidates, with respect to $\subseteq$.
    For ordinals $\alpha$ we define $F^\alpha(A)$ by transfinite induction on $\alpha$:
    \begin{itemize}
        \item $F^0 (A) \dfn \emptyset$
        \item $F^{\succ \alpha}(A) \dfn F(F^\alpha(A))$
        \item $F^{\lambda}(A) \dfn \bigcup\limits_{\alpha<\lambda} F^\alpha(A)$, when $\lambda$ is a limit ordinal.
    \end{itemize}
\end{definition}
For our purposes we will only need the special case of the definition above when $A=\emptyset$.
Writing $\Ord$ for the class of all ordinals, the following is well known:
\begin{proposition}
    [Fixed points via approximants]
    \label{prop:fixed-points-via-approximants}
Let $\mu F$ be the least fixed point of a monotone operation $F$. 
    We have
    $\mu F = \bigcup\limits_{\alpha \in \Ord} F^\alpha(\emptyset)$.
\end{proposition}
    
From here it is convenient to admit formal type expressions representing approximants.
\begin{convention}
    [Transfinite types]
    We henceforth expand the language of types to be closed under:
\begin{itemize}
    \item for $\sigma(X)$ positive in $X$, $\alpha$ an ordinal, $\sigma^\alpha(\tau)$ is a type.
\end{itemize}
Again we shall only need the special case of $\tau = \emptyset$ for our purposes.
We call such types \emph{transfinite} when we need to distinguish them from ordinal-free types.
\end{convention}

\begin{definition}
    [Type structure, continued]
    We expand Definition~\ref{def:hr} to account for transfinite types by setting $\hr{\sigma^\alpha(\tau)} \dfn \hr{\sigma(\cdot)}^\alpha (\hr \tau)$. 
    In other words:
\begin{itemize}
    \item $\hr{\sigma^0(\tau)} := \hr\tau$
    \item $\hr{\sigma^{\succ \alpha}(\tau)}:= \hr{\sigma(\hr{\sigma^\alpha(\tau)})}$
    \item $\hr{\sigma^\lambda (\tau)}:= \bigcup\limits_{\alpha<\lambda} \hr{\sigma^\alpha(\tau)}$, when $\lambda$ is a limit ordinal.
\end{itemize}
\end{definition}

We have immediately from Proposition~\ref{prop:fixed-points-via-approximants}:
\begin{corollary}
\label{cor:hr-mu-as-limit-of-approximants}
$\hr{\mu X \sigma(X)} = \bigcup\limits_{\alpha \in \Ord} \hr{\sigma^\alpha(\emptyset)}$.
\end{corollary}

\subsection{Ordinal assignments}

We shall write $\sigma \subform \tau$ if $\sigma$ is a subformula of $\tau$.

The \emph{Fischer-Ladner (well) preorder}, written $\flleq$,
is the smallest extension of $\subform$, restricted to closed formulas, satisfying
$\sigma(\mu X \sigma(X)) \flleq \mu X \sigma(X)$.
We write $\sigma \fleq \tau$ if $\sigma \flleq \tau \flleq \sigma$ and $\sigma \flnleq \tau$ if $\sigma \flleq \tau \not\flleq \sigma$.
Note that $\fleq$-equivalence classes are naturally (well) partially ordered by $\flleq$.

The \emph{Fischer-Ladner closure} of a type $\tau$, written $\fl \tau$, is the set $\{\sigma \mid \sigma \flleq \tau\}$.
Note that $\fl \tau$ is the smallest set of closed types closed under subformulas and, whenever $\mu X\sigma(X) \in \fl \tau$, then also $\sigma(\mu X\sigma(X)) \in \fl \tau$; this is necessarily a finite set.

\begin{definition}
[Priority]
We say that a type $\sigma$ has higher \emph{priority} than a type $\tau$, written $\sigma>\tau$, if $\tau \flnleq \sigma$ or $\sigma \fleq \tau $ and $\sigma \subform \tau$.

\end{definition}
Note that $<$ is a (strict) well partial order on types. 
The priority order is commonly used in modal fixed point logics, e.g., \cite{Studer08:mu-calc,Doumane17thesis,venema2008lectures}.

\begin{convention}
In what follows, we shall assume an arbitrary extension of $<$ to a \emph{total} well order.    
\end{convention}

Let $\tau$ be a type whose $<$-greatest fixed point subformula occurring in \emph{positive} position is $\mu X \sigma(X)$.
We write $\tau^\alpha$ for $\tau[\sigma^\alpha(\emptyset) / \mu X\sigma(X)]$, i.e.\ $\tau$ with each occurrence of $\mu X\sigma(X)$ in positive position replaced by $\sigma^\alpha(\emptyset)$.
$\tau_\alpha$ is defined the same way by for the $<$-greatest fixed point subformula occurring in \emph{negative} position.

Note that, if $\tau$ has $n$ fixed point subformulas in positive position and $\vec \alpha = \alpha_1 ,\dots, \alpha_n$ then $\tau^{\vec \alpha} = (\cdots ((\tau^{\alpha_1})^{\alpha_2})\cdots)$ is a positive-$\mu$-free (transfinite) type.
Similarly for negatively occurring fixed points.
We shall call such sequences \emph{(positive or negative) assignments} (respectively).

We shall order assignments by a lexicographical product order, i.e.\ by setting $\vec \alpha < \vec \beta$ when there is some $i$ with $\alpha_i < \beta_i$ but $\alpha_j = \beta_j$ for all $j<i$. 
Note that this renders the order type of $\vec \alpha$ simply $\alpha_n \times \cdots \times \alpha_1$, but it will be easier to explicitly work with the lexicographical order on ordinal sequences.

By the Monotonicity Lemma~\ref{lem:mon-pos-types-in-hr} we have:
\begin{proposition}
[Positive and negative approximants]
\label{prop:pos-neg-approximants}
If $t \in \hr\tau$ there are (least) ordinal(s) $\vec \alpha$ s.t.\ $t \in \hr{\tau^{\vec \alpha}}$.
Dually, if $t \notin \hr \tau$ there are (least) ordinal(s) $\vec \alpha$ s.t.\ $t \notin \hr{\tau_{\vec \alpha}}$.
\end{proposition}

\anupam{commented `approximants are successors' below, now obsolete}
Since the ordinals given by the above Proposition are points of first entry for an element into a fixed point,
note that, for $\vec \alpha$ as in the Proposition above, each $\alpha_i$ must be a successor ordinal.



\subsection{Reflecting non-totality in rules of $\muLJnegnorec$}


\anupam{commented some narrative above}

Before giving our non-total branch construction, let us first make a local definition that will facilitate our construction:
\begin{definition}
    [Reflecting non-totality]
    \label{def:refl-non-totality}
    Fix a $\muLJnegnorec$-step,
    \[
    \vliiinf{ \rrule}{}{\Sigma \seqar \tau}{\Sigma_0 \seqar \tau_0}{\cdots}{\Sigma_{n-1} \seqar \tau_{n-1}}
    \] 
    for some $n\leq 2$ and regular coderivations $\der_i:\Sigma_i\seqar \tau_i$. 
    
    For $\vec s \in \hr{\Sigma}$ s.t.\ $ \rrule {\vec \der} \vec s \notin \hr \tau$,
    we define a premiss $\Sigma' \seqar \tau'$ and $\der':\Sigma'\seqar \tau'$ and some inputs $\vec s' \in \hr{ \Sigma'}$ such that $\der'\vec s'
    \notin \hr{\tau'}$
    as follows:
\begin{itemize} 
\item $\rrule $ cannot be $\id$ since $\id s \convertsnoreceta s$.
\item If $\rrule $ is $\vlinf{\ex}{}{\Gamma, \rho, \sigma, \Delta \seqar \tau}{\Gamma , \sigma, \rho , \Delta \seqar \tau}$ and $\vec s = (\vec r, r , t,\vec t)$ with $\vec r \in \hr \Gamma, r\in \hr \rho , s \in \hr \sigma, \vec t \in \hr \Delta$, we set $(\Sigma'\seqar \tau'):= (\Gamma, \sigma, \rho, \Delta \seqar \tau)$, $\der':=\der_0$ and $\vec s':= (\vec r,t,r,\vec t)$.
\item If $\rrule $ is $\vlinf{\wk}{}{\Sigma_0,\sigma \seqar \tau}{\Sigma_0 \seqar \tau}$ and $\vec s = (\vec s_0,s) $ with $\vec s_0 \in \hr{\Sigma_0}$ and $s \in \hr \sigma$ then we set $\der':= \der_0$, $(\Sigma'\seqar \tau'):= (\Sigma_0 \seqar \tau)$ and $\vec s' := \vec s_0$. 
\item If $\rrule $ is $\vlinf{\contr}{}{\Sigma_0 , \sigma\seqar \tau}{\Sigma_0, \sigma, \sigma \seqar \tau}$ and $\vec s = (\vec s_0,s)$ with $\vec s_0 \in \hr {\Sigma_0}, s\in \hr \sigma$, then we set $(\Sigma' \seqar \tau'):= \allowbreak (\Sigma_0\seqar \tau)$, $\der':= \der_0$ and $\vec s' := (\vec s_0,s,s)$.
\item If $\rrule $ is 
\(
\vliinf{\cut}{}{\Sigma_0,\Sigma_1 \seqar \tau}{\Sigma_0 \seqar \tau_0}{\Sigma_1, \tau_0 \seqar \tau }
\)
and $\vec s_0 \in \hr{\Sigma_0}, \vec s_1 \in \hr{\Sigma_1}$ we have:
\[
\begin{array}{rll}
     & \cut \der_0 \der_1 \vec s_0 \vec s_1\, \notin \hr\tau \\
     \implies & \der_1  \vec s_1 (\der_0 \vec s_0) \, \notin \hr \tau & \text{by $\cut$ reduction}
\end{array}
\]
Now, if $\der_0\, \vec s_0 \ \notin \hr{\tau_0}$ then we set $(\Sigma'\seqar \tau') := (\Sigma_0 \seqar \tau_0)$, $\der' := \der_0$ and $\vec s' := \vec s_0$.
Otherwise $\der_0\, \vec s_0 \ \in \hr{\tau_0}$ so we set 
$(\Sigma'\seqar \tau') := (\Sigma_1 \seqar \tau_1)$, $\der' := \der_1$ and $\vec s' := (\vec s_1, \der_0\, \vec s_0)$.
\item $\rrule$ cannot be 
\(
\vlinf{\Nzero}{}{\seqar \nat}{}
\)
as $\numeral 0 \in \hr \nat$.
\item If $\rrule$ is
\(
\vlinf{\Nsucc}{}{\Sigma \seqar \nat}{\Sigma \seqar \nat}
\)
and $\vec s \in \hr{\Sigma}$,
\[
\begin{array}{rll}
     & \Nsucc \der_0  \vec s \, \notin \hr\nat \\
     \implies & \succ (\der_0 \vec s) \, \notin \hr \nat & \text{by $\Nsucc $ reduction} \\
     \implies & \der_0 \vec s \, \notin \hr \nat  & \text{by context closure of $\convertsnoreceta$}
\end{array}
\]
so we set $(\Sigma'\seqar \tau'):= (\Sigma \seqar \nat)$, $\der' := \der_0$ and $\vec s' := \vec s$.
\item If $\rrule$ is
\(
\vliinf{\Ncnd}{}{\Sigma_0,\nat \seqar \tau}{\Sigma_0 \seqar \tau}{\Sigma_0,\nat \seqar \tau}
\)
and $\vec s_0 \in \hr{\Sigma_0}$, $s \in \hr \nat$, then $s\convertsnoreceta \numeral n$ for some $n\in \Nat$ by definition of $\hr\nat$.
If $n=0$ then,
\[
\begin{array}{rll}
     &  \Ncnd \der_0 \der_1 \vec s s \, \notin \hr \tau \\
     \implies & \Ncnd \der_0\der_1 \vec s \numeral 0 \, \notin \hr \tau & \text{by closure under conversion} \\
     \implies & \der_0 \vec s \, \notin \hr \tau & \text{by $\Ncnd$ reduction}
\end{array}
\]
so we set $(\Sigma'\seqar \tau'):= (\Sigma_0 \seqar \tau)$, $\der':= \der_0$, $\vec s' := \vec s$.

Otherwise if $n= n'+1$ then,
\[
\begin{array}{rll}
     & \Ncnd \der_0 \der_1 \vec s s \, \notin \hr \tau \\
     \implies & \Ncnd \der_0  \der_1 \vec s  (\succ \numeral n') \, \notin \hr \tau & \text{by closure under conversion} \\
     \implies & \der_1  \vec s \numeral n' \, \notin \hr \tau & \text{by $\Ncnd$ reduction}
\end{array}
\]
so we set $(\Sigma'\seqar \tau'):= (\Sigma_0, \nat \seqar \tau)$, $\der':= \der_1$, $\vec s' := (\vec s,\numeral n')$.
        \item If $\rrule$ is
        \(
        \vliinf{\times r }{}{\Sigma_0, \Sigma_1 \seqar \tau_0 \times \tau_1}{\Sigma_0 \seqar \tau_0}{\Sigma_1 \seqar \tau_1}
        \)
        and $\vec s_0 \in \hr{\Sigma_0}, \vec s_1 \in \hr{\Sigma_1}$,
        \[
        \begin{array}{rll}
        & {\rr \times} {\der_0}{\der_1}\vec s_0 \vec s_1 \notin \hr{\tau_0\times \tau_1} & 
        \\ 
        \implies & \pair{\der_0\vec s_0}{\der_1\vec s_1} \notin\hr{\tau_0\times \tau_1} & \text{by $\rr\times$ reduction} \\
        \implies & \proj i \pair{\der_0\vec s_0}{\der_1\vec s_1}  \notin \hr{\tau_i} & \text{by $\hr{\times}$} \\
        \implies & {\der_i}\vec s_i \notin \hr{\tau_i} & \text{by {$\proj i \pair{\cdot}{\cdot}$ reduction}}
        \end{array}
        \]
        for some $i\in \{0,1\}$,
        so we set $(\Sigma' \seqar \tau') := (\Sigma_i \seqar \tau_i)$, $\der' := \der_i$, $\vec s' := \vec s_i$.
    \item If $\rrule$ is $\vlinf{  {\lr \times}}{}{\Gamma, \sigma_0 \times \sigma_1 \seqar \tau  }{\Gamma, \sigma_0,\sigma_1 \seqar \tau}$ and $\vec r \in \hr{\Gamma}$ and $s \in \hr{\sigma_0 \times \sigma_1}$:
    \[
    \begin{array}{rll}
        & {\lr\times }\,  {\der_0}\,  \vec r \, s\ \notin\  \hr\tau & 
        \\
        \implies & {\der_0}\, \vec r \, \proj 0 s\, \proj 1 s \ \notin\  \hr\tau & \text{by $\lr \times$ reduction}
    \end{array}
    \]
    Now, by definition of $\hr{\sigma_0 \times \sigma_1}$ we have indeed $\proj i s \in \hr{\sigma_i}$ for $i\in \{0,1\}$,
    so we set $(\Sigma' \seqar \tau') = (\Gamma,\sigma_0,\sigma_1 \seqar \tau)$, $\der' = \der_0$ and $\vec s' = (\vec r, \proj i s)$.
\item If $\rrule$ is
    \(
    \vlinf{\rr\arrow}{}{\Sigma \seqar \rho \to \sigma }{\Sigma, \rho \seqar \sigma}
    \)
    and $\vec s \in \hr{\Sigma}$,
    \[
    \begin{array}{rll}
         & {\rr\arrow } {\der_0}\vec s\  \notin\  \hr{\rho\to \sigma} & \\
         \implies & {\rr \arrow} {\der_0} \vec s \  \notin\  \hr{(\rho \to \sigma)_{\vec \alpha}} & \text{for some least $\vec \alpha$}\\
         \implies & \rr \arrow  \der_0  \vec s\  \notin \ \hr{\rho^{\vec \alpha_0} \to \sigma_{\vec \alpha_1}} & \text{by definition} \\
         \implies & \rr \arrow{\der_0}  \vec s s \ \notin\  \hr{ \sigma_{\vec \alpha_1}} & \text{for some $s \in \hr{\rho^{\vec \alpha_0}}$} \\
         \implies & \der_0 \vec s  s\  \notin \ \hr{ \sigma_{\vec \alpha_1}} & \text{by $\rr\arrow$ reduction} \\
         \implies & \der_0 \vec s s \notin \hr{ \sigma} & \text{by monotonicity}
    \end{array}
    \]
    where $\vec \alpha_0$ and $\vec \alpha_1$ are appropriate subsequences of $\vec \alpha$ in case not all fixed point subformulas of $\rho \to \sigma$ occur in $\rho$ or in $\sigma$.
    So we set $(\Sigma' \seqar \tau') := (\Sigma, \rho \seqar \sigma)$, and $\der' \dfn \der_0$ and $\vec s' := (\vec s,s)$.

    \item If $\rrule$ is
    \(
    \vliinf{\lr\arrow }{}{\Gamma,\Delta, \rho \to \sigma \seqar \tau}{\Gamma \seqar \rho}{\Delta , \sigma \seqar \tau}
    \)
    let $\vec r \in \hr{\Gamma}, \vec t\in \hr{\Delta}$ and $s \in \hr{\rho \to \sigma}$.
    Like before, let
    $\vec \alpha$ be the least assignment such that $s \in \hr{(\rho \to \sigma)^{\vec \alpha}}=\hr{\rho_{\vec \alpha_0}\to \sigma^{\vec \alpha_1}}$.
    We have two cases:
    \begin{itemize}
        \item if $\der_0\vec r \notin \hr{\rho_{\vec \alpha_0}}$ then we set $(\Sigma' \seqar \tau') \dfn (\Gamma \seqar \rho) $, and $\der'\dfn \der_0$ and $\vec s' \dfn \vec r$.
        \item otherwise ${\der_0} \vec r \in \hr{\rho_{\vec \alpha_0}}$ and we have:
        \[
        \begin{array}{rll}
             & {\lr \arrow}{\der_0}{\der_1} \vec r \vec t s \notin \hr{\tau} & \\
             \implies & {\der_1} \vec t (s ({\der_0} \vec r)) \notin \hr\tau & \text{by $\lr\arrow$ reduction}
        \end{array}
        \]
        Since $\der_0 \vec r \in \hr{\rho_{\vec \alpha_0}}$ and $s \in \hr{\rho_{\vec \alpha_0} \arrow \sigma^{\vec \alpha_1}}$ 
        we have $s( \der_0 \vec r) \in \hr{\sigma^{\vec\alpha_1}}$ so we set
        $(\Sigma' \seqar \tau') \dfn (\Delta, \sigma \seqar \tau)$, and $\der' \dfn \der_1$ and $\vec s' \dfn (\vec t, s(\der_0 \vec r) )$.
    \end{itemize}
    \item 
    If $\rrule$ is 
    \(
    \vlinf{\rr \mu}{}{\Sigma \seqar \mu X \sigma(X)}{\Sigma \seqar \sigma(\mu X \sigma(X))}
    \)
    and $\vec s \in \hr{\Sigma}$,
    \[
    \begin{array}{rll}
         &  \rr\mu\der_0 \vec s\  \notin \ \hr{\mu X\sigma(X)} \\
         \implies & \der_0 \vec s \notin \hr{\mu X\sigma(X)} & \text{by $\rr\mu$ reduction } \\
         \implies & \der_0 \vec s\  \notin\ \hr{\sigma(\mu X \sigma(X))} & \text{by Proposition~\ref{prop:hr-mu-is-lfp}}
    \end{array}
    \]
    so we set $(\Sigma' \seqar \tau') \dfn (\Sigma \seqar \sigma(\mu X \sigma (X)))$, $\der' \dfn \der_0$ and $\vec s' \dfn \vec s$.
    \item if $\rrule $ is 
    \(
    \vlinf{\mulunf}{}{\Gamma, \mu X \sigma(X) \seqar \tau}{\Gamma, \sigma(\mu X \sigma(X)) \seqar \tau}
    \)
    and $\vec r \in \hr{\Gamma}$ and $s \in \hr{\mu X \sigma(X)}$ we have:
    \[
    \begin{array}{rll}
         &  \lr \mu\der_0 \vec r s \ \notin \ \hr\tau \\
        \implies & \der_0 \vec r s \ \notin\  \hr \tau  & \text{by $\lr \mu$ reduction}
    \end{array}
    \]
    Since $s \in \hr{\mu X \sigma(X)}$ then also $s \in \hr{\sigma(\mu X \sigma(X))}$ by Proposition~\ref{prop:hr-mu-is-lfp},
    so we set $(\Sigma' \seqar \tau') \dfn (\Gamma , \sigma(\mu X \sigma(X)) \seqar \tau)$, and $ \der' \dfn \der_0$ and $\vec s' \dfn (\vec r , s)$.
    \end{itemize}
\end{definition}

   Note that the $\arrow$ cases for the Definition above are particularly subtle, requiring consideration of least ordinal assignments similar to the handling of $\lor$-left in fixed point modal logics, cf.\ e.g.\ \cite{NiwWal96:games-for-mu-calc}.

\anupam{make a remark about extension to $+$ cases and first component}

\subsection{Non-total branch construction}
From here the proof of the Interpretation Theorem \ref{thm:prog-implies-hr} proceeds by contradiction, as is usual in cyclic proof theory.
The definition above is used to construct an infinite `non-total' branch, along which there must be a progressing thread.
We assign ordinals approximating the critical fixed point formula and positive formulas of higher priority, which must always be present as positive subformulas on the LHS, or negative on the RHS, along the thread.
Tracking the definition above, we note that this ordinal assignment sequence is non-increasing; moreover at any $\mulunf$ step on the critical fixed point, the corresponding ordinal must be a successor and strictly decreases. 
Thus the ordinal sequence does not converge, yielding a contradiction.


\anupam{commented remark about fixed points approximants below, pertinent to old proof, add similar remark back if time.}

\begin{proof}
[Proof of \Cref{thm:prog-implies-hr}]
    Let $\der : \Sigma\seqar \tau$ and $\vec s \in \hr\Sigma$ be as in the statement, and suppose for contradiction that $\der\vec s \notin \hr{ \tau}.$
Setting $\der_0 \dfn \der$ and $\vec s_0 \dfn \vec s$, we use Definition~\ref{def:refl-non-totality} to construct an infinite branch $(\der_i : \Sigma_i \seqar \tau_i)_{i < \omega}$ and associated inputs $(\vec s_i)_{i < \omega}$ by always setting $\der_{i+1} \dfn \der_i'$, $(\Sigma_{i+1}\seqar \tau_{i+1}) \dfn (\Sigma_i'\seqar \tau_i')$ and $\vec s_{i+1}\dfn \vec s_i'$.

Now let $(\rho_i)_{i\geq k}$ be a progressing thread along $(\der_i)_{i <\omega}$, since $\der$ is progressing, and let $r_i \in \vec s_i$ be the corresponding input, for $i\geq k$, when $\rho_i$ is on the LHS.
%
%
Since $(\rho_i)_{i\geq k}$ is progressing, let $\mu X_1\sigma_1(X_1) > \cdots > \mu X_n \sigma_n (X_n)$ enumerate the fixed points occurring in every $\rho_i$ positively in the LHS, and negatively in the RHS, such that $\mu X_n\sigma_n$ is the smallest infinitely often principal formula along $(\rho_i)_{i\geq k}$. 
Note that such a finite set of fixed points must exist since 
necessarily $\rho_{i+1}\flleq \rho_{i}$, and so the thread must eventually stabilise within some Fischer-Ladner class. 
WLoG, no $\mu X_j \sigma_j (X_j)$, for $j<n$, is principal along $(\rho_i)_i$.

By Proposition~\ref{prop:pos-neg-approximants}, let $\vec \alpha_i = \alpha_{i1}, \dots, \alpha_{in}$, for $i \geq k$, be the least assignments such that:
\begin{itemize}
    \item if $\rho_i$ is on the LHS then $r_i \in \hr{\rho_i^{\vec \alpha_i}}$; 
    \item if $\rho_i$ is on the RHS then $\der_i\vec s_i \notin (\rho_i)_{\vec \alpha_i}$. 
\end{itemize}

We claim that $(\vec \alpha_i)_{i\geq k}$ is a monotone non-increasing sequence of ordinal assignments that does not converge:
\begin{itemize}
    \item By construction of $\der_i$ and $\vec s_i$, appealing to Definition~\ref{def:refl-non-totality}, note that for each step for which $\mu X_n\sigma_n(X_n)$ is not principal (along $(\rho_i)_i$ on the LHS) we have that $\vec \alpha_{i+1}\leq \vec \alpha_i$.
    Note in particular that the $\arrow$-cases of Definition~\ref{def:refl-non-totality} are designed to guarantee $\vec \alpha_{i+1}\leq \vec \alpha_i$ at $\arrow$-steps.
    \item Now, consider a $\mulunf$-step along the progressing thread on the critical fixed point formula,
    \[
    \vlinf{\mulunf}{}{\Gamma, \mu X_n \sigma_n(X_n)\seqar \pi}{\Gamma , \sigma_n(\mu X_n \sigma_n(X_n))\seqar \pi}
    \]
    with $\rho_i = \mu X_n \sigma_n (X_n)$.
    Writing $\vec \alpha_{in} = \alpha_{i,1}, \dots, \alpha_{i,n-1}$ (a prefix of $\vec \alpha_i$) we have:
    \[
    \begin{array}{rl}
        \hr{(\mu X_n \sigma_n (X_n))^{\vec \alpha_{in}\alpha_{in}}} 
        = & \hr{(\mu X_n \sigma_n^{\vec \alpha_{in}} (X_n))^{\alpha_{in}}} \\
        = & \hr{\sigma_n^{\vec \alpha_{in} \alpha_{in}}(\emptyset)} \\
        = & \hr{\sigma_n^{\vec \alpha_{in} \succ \alpha_{in}'}(\emptyset)} \\
        = & \hr{\sigma_n^{\vec \alpha_{in}} (\hr{\sigma_n^{\vec \alpha_{in} \alpha_{in}' } (\emptyset) } ) } \\
         = & \hr{\sigma_n^{\vec \alpha_{in}} (\hr{(\mu X_n \sigma_n^{\vec \alpha_{in}} (X_n))^{\alpha_{in}'}} ) } \\
           = & \hr{\sigma_n^{\vec \alpha_{in}} ( \hr{(\mu X_n \sigma_n(X_n))^{\vec \alpha_{in}\alpha_{in}'}}) } \\
    \end{array}
    \]
    for some $\alpha_{in}'$ since $\alpha_{ij}$ must be a successor ordinal.
    Thus indeed $\vec \alpha_{i+1} < \vec \alpha_{i}$.
\end{itemize}
%
%
%
    
\noindent
This contradicts the well-foundedness of ordinals.
\end{proof}
 \section{Some reverse mathematics of Knaster-Tarski}
\label{sec:rev-math-kt}
In order to obtain sub-recursive upper bounds on the functions represented by $\cmuLJ$, we will need to formalise the totality argument of the previous section itself within fragments of `second-order' arithmetic.
To this end we will need to formalise some of the reverse mathematics about fixed point theorems on which our type structure relies.

\subsection{Language  and theories of `second-order' arithmetic}
Let us recall
$\langsoarith$, the language of `second-order' arithmetic, e.g.\ as given in \cite{Simpson99:monograph}.
It extends the language of arithmetic $\langarith$ by:
\begin{itemize}
    \item an additional sort of \emph{sets}, whose variables are written $X,Y $ etc. Individuals of $\langarith$ are considered of \emph{number} sort.
    \item an \emph{elementhood} (or \emph{application}) relation $\in$ relating the number sort to the set sort. I.e.\ there are formulas $t \in X$ (also $Xt$) when $t $ is a number term and $X$ is a set variable. (We will not consider any non-variable set terms here.)
\end{itemize}

When speaking about the `free variables' of a formula, we always include set variables as well as individual variables.
We shall assume a De Morgan basis of connectives, namely $\lor, \land, \exists, \forall$ with negation only on atomic formulas.
Hence, we say that a formula $\phi$ is \emph{positive} (or \emph{negative}) in $X$ if no (every, resp.) subformula $Xt$ occurs under a negation.
Let us note that we have an analogue of `functoriality' in predicate logic:
\begin{lemma}[Monotonicity]
\label{lem:pos-imp-mon}
Let $\phi(X,x)$ be positive in $X$. 
\[
\proves \forall x (Xx \limp Yy) \limp \forall x (\phi(X,x) \limp \phi(Y,x))
\]
\end{lemma}
\begin{proof}[Proof sketch]
By (meta-level) induction on the structure of $\phi$.
\end{proof}
In what follows this will often facilitate arguments by allowing a form of `deep inference' reasoning.

We shall work with subtheories of full second-order arithmetic ($\SOPA$) such as $ \ACA $, $  \PCA$, $  \PSCA$ etc., whose definitions may be found in standard textbooks, e.g.\ \cite{Simpson99:monograph}.
We shall freely use basic facts about them.
For instance:

\begin{proposition}
    [Some basic reverse mathematics, e.g.\ \cite{Simpson99:monograph}]
    \label{prop:basic-rev-math}
    We have the following:
    \begin{enumerate}
        \item\label{item:inclusions-basic-rev-math} $ \ACA \subseteq \PCA = \SCA \subseteq \PSCA = \SPCA $
        \item\label{item:pca-proves-atr-basic-rev-math} $\PCA \proves \ATR$ (\emph{arithmetical transfinite recursion}).
        \item\label{item:psca-proves-spac-basic-rev-math} $\ac{\Din 1 2}_0 \proves \ac{\Sin 1 2}$  (\emph{axiom of choice})
    \end{enumerate}
\end{proposition}

A simple consequence of $\Sin 1 2$-choice in $\PSCA$ is that $\Sin 1 2$ (and $\Pin 1 2 $) is provably closed under positive $\Sin 1 1 $ (resp.\ $\Pin 1 1 $) combinations (even with $\Sin 1 1 $ and $\Pin 1 1 $ parameters).
We shall actually need a refinement of this fact to take account of polarity, so we better state it here.
First let us set up some notation.

\begin{definition}
    [Polarised analytical hierarchy]
    We write $\posnegSin 1 n {\vec X}{\vec Y}$ for the class of $\Sin 1 n $ formulas positive in $\vec X$ and negative in $\vec Y$.
Similarly for $\posnegPin 1 n {\vec X}{\vec Y}$.
$\phi$ is $\posnegDin 1 n {\vec X}{\vec Y}$, in a theory $T$, if it is $T$-provably equivalent to both a $\posnegSin 1 n {\vec X}{\vec Y} $-formula and a $\posnegPin 1 n{\vec X}{\vec Y}$-formula.
\end{definition}

\begin{lemma}
[Polarised substitution lemma, $\PSCA$]
\label{lem:pol-sub}
If $\phi(X) \in \posnegSin 1 1 {X,\vec X}{\vec Y}$ and $\psi\in \posnegSin 1 2 {\vec X}{\vec Y}$ then $\phi(\psi) \in \posnegSin 1 2 {\vec X}{\vec Y}$.
Similarly for $\Pi$ in place of $\Sigma$.
\end{lemma}
\begin{proof}
    [Proof sketch]
    By induction on the structure of $\phi$, using $\ac{\Sin 1 2}$ at each alternation of a FO and SO quantifier.
\end{proof}

Note that the Lemma above, in particular, allows for arbitrary substitutions of $\Sin 1 1 $ and $\Pin 1 1 $ formulas free of $\vec X,\vec Y$, which we shall rely on implicitly in the sequel.


\subsection{Countable orders}
We can develop a basic theory of (countable) ordinals in even weak second-order theories, as has been done in \cite{Simpson99:monograph} and also comprehensively surveyed in \cite{hirst05:ord-arith-rm-survey}.
Let us point out that, while distinctions between natural notions of `order comparison' are pertinent for weak theories, the theories we mainly consider contain $\ATR_0$, for which order comparison is robust.

A \emph{(countable) binary relation} is a pair $ (X,\leq)$ where $X$ and $\leq$ are sets, the latter construed as ranging over pairs. 
(Sometimes we write $\Rel(X,\leq)$ to specify this.)
We say that $(X,\leq)$ is a \emph{partial order}, written $\PO(X,\leq)$, if:
\begin{itemize}
    \item $\forall x \in X \, x\leq x$
    \item $\forall x,y \in X (x\leq y \leq x \limp x=x)$
    \item $\forall x,y,z (x\leq y\leq z \limp x\leq z)$
\end{itemize}

$(X,\leq)$ is a \emph{total order}, written $\TO(X,\leq)$, if it is a partial order that is \emph{total}:
\begin{itemize}
    \item $\forall x,y (x\leq y \lor y\leq x)$
\end{itemize}

Given a relation $\leq$, we may write $<$ for its \emph{strict} version, given by
\[
x<y := x\leq y \land \lnot y\leq x
\]
We employ similar notational conventions for other similar order-theoretic binary symbols.

We say that a binary relation $(X,\leq)$ is \emph{well-founded}, written $\WF(X,\leq)$, if:
\begin{itemize}
    \item $\forall f: \Nat \to X \exists x \lnot f(x+1)<f(x)$
\end{itemize}


$(X,\leq)$ is a \emph{well-order}, written $\WO(X,\leq)$, if it is a well-founded total order, i.e.:
\[
\WO(X,\leq) := \TO(X,\leq) \land \WF(X,\leq)
\]

Henceforth we shall write $\alpha,\beta$ etc.\ to range over countable binary relations.
If $\alpha = (X,\leq)$ we may write $x\leq_\alpha y := x,y \in X \land x\leq y$, and similarly $x<_\alpha y$ for $x,y\in X \land x<y$.
We may also write simply $x \in \alpha$ instead of $x \in X$, as abuse of notation.

It is not hard to see that $\ACA$ admits an induction principle over any provable well-order (see, e.g., \cite{Simpson99:monograph}):
\begin{fact}
[Transfinite induction, $\ACA$]
\label{fact:wo-ind-in-aca}
If $\WO(\alpha)$ then:
\[
\forall X (\forall x\in \alpha (\forall y<_\alpha x\,  X(y) \limp X(x)) \limp \forall x \in \alpha\,  X(x))
\]
\end{fact}
In fact, in extensions of $\ACA$, we even have (transfinite) induction on arithmetical formulas, a fact that we shall use implicitly in what follows.
For instance, in $\PSCA$, we may admit (transfinite) induction on arithmetical combinations of $\Pin 1 2 $ formulas.
\anupam{check this again}

\subsection{Comparing orders}
Following Simpson in \cite{Simpson99:monograph},
given $\alpha,\beta \in \WO$ we write $\alpha \ordlneq \beta$ if there is an order-isomorphism from $\alpha$ onto a proper initial segment of $\beta $. We also write $\alpha \ordeq \beta$ if $\alpha $ and $\beta $ are order-isomorphic, and $\alpha \ordleq \beta $ if $\alpha \ordlneq \beta \lor \alpha \ordeq \beta$.
Thanks to the uniqueness of comparison maps, we crucially have

\begin{proposition}
[$\ATR_0$]
    $\ordlneq, \ordleq, \ordeq$ are $\Din 1 1 $.
\end{proposition}


\anupam{commented remark above, it seems obsolete and incorrect now (since we follow Simpson's `strong comparison').}


We shall now state a number of well-known facts about comparison, all of which may be found in, e.g., \cite{Simpson99:monograph} or \cite{hirst05:ord-arith-rm-survey}.\footnote{Note that we have intentionally refrained from specifying the `optimal' theories for each statement, for simplicity of exposition}

\begin{proposition}
    [Facts about ordinal comparison, $\PCA$]
    \label{prop:ord-comparison-facts}
    Let $\alpha,\beta,\gamma$ be well-orders.
    We have the following:
    \begin{enumerate}
        \item\label{item:comparison-is-preorder} (Comparison is a preorder)
\begin{enumerate}
    \item $\alpha\ordleq \alpha$
    \item $\alpha \ordleq \beta \ordleq \gamma \limp \alpha \ordleq \gamma$
\end{enumerate}
\item\label{item:comparison-pseudo-antisymmetric} (Comparison is pseudo-antisymmetric) If $\alpha \ordleq \beta \ordleq \alpha $ then $\alpha \ordeq \beta$.
\item\label{item:comparison-total} (Comparison is total) 
$\alpha \ordleq \beta \lor \beta \ordleq \alpha$
\item\label{item:comparison-wf} (Comparison is well-founded)
$\forall F: \Nat\to \WO \,  \exists x\  \lnot F(x+1)\ordlneq F(x)$
    \end{enumerate}
\end{proposition}


\anupam{commented ordinal arithmetic proposition below.}


We shall assume basic ordinal existence principles, in particular constructions for successor ($\succ \alpha$), addition ($\alpha + \beta$) and maximum ($\max(\alpha,\beta))$, initial segments ($\alpha_b$ for $b\in \alpha$), all definable and satisfying characteristic properties provably in $\ATR_0$.
We shall also make crucial use of a `bounding' principle:
\begin{proposition}
    [$\Sin 1 1 $-Bounding, $\ATR_0$]
    \label{prop:bounding}
    Let $\phi(\alpha) \in \Sin 1 1 $ with $\forall \alpha  (\phi(\alpha)\limp \WO(\alpha))$. 
    Then $\exists \beta \in \WO \forall \alpha (\phi(\alpha) \limp \alpha \ordleq \beta)$.
\end{proposition}
When appealing to Bounding, we use notation such as $\upbnd \phi$ for an ordinal bounding all $\alpha \in \WO$ such that $\phi(\alpha)$.

\anupam{commented positive combinations of operators section below. Now obsolete.}

\subsection{Knaster-Tarski theorem and approximants}
\label{sec:kt-thm-and-approximants-in-psca}
\anupam{commented previous setup based on operators}



Throughout this section let $\phi(X,x) \in \posnegDin 1 2 {X,\vec X}{\vec Y}$.
We shall typically ignore/suppress the parameters $\vec X, \vec Y$, focussing primarily on $X$ and $x$, and sometimes even write $\phi$ instead of $\phi(X,x)$.
We shall work within $\PSCA$ unless otherwise stated.


\begin{remark}
    [Reverse mathematics of fixed point theorems]
    Let us point out that, while previous work on the reverse mathematics of fixed point theorems exist in the literature, e.g.\ \cite{PengYamazaki17:rev-math-fixed-point-thms-no-kt}, even for the Knaster-Tarski theorem \cite{SatoYamazaki17:rev-math-fixed-point-thms-with-kt}, these results apply to situations when the lattice or space at hand is countable (a subset of $\Nat$).
Here we require a version of the theorem where sets themselves are elements of the lattice, under inclusion.
Our operators are specified \emph{non-uniformly}, namely via positive formulas.
This amounts to a sort of non-uniform `3rd order reverse mathematics' of Knaster-Tarski, peculiar to the powerset lattice on $\Nat$.
\end{remark}

While we can define (bounded) approximants along any well-order simply by appeal to $\ATR$, it will be helpful to retain the well-order (and other free set variables) as a parameter of an explicit formula to be bound later in comprehension instances, and so we give the constructions explicitly here.


\begin{definition}
[(Bounded) approximants]
If $\WO(\alpha)$ and $a \in \alpha$ we write $\apprx X x {\phi} \alpha (a,x) $ (or even $\phi^\alpha (a,x)$) for:
\[
\exists F\subseteq \alpha \times \Nat \ (\forall b \in \alpha \forall y (Fby  \limp \exists c<_\alpha b\,  \phi (Fc,y)) \land  Fax)
\]
We also write,
\[
\begin{array}{r@{\ := \ }l}
     \apprx X x \phi \alpha(x) & \exists a \in \alpha \, \apprx X x\phi \alpha(a,x) \\
     \apprx X x \phi \WO (x) & 
     \exists \alpha ( \WO(\alpha) \land  \apprx X x \phi \alpha(x) ) 
\end{array}
\]
sometimes written simply $\phi^\alpha(x)$ and $\phi^\WO(x)$ respectively.
\end{definition}


\begin{proposition}
\label{prop:approx-is-posD12}
If $\WO(\alpha)$ then $\phi^\alpha$ is $\posnegDin 1 2 {\vec X}{\vec Y}$.
In particular, $\phi^\alpha(a,x)$ is equivalent to:
\begin{equation}
    \label{eq:G-phi-inductive-Gax}
\forall G \subseteq \alpha \times \Nat
(\forall b \in \alpha \forall y (\exists c <_\alpha b \phi (Gc,y) \limp Gby) \limp Gax )
\end{equation}
\end{proposition}
\begin{proof}
As written $\phi^\alpha (a,x)$ is a positive $\Sin 1 1 $ combination of $\phi$, so it is certainly $\posnegSin 1 2 {\vec X}{\vec Y}$ by Lemma~\ref{lem:pol-sub}.
So it suffices to prove the `in particular' clause,
since \eqref{eq:G-phi-inductive-Gax} is already a positive $\Pin 1 1 $ combination of $\phi$.%

For this we note that we can `merge' the two definitions into an instance of $\ATR$, obtaining (after $\phi$-comprehension) some $H$ satisfying for all $b\in \alpha $ and $y$:
\begin{equation}
    \label{eq:H-phi-atr}
    Hby \liff \exists c <_\alpha b\,  \phi(Hc,y))
\end{equation}
Now we show that in fact $Hax \liff \phi^\alpha(a,x) $. For the left-to-right implication, we simply witness the $\exists F$ quantifier in the definition of $\phi^\alpha$ by $H$.
For the right-to-left implication, let $F\subseteq \alpha \times \Nat$ such that:
\begin{equation}
    \label{eq:F-phi-recursive}
\forall b \in \alpha \forall y (Fby \limp \exists c <_\alpha b \phi (Fc,y))
\end{equation}
We show $\forall x(Fax \limp Hax)$ by $\alpha$-induction on $a$:
\[
\begin{array}{r@{\ \implies \ }ll}
     Fax & \exists b<_\alpha a \phi (Fb,x) & \text{by \eqref{eq:F-phi-recursive}} \\
        & \exists b<_{\alpha} a \phi (Hb,x) & \text{by inductive hypothesis} \\
        & Hax & \text{by \eqref{eq:H-phi-atr}}
\end{array}
\]
We may prove $Hax \liff \eqref{eq:G-phi-inductive-Gax} $ similarly, concluding the proof.
\end{proof}

Eventually we will also show that $\phi^\WO \in \posnegDin 1 2 {\vec X}{\vec Y}$ too, but we shall need to prove the fixed point theorem first, in order to appeal to a duality property.
First let us note a consequence of the above proposition:
\begin{corollary}
    [(Bounded) recursion]
    \label{cor:(bdd)-recursion}
   Let $\alpha \in \WO$. We have the following:
    \begin{enumerate}
        \item\label{item:bdd-recursion} (Bounded recursion) $\phi^\alpha(a) = \bigcup\limits_{b<_\alpha a} \phi(\phi^\alpha(b))$.
    I.e.\
    \[\phi^\alpha(a,x) \liff \exists b<_\alpha a \phi(\phi^\alpha(b),x)
    \]
    \item\label{item:recursion} (Recursion) $\phi^\alpha = \bigcup\limits_{\beta \ordlneq \alpha} \phi(\phi^\beta)$, i.e.\ 
    \[
    \phi^\alpha(x) \liff \exists \beta \ordlneq \alpha\, \phi(\phi^\beta, x)
    \]
    \end{enumerate}
\end{corollary}
\begin{proof}
    \ref{item:bdd-recursion} follows immediately from the equivalence between $\phi^\alpha$ and $H$ satisfying \eqref{eq:H-phi-atr} in the preceding proof. 
    For \ref{item:recursion} we first need an intermediate result.
    For $b\in \alpha \in \WO$, let us write $\alpha_b$ for the initial segment of $\alpha$ up to (and including) $b$.
    We show,
    \begin{equation}
        \label{eq:bdd-approximant-equals-approximant-upto-initial-segment}
        \forall x (\phi^\alpha (b,x) \liff \phi^{\alpha_b}(x))
    \end{equation}
    by transfinite induction on $b\in \alpha$.\anupam{discuss induction on arithmetical combinations} 
    We have:
    \[
    \begin{array}{r@{\ \iff \ }ll}
         \phi^\alpha(b,x) & \exists c<_\alpha b\, \phi(\phi^\alpha(c),x) & \text{by \ref{item:bdd-recursion}}\\
         & \exists c<_\alpha b \, \phi (\phi^{\alpha_c},x) & \text{by IH} \\
         & \exists c<_\alpha b \, \phi (\phi^{\alpha_b}(c),x) & \text{by IH} \\
         & \exists b' \in \alpha_b \exists c <_{\alpha_b} b' \phi (\phi^{\alpha_b}(c),x) & \text{set $b'=b$} \\
         & \exists b' \in \alpha_b \phi^{\alpha_b}(b',x) & \text{by \ref{item:bdd-recursion}}\\
         & \phi^{\alpha_b}(x) &\text{by dfn.\ of $\phi^{\alpha_b}$}
    \end{array}
    \]

    Now, turning back to \ref{item:recursion}, we have:
    \[
    \begin{array}[b]{r@{\ \iff \ }ll}
         \phi^\alpha(x) & \exists c \in \alpha \, \phi^\alpha(c,x) & \text{by definition} \\
         & \exists c \in \alpha \exists d<_\alpha c\, \phi(\phi^\alpha(d),x) &\text{by \ref{item:bdd-recursion}} \\
         & \exists c \in \alpha \exists d<_\alpha c \, \phi(\phi^{\alpha_d},x) & \text{by \eqref{eq:bdd-approximant-equals-approximant-upto-initial-segment}} \\
         & \exists c \in \alpha \exists \beta \ordlneq \alpha_c \, \phi(\phi^\beta,x) & \text{by basic ordinal properties \comment{?}} \\
         & \exists \beta \ordlneq \alpha \, \phi(\phi^\beta,x) \tag*{\qed}
    \end{array}
    \]
    \anupam{in penultimate step, might need invariance of approximants under isomorphic ordinals.}
    \renewcommand{\qed}{}
\end{proof}

\begin{proposition}
    [(Bounded) approximants are inflationary]
\label{prop:approximants-are-inflationary}
Let $\alpha,\beta \in \WO$. We have the following:
\begin{enumerate}
    \item\label{item:bdd-approximants-are-inflationary} $a\leq_\alpha b \limp \forall x (\phi^\alpha(a,x) \limp \phi^\alpha(b,x)) )$
    \item\label{item:approximants-are-inflationary} $\alpha \ordleq \beta \limp \forall x (\phi^\alpha(x) \limp \phi^\beta(x))$
\end{enumerate}
\end{proposition}
\begin{proof}
    \ref{item:bdd-approximants-are-inflationary} follows directly from Bounded Recursion:
    \[
    \begin{array}{rcll}
         \phi^\alpha(a,x) & \implies & \exists c<_\alpha a \, \phi(\phi^\alpha(c),x) & \text{by Corollary~\ref{cor:(bdd)-recursion}.\ref{item:bdd-recursion}}\\
         & \implies & \exists c<_\alpha b\, \phi(\phi^\alpha(c),x) & \text{since $a\leq_\alpha b$}\\
         & \implies &\phi^\alpha(b,x) & \text{by Corollary~\ref{cor:(bdd)-recursion}.\ref{item:bdd-recursion}} 
    \end{array}
    \]
\ref{item:approximants-are-inflationary} follows directly from Recursion:
\[
\begin{array}[b]{r@{\ \implies \ }ll}
     \phi^\alpha(x) & \exists \gamma \ordlneq \alpha \, \phi(\phi^\alpha,x) & \text{by Corollary~\ref{cor:(bdd)-recursion}.\ref{item:recursion}}\\
     & \exists \gamma \ordlneq \beta \, \phi(\phi^\gamma, x) & \text{since $\alpha \ordleq \beta$} \\
     & \phi^\beta(x) & \text{by Corollary~\ref{cor:(bdd)-recursion}.\ref{item:recursion}} 
\end{array}
\]
    \anupam{commented old direct proof below}
    \renewcommand{\qed}{}
\end{proof}

\anupam{commented bounded approximant collection below}

\begin{definition}
[Least (pre)fixed points]
Define $(\mufo X x \phi)(y)$ (or even $(\mu\phi)(y)$) by:
\[
(\mufo X x \phi) (y) := \forall X (\forall x ( \phi(X,x) \limp Xx ) \limp Xy)
\]
\end{definition}
Note that we may treat $\mu\phi$ as a set in $\PSCA$ since $\phi\in \Din 1 2$, and so $(\mu\phi)(y) $ is $ \Pin 1 2$.
By mimicking a text-book proof of the Knaster-Tarski theorem we obtain:
\begin{proposition}
[`Knaster-Tarski']
    \label{prop:kt-for-least-pre-fixed-point}
    $\phi(\mu\phi)= \mu \phi$, i.e.\ 
    \[
    \forall x (\phi(\mu\phi,x) \limp \mu\phi\, x)
    \]
\end{proposition}
\begin{proof}
    For $\phi(\mu\phi)\subseteq \mu\phi$, note that for any $X$ with $\phi(X)\subseteq X$ we have:
    \[
    \begin{array}{r@{\ \limp \ }ll}
         \mu\phi\, x & Xx & \text{by definition of $\mu\phi$} \\
         \phi(\mu\phi,x) & \phi(X,x) & \text{by Lemma~\ref{lem:pos-imp-mon}}\\ 
         \phi(\mu\phi,x) & Xx & \text{since $\phi(X)\subseteq X$}
    \end{array}
    \]
    Thus $\forall X(\phi(X)\subseteq X \limp \forall x ( \phi(\mu\phi,x) \limp Xx)) $, and so indeed $\phi(\mu\phi)\subseteq \mu\phi$.

    For $\mu\phi \subseteq \phi(\mu\phi)$, we have:
    \[
    \begin{array}[b]{rcll}
         \phi(\mu\phi) &\subseteq & \mu\phi &\text{by above}\\
         \phi(\phi(\mu\phi)) & \subseteq &\phi(\mu\phi) & \text{by Monotonicity Lemma~\ref{lem:pos-imp-mon}}\\
         \phi(\mu\phi) & \subseteq & \mu\phi & \text{by comprehension on $\phi(\mu\phi)$} \tag*{\qed}
    \end{array}
    \]
    \renewcommand{\qed}{}
\end{proof}

\begin{lemma}
    [Closure ordinals]
    \label{lem:closure-ordinals}
    $\exists \alpha \in \WO \, \phi^\WO \subseteq \phi^\alpha$.
    I.e.\ 
    \[
    \forall x (\phi^\WO(x) \limp \phi^\alpha(x))
    \]
\end{lemma}
\begin{proof}
    [Proof sketch]
    We have:\footnote{For the penultimate step, recall that we have access to $\phi^\WO$ as a set, under $\Sin 1 2$-comprehension.}
    \[
    \begin{array}{rll}
         & \forall x (\phi^\WO(x) \limp \exists \alpha \in \WO \phi^\alpha(x)) & \text{by dfn.\ of $\phi^\WO$}\\
         \therefore & \forall x \exists \alpha \in \WO (\phi^\WO(x) \limp \phi^\alpha(x)) & \text{by pure logic}\\
         \therefore & \exists F : \Nat \to \WO \, \forall x (\phi^\WO(x) \limp \phi^{Fx}(x))  & \text{by $\Sin 1 2 $-choice}\\
         \therefore & \forall x (\phi^\WO \limp \phi^{\upbnd F }(x)) & \text{by Prop.~\ref{prop:approximants-are-inflationary}.\ref{item:approximants-are-inflationary}}
    \end{array}
    \]
    where $\upbnd F\in \WO$ is obtained by Bounding,\footnote{To be precise, we apply bounding on the formula $\exists x \forall a (Xa \liff Fxa)$.} Proposition~\ref{prop:bounding}, such that $\forall x \upbnd F \ordgeq Fx$.
\end{proof}

\anupam{commented approximant collection below}



The main result of this subsection is:
\begin{theorem}
[LFP dual characterisation]
\label{thm:muphi-is-phiwo}
$\mu \phi =  \phi^\WO$, i.e.:
\[
\forall x \left(
\mu \phi (x) \liff \phi^\WO(x)
\right)
\]
\end{theorem}
\begin{proof}
For the left-right inclusion, $\mu \phi \subseteq  \phi^\WO$, we show that $\phi^\WO$ is a pre-fixed point of $\oper X x \phi$, i.e.:
\[
\forall x (\phi (\phi^\WO,x) \limp \phi^\WO (x))
\]
By Lemma~\ref{lem:closure-ordinals} let $\alpha $ such that $\phi^\WO = \phi^\alpha$.
We have:
\[
\begin{array}{r@{\ \implies \ }ll}
     \phi(\phi^\WO,x) & \phi(\phi^\alpha,x) & \text{by assumption} \\
     & \phi^{\succ \alpha}(x) & \text{by Cor.~\ref{cor:(bdd)-recursion}.\ref{item:recursion} } \\
     & \phi^\WO (x) & \text{by definition of $\phi^\WO$}
\end{array}
\]
Since we have access to $\phi^\WO$ as a set, under $\ca{\Sin 1 2}$, this indeed yields $\mu \phi \subseteq \phi^\WO$.

For the right-left inclusion, $\phi^\WO \subseteq \mu \phi$, let $\alpha \in \WO$ and we show:
\[
\forall a \in \alpha ( \phi^\alpha(a,x) \limp \mu \phi\, x)
\]
by transfinite induction on $a \in \alpha$.
For logical complexity, recall that we indeed have access to $\mu\phi$ as a set, since $\phi\in \Din 1 2$ so $\mu\phi \in \Pin 1 2$.
We have:
\[
\begin{array}[b]{r@{\ \implies \ }ll}
     \phi^\alpha(a,x) & \exists b<_\alpha a \, \phi(\phi^\alpha(b),x) & \text{by Corollary~\ref{cor:(bdd)-recursion}.\ref{item:bdd-recursion}}\\
     & \exists b<_\alpha a \, \phi(\mu\phi,x) & \text{by inductive hypothesis} \\
     & \phi(\mu\phi,x) & \text{by vacuous quantification} \tag*{\qed}
\end{array}
\]
\renewcommand{\qed}{}
\end{proof}

One of the main consequences of the above result is:
\begin{corollary}
    [$\PSCA$]
    \label{cor:lfp-of-posD12-is-posD12}
    $\mufo X x \phi$ is $\posnegDin 1 2 {\vec X}{\vec Y}$.
\end{corollary}

\anupam{commented old nu stuff below. obsolete.}







Let us point out that the above result amounts to a partial arithmetisation of purely descriptive characterisation of $\mu$-definable sets in Lubarsky's work \cite{Lubarsky93}.
Note also that, while we (in particular) obtain a $\Delta^1_2$ bound on fixed point formulas, we crucially required $\CA{\Pi^1_2}$ to prove this, consistent with results of \cite{Moellerfeld02:phd-thesis} that we shall exploit in the next section to fill in our `grand tour'. 

\subsection{Arithmetising the totality argument}

Thanks to the results of this section, we may duly formalise the type structure $\HR$ from Section~\ref{sec:totality} within $\PSCA$ and prove basic properties.
Throughout this section we shall identify terms with their codes. In the case of $\clomuLJneg$ and subsystems, we note that terms may still be specified finitely thanks to regularity of coderivations in $\cmuLJneg$, and that syntactic equality is verifiable already in $\RCA$ between different representations~\cite{Das21:CT-preprint,Das21:CT-fscd}.
Let us also note that the various notions of reduction, in particular $\convertsnorec$ and $\convertsnoreceta $ are $\Sin 0 1$ relations on terms. (Note here that it is convenient that $\convertsnoreceta$ is formulated using only the extensionality rule, for $\Sin 0 1$, rather than being fully extensional.)

Let us recall the $\HR$ structure from Section~\ref{sec:totality}. 
Note that, as written in Definition~\ref{def:hr},
these sets, a priori, climb up the analytical hierarchy due to alternation of $\mu$ and $\arrow$.
This is where the results of the previous section come into play and serve to adequately control the logical complexity of $\HR$.




\begin{definition}
[Type structure, formalised]
We define the following second-order formulas:
\begin{itemize}
    \item $\hr X (t) \dfn Xt$
    \item $\hr \nat (t) \dfn  \exists n\,  (t \convertsnoreceta \numeral n) $ 
    \item $\hr {\sigma \to \tau} (t) \dfn \forall s ( \hr\sigma (s) \limp \hr \tau (ts))$
    \item $\hr{\sigma \times \tau} (t) \dfn \hr \sigma (\proj 0 t) \land \hr \tau (\proj 1 t) $
    \item  $\hr {\mu X \sigma} (t) \dfn (\mufo X x {\hr \sigma (x)}) (t)$
\end{itemize}
\end{definition}

Weak theories such as $\RCA$ are able to verify closure of $\clocmuLJneg$ under conversion by a syntax analysis (see~\cite{Das21:CT-fscd}).
That each $\hr\sigma$ itself is closed under conversion, cf.~\ref{prop:hr-closed-under-conversion}, is also available already in $\RCA$ by (meta-level) induction on the structure of $\sigma$.
More importantly, and critically, we have as a consequence of the previous section:

\begin{corollary}
    \label{cor:hr-is-D12}
    Let $\sigma$ be positive in $\vec X$ and negative in $\vec Y$.
    $\hr \sigma$ is $\PSCA$-provably $\posnegDin 1 2 {\vec X}{\vec Y}$. 
\end{corollary}
\begin{proof}[Proof sketch]
    We proceed by induction on the structure of $\sigma$, for which the critical case is when $\sigma$ has the form $\mu X \tau$.
    By definition $\tau$ is positive in $X$, and so by IH $\hr{\tau} \in \posnegDin 1 2 {X,\vec X}{\vec Y}$.
    We conclude by Corollary~\ref{cor:lfp-of-posD12-is-posD12}.
\end{proof}



Now, recall that a function $f:\Nat^k \to \Nat$ is \emph{provably recursive} in a theory $T\subseteq \SOPA$ if there is a $\Sin 0 1 $-formula $\phi_f(\vec x,y)$ with:
\begin{itemize}
    \item $ \nmod \models \phi_f (\vec {\numeral m} , \numeral n) $ if and only if $f(\vec m) = n$; and,
    \item $T \proves \forall \vec x \exists y \phi_f (\vec x, y)$.
\end{itemize}

We may formalise the entire totality argument within $\PSCA$, in a similar fashion to analogous arguments in \cite{Das21:CT-fscd,KuperbergPP21systemT,Das21:CT-preprint}, only peculiarised to the current setting.
This requires some bespoke arithmetical arguments and properties of our type structure.
As a consequence we shall obtain:

\begin{theorem}
\label{thm:cmuLJneg-to-psca}
Any $\cmuLJneg$-representable function on natural numbers is provably recursive in $\PSCA$.
\end{theorem}

We give some more details for establishing the above theorem. In particular, we apply the fixed point theorems within second-order arithmetic from Section~\ref{sec:kt-thm-and-approximants-in-psca} to arithmetise the totality argument in Section~\ref{sec:totality} within $\PSCA$.
We shall focus on explaining at a high level the important aspects of the formalisation, broadly following the structure of Section~\ref{sec:totality}.

Notice that it will not actually be possible to prove the Interpretation Theorem~\ref{thm:prog-implies-hr} uniformly, due to G\"odelian issues.
Instead we will demonstrate a non-uniform version of it:
\begin{theorem}
    [Nonuniform Interpretation, formalised]
    \label{thm:interpretation-nonuniform-formalised}
    Let $\der:\Sigma \seqar \tau$.  
    Then
    $\PSCA \proves \forall \vec s \in \hr\Sigma .\  \der \vec s \in \hr \tau $.
\end{theorem}

Note that, from Theorem~\ref{thm:interpretation-nonuniform-formalised} above, Theorem~\ref{thm:cmuLJneg-to-psca} will follow directly via a version of Corollary~\ref{cor:cmuLJ-represents-only-total-fns} internal to $\PSCA$.
In particular, when $\Sigma$ contains only $\N$s and $\tau$ is $\N$, note that the statement of Theorem~\ref{thm:interpretation-nonuniform-formalised}, $\forall \vec s \in \hr\Sigma .\  \der \vec s \in \hr \tau $, is indeed $\Pin 0 2$.

The rest of this section is devoted to justifying \Cref{thm:interpretation-nonuniform-formalised} above.
Let us fix $P:\Sigma \seqar \tau$ for the remainder of this section.

\subsubsection*{Formalising monotonicity and transfinite types}
All the technology built up in Subsection~\ref{subsec:monotonicity-and-transfinite-types} has been appropriately formalised in the previous Section~\ref{sec:kt-thm-and-approximants-in-psca}.

\subsubsection*{Formalising closures and priorities}
All the notions about closures and priorities in the totality argument involve only finitary combinatorics and are readily formalised within $\RCA$.

\subsubsection*{Formalising ordinal assignments}
We define ordinal assignments within $\PSCA$ relative to some $\vec \alpha$ varying over $\WO$. 
%
What is important is to establish the `positive and negative approximants' Proposition~\ref{prop:pos-neg-approximants}. 
However this follows directly from the Closure Ordinal Lemma~\ref{lem:closure-ordinals} and the Monotonicity Lemma~\ref{lem:pos-imp-mon}.

\subsubsection*{Formalising reflection of non-totality.} 

Since $P$ has only finitely many distinct lines, and so only finitely many distinct formulas, we will need to consider only finitely many distinct $\hr\sigma$ henceforth which we combine to describe to establish the reflection of non-totality Definition~\ref{def:refl-non-totality} all at once within SO arithmetic.
Working in $\PSCA$, let us point out that the description of $P', \Sigma', \tau',\vec s'$ from $\rrule$ and $P$'s (finitely many) sub-coderivations $\vec P$ is recursive in the following oracles:
\begin{itemize}
    \item $\hr\sigma$, for each $\sigma $ occurring in $P$, which is $\Din 1 2 $ by Corollary~\ref{cor:hr-is-D12}; and,
    \item finding the least $\alpha \in \WO$ such that $ \hr{\sigma^\alpha}(t)$; this is $\Din 1 2 $ by fixing the closure ordinal, cf.~\ref{lem:closure-ordinals}, say $\gamma \in \WO$ of $\mu \hr\sigma$, and searching for the least appropriate $a\in \gamma$ instead; again this is $\Din 1 2$;
    \item finding an inhabitant $s$ of some nonempty $\hr\sigma$, for which we can simply take the `least' (seen as a natural number coding it), and so is $\Din 1 2$.
\end{itemize}

Consequently we have that Definition~\ref{def:refl-non-totality}, the description of $P', \Sigma', \tau',\vec s'$ from $\rrule, \vec P$, is indeed a $\PSCA$-provably $\Din 1 2$ formula in $P', \Sigma', \tau',\vec s',\rrule, \vec P$.

\anupam{point out we need KT theorem for $\mu$-right?}

\subsubsection*{Formalising the non-total branch construction}

The `non-total' branch constructed in the proof of the Interpretation Theorem~\ref{thm:prog-implies-hr} is recursive in the `reflecting non-totality' definition, and so we have access to it as a set within $\PSCA$.

One subtlety at this point is that $\PSCA$ needs to `know' that $\der$ is indeed progressing.
However earlier work on the cyclic proof theory of arithmetic~\cite{Simpson17:cyc-arith,Das20:ca-log-comp}, building on the reverse mathematics of $\omega$-automaton theory \cite{KMPS19:buchi-rev-math}, means that this poses us no problem at all:

\begin{proposition}
    [Formalised cyclic proof checking~\cite{Das20:ca-log-comp}]
    $\RCA$ proves that $P$ is progressing, i.e.\ that every infinite branch has a progressing thread.
\end{proposition}

From here we readily have that there is a progressing thread $(\rho_i)_{i\geq k}$ and inputs $r_i$ along along the non-total branch.
The assignment of ordinals $(\vec\alpha_i)_{i\geq k}$ along the branch follows by the `positive and negative approximants' result (formalisation of \Cref{prop:pos-neg-approximants}).
The verification that $\vec \alpha_{i+1}\leq \vec \alpha_i$ for non-critical steps follows by inspection of the `reflecting non-totality' definition (formalisation of \Cref{def:refl-non-totality}).
Finally, for the critical fixed point unfolding, we need that $\alpha_{in}$ is a successor. 
For this we need that the approximant at a limit ordinal is a union of smaller approximants, for which we use Recursion, cf.~Corollary~\ref{cor:(bdd)-recursion}.

\medskip

This concludes the argument for Theorem~\ref{thm:interpretation-nonuniform-formalised}, and so also of Theorem~\ref{thm:cmuLJneg-to-psca}.

\section{Realisability with fixed points}
\label{sec:realisability}
\todonew{
\begin{itemize}
    \item include proof of $\Pi_2$ conservativity of $\muPA$ over $\muHA$ in main text?
    \item switch typesetting of predicate $\pNat$ and type $\nat$ macros?
    \item could make further remarks about abstract realisability just after definition of judgement (some commented copy there).
    \item check and try improve remark about comparison to system F.
    \item at same time can simplify the typed version of $\typedclomuLJneg$ from earlier to be more intuitive.
    \item for functoriality lemma can omit all noncriticial cases?
    \item add comments about abstract realisability inspired by Berger and Tsuiki's appropach, but point out that our realisability model is different as we do not interpret fixed points as fixed points, but are rather inspired by SO encodings.
\end{itemize}
}

So far we have shown an \emph{upper bound} on the representable functions of $\muLJ$ and $\cmuLJ$, namely that they are all provably total in $\PSCA$. 
In this section we turn our attention to proving the analogous \emph{lower bound}, namely by giving a \emph{realisability} interpretation into $\muLJ$ (in fact typed-$\closure \muLJneg$) from a theory over which $\PSCA$ is conservative.

In particular, we consider a version of first-order arithmetic, $\muPA$, with native fixed point operators, (essentially) introduced by M\"ollerfeld in \cite{Moellerfeld02:phd-thesis}.
Unlike that work, we shall formulate $\muPA$ in a purely first-order fashion in order to facilitate the ultimate realisability interpretation.




\subsection{Language of arithmetic with fixed points}
We consider an extension of $\PA$ whose language is closed under (parameterised) least (and greatest) fixed points.
Throughout this section we shall only use logical symbols among $\land, \limp, \exists , \forall$, without loss of generality.

\begin{convention}
\label{conv:langarith-has-all-pr-symbols}
We shall henceforth assume, without loss of generality, that the language of arithmetic $\langarith$ contains a function symbol for each primitive recursive function definition.
All arithmetic theories we consider, like $\PA$, $\HA$ and their extensions, will contain the defining equational axioms for each of these function symbols.
Note that these formulations are just \emph{definitional extensions} of the usual versions of $\PA$ and $\HA$.
However, they allow us to construe $\Delta_0$ formulas as just equations, simplifying some of the metamathematics herein.
\end{convention}

Let us extend the language of arithmetic $\langarith$ by countably many predicate symbols, written $X,Y$ etc., that we shall refer to as `set variables'. 
In this way we shall identify $\langarith$-formulas with the arithmetical formulas of the language of second-order arithmetic $\langsoarith$, but we shall remain in the first-order setting for self-containment.
As such, when referring to the `free variables' of a formula, we include both set variables and number variables.

$\langmuarith$-formulas are generated just like $\langarith$-formulas with the additional clause:

\begin{itemize}
    \item if $\phi$ is a formula with free variables $X,\vec X, x, \vec x$, and in which $X$ occurs positively, and $t $ is a (number) term with free variables $\vec y$ then $t \in \mufo X x \phi$ \oldtodo{notation can be optimised, this is a macro} (or even $\mu \phi\, t$ when $X,x$ are clear from context) is a formula with free variables $\vec X,\vec x,\vec y$.
\end{itemize}

We also write $t \in \nufo X x \phi(X)$ for $t \notin \mufo X x \lnot \phi(\lnot X)$, a suggestive notation witnessing the De Morgan duality between $\mu$ and $\nu$ (in classical logic).

\begin{remark}
[$\langmuarith$ as a proper extension of $\langsoarith$]
\label{rem:langmuarith-extends-langsoarith}
Again referring to the identification of $\langarith$ formulas with arithmetical $\langsoarith$ formulas, we may construe $\langmuarith$ as a formal extension of $\langsoarith$ by certain relation symbols.
Namely
$\langmuarith$ may be identified with the closure of $\langsoarith$ under:
    \begin{itemize}
        \item if $\phi$ is an \emph{arithmetical} formula with free variables among $\vec X, X, \vec x, x$, and in which $X$ occurs positively, then there is a relation symbol $\mufo X x \phi$  taking $|\vec X|$ set inputs and $|\vec x,x|$ number inputs.
    \end{itemize}
In this case, for arithmetical $\phi(\vec X, X, \vec x, x)$, we simply write, say, $t \in \mufo X x {\phi(\vec A, X, \vec t, x)} $ instead of $\mufo X x {\phi(X,\vec X,x,\vec x)} \vec A \,  \vec t\, t$.
Let us note that this is indeed the route taken by M\"ollerfeld.
Instead we choose to treat the construct $\mu \cdot \lambda \cdot $ syntactically as `binder'.
\end{remark}

Semantically we construe $\mufo X x \phi$ in the intended model as a bona fide least fixed point of the (parametrised) operator defined by $\phi$. 
Namely, in the sense of the remark above, we can set $(\mufo X x \phi)^\nmod (\vec A,\vec a)$ to be the least fixed point of the operator $(\Lambda X \lambda x \phi(\vec A, X, \vec a, x))^\nmod$, for all sets $\vec A$ and numbers $\vec a$.
Note that by simple algebraic reasoning this forces $\nufo X x \phi$ to be interpreted as the analogous \emph{greatest} fixed point in $\nmod$.
\anupam{this is a bit brief, but is fine for now}

As in earlier parts of this work, we shall frequently suppress variables in formulas to denote abstractions, e.g.\ for formulas $\phi(X)$ and $\psi(x)$, with $X,x$ clear from context, we may write $\phi(\psi)$ for the formula obtained by replacing each subformula $Xt$ of $\phi(X)$ by $\psi(t)$.

\subsection{Theories $\muPA$ and $\muHA$}
Let us expand Peano Arithmetic ($\PA$) to the language $\langmuarith$, i.e.\ by including induction instances for all $\langmuarith$-formulas.
The theory we consider here is equivalent to (the first-order part of) M\"ollerfeld's $\ACA(\langmuarith)$.

\begin{definition}
    [Theory]
    \label{def:muPA}
    The theory $\muPA$ is the extension of $\PA$ by the following axioms for formulas $\phi(X,x)$ and $\psi(x)$:
    \begin{itemize}
        \item $\preaxiom_\phi$: $\forall x (\phi(\mu \phi, y) \limp \mu \phi\,  x)$
        \item $\indaxiom_{\phi,\psi}$: $\forall x (\phi(\psi, x) \limp \psi(x)) \limp \forall x (\mu \phi\, x  \limp \psi(x))$
    \end{itemize}
\end{definition}

We sometimes omit the subscripts of the axiom names above.
We may construe $\muPA$ as a proper fragment of full second-order arithmetic $\SOPA$ by the interpretation:
\begin{equation}
    \label{eq:mufo-as-so-int-pre}
    t \in \mufo X x \phi 
    \quad := \quad
    \forall X (\forall x (\phi(X,x) \limp Xx) \limp Xt)
\end{equation}
The axioms above are readily verified by mimicking a standard textbook algebraic proof of Knaster-Tarski in the logical setting, cf.~Proposition~\ref{prop:kt-for-least-pre-fixed-point}.

M\"ollerfeld's main result was that $\PSCA$ is in fact $\Pin 1 1 $-conservative over his theory $\ACA(\langmuarith)$, and so we have:

\begin{theorem}
    [Implied by \cite{Moellerfeld02:phd-thesis}]
    \label{thm:moellerfeld-psca-cons-over-muPA}
    $\PSCA$ is arithmetically conservative over $\muPA$.
\end{theorem}

\anupam{commented fo-parts below}

We set $\muHA$ to be the intuitionistic counterpart of $\muPA$.
I.e.\ $\muHA$ is axiomatised by Heyting Arithmetic $\HA$ (extended to the language $\langmuarith$) and the schemes $\preaxiom$ and $\indaxiom$ in Definition~\ref{def:muPA} above.

Once again, we may construe $\muHA$ as a proper fragment of full second-order Heyting Arithmetic $\SOHA$ by \eqref{eq:mufo-as-so-int-pre} above.
By essentially specialising known conservativity results for second-order arithmetic, we may thus show that $\muPA$ and $\muHA$ provably define the same recursive functions on natural numbers.
\begin{proposition}
[Implied by \cite{Tupailo04doubleneg}]
\label{thm:muPA-pi02-cons-muHA}
    $\muPA$ (so also $\PSCA$) is $\Pin 0 2$-conservative over $\muHA$.
\end{proposition}

We give a self-contained argument for this result in \Cref{sec:conservativity-via-double-negation}, essentially by composing the Friedman-Dragalin $A$-translation and G\"odel-Gentzen negative translation.

\anupam{commented proof idea below}

\anupam{commented `target of realisability' above}

\subsection{Relativisation to $\Nat$}
It will be convenient for our realisability argument to work with a notion of \emph{abstract realisability}, where realisability commutes with quantifiers in favour of explicit relativisation of quantifiers to suitable domains \cite{Troelstra98:realizability-handbook}.
The reason for this is that, a priori, all quantifiers of arithmetic are relativised to $\Nat$, but construing so for the fixed point axioms leads to type mismatch during realisability.
For instance, one would naturally like to realise induction axiom for natural numbers by $\Niter$, for which it is natural to consider the $\forall$ of the inductive step unrelativised, whereas the $\forall$ of the conclusion should certainly be relativised to natural numbers.
The same phenomenon presents for the $\indaxiom$ axioms more generally. 
Thus we shall include such relativisations explicitly to handle this distinction.

We introduce a new unary predicate symbol $\pNat $ and write $\langmuarith^\pNat \dfn \langmuarith \cup \{\pNat\}$.\footnote{We use $\pNat$ to avoid confusion with our earlier type $\nat $ for natural numbers.}
$\pNat$ will morally stand for the fixed point $ \mu X\lambda x (x=0 \lor \exists y (Xy \land x=\succ y))$, computing the natural numbers.
However we shall rather take logically equivalent formulations of the prefix and induction axioms for $\pNat$ that are \emph{negative}, in fact ultimately realised by our analogous specialisations of the $\iter{}$ and $\inj{}$ rules for $\nat$ earlier:
\begin{itemize}
    \item $\preaxiom_\pNat^0$: $\pnat \numeral 0$
    \item $\preaxiom_\pNat^\succ$: $\forall x (\pnat x \limp \pnat{ \succ x})$
    \item $\indaxiom_{\pNat,\phi}$:
    $\phi(0) \limp \forall x (\phi(x) \limp \phi(\succ x )) \limp \forall x (\pnat x \limp \phi(x))$
\end{itemize}

We sometimes write $\forall x^\pNat \phi \dfn \forall x (\pnat x \limp \phi)$ indicating that the universal quantifier is \emph{relativised} to $\pNat$.
We similarly write $\exists x^{\pNat} \phi \dfn \exists x (\pnat x \land \phi)$.
In arithmetic all quantifiers are implicitly relativised in this way, by virtue of the induction axiom schema.
However note that the `correct' formulation of induction above, induced by the fixed point definition of $\pNat$, has a non-relativised universal quantifier for the step case. 
This will turn out to be an important refinement in order to avoid type mismatches when conducting realisability. 
For this reason, our realisability argument must work from a theory in which this distinction is native:

\begin{definition}
    Write $\muHAneg$ for the theory defined like $\muHA(\langmuarith^\pNat)$ but, instead of the numerical induction axioms, we include the axioms $\preaxiom_\pNat^0$, $\preaxiom_\pNat^\succ$ and $\indaxiom_\pNat$ defined earlier.
    I.e.\ $\muHAneg$ is the intuitionistic theory over $\langmuarith^\pNat$ including all the axioms of Robinson Arithmetic, all the defining equations of primitive recursive function symbols, all the $\preaxiom_\phi$ and $\indaxiom_{\phi,\psi}$ axioms,  and all the axioms $\preaxiom_\pNat^0$, $\preaxiom_\pNat^\succ$ and $\indaxiom_{\pNat,\phi}$.
\end{definition}

The notation here is suggestive of our analogous notation for negative fragments of $\muLJ$ and $\cmuLJ$ earlier.
Indeed we shall soon see that our notion of `realising type' for $\muHAneg$ will have image over $\nat, \times, \limp,\mu$.
First, let us verify that $\muHAneg$ indeed interprets $\muHA$.

\begin{definition}
    For formulas $\phi$ of $\langmuarith$ we define $\nattrans \phi$ a formula of $\langmuarith^\pNat$ by:
    \begin{itemize}
        \item $ \nattrans{(s=t)} \dfn s=t $
        \item $\nattrans {(\phi \star \psi)} \dfn \nattrans \phi \star \nattrans \psi $ for $\star \in\{ \land, \limp\}$
        \item $\nattrans {(\exists x \phi)} \dfn \exists x (\pnat x \land \phi)$
        \item $\nattrans{(\forall x \phi)}\dfn \forall x (\pnat x \limp \phi)$
          \item $ \nattrans{(t \in\mu X \lambda x \phi)} \dfn t \in \mu X \lambda x (\pnat x \land \nattrans\phi)$
    \end{itemize}
\end{definition}

We extend this translation to (definable) predicates, e.g.\ writing $\nattrans \phi(\vec X,\vec x)\dfn \nattrans{\phi(\vec X,\vec x)}$ and, in particular, $\nattrans{(\mu\phi)}x\dfn \nattrans{(x \in \mu\phi)}$.
%
%
%
Specialising a well-known reduction from second-order arithmetic to pure second-order logic, we have:

\begin{proposition}
\label{prop:muHA-to-muHAneg}
    If $\muHA \proves \phi$ then $\muHAneg \proves \nattrans \phi$.
\end{proposition}
\begin{proof}
    The Robinson axioms are already part of $\muHAneg$ and, since they are universal statements, so too their relativisations.
    The same argument applies for the defining equations of primitive recursive function symbols.
    
    Each number induction axiom,
    \begin{equation}
        \label{eq:n-ind-instance-for-relativisation}
        \psi(\numeral 0) \limp \forall x (\psi(x) \limp \psi(\succ x)) \limp \forall x \psi(x)
    \end{equation}
    has $\pNat$-translation,
    \[
    \nattrans\psi(\numeral 0)
    \limp 
    \forall x^\pNat (\nattrans\psi(x) \limp \nattrans\psi(\succ x))
    \limp 
    \forall x^\pNat (\nattrans\psi(x))
    \]
    which, by the $\preaxiom_\pNat$ axioms and pure logic, is equivalent to:
    \[
    (\pnat {\numeral 0} \land \nattrans\psi(\numeral 0)
    \limp
    \forall x ((\pnat x \land \nattrans\psi(x)) \limp (\pnat{\succ x } \land \nattrans \psi(\succ x)))
    \limp
    \forall x^\pNat (\pnat x \land \nattrans \psi(x))
    \]
    This is just an instance of $\indaxiom_\pNat$ with invariant $\pnat x \land \nattrans \psi(x)$.

    Finally let us consider the fixed point axioms.
    First notice that a $\preaxiom$ axiom,
    \[
    \forall y (\phi(\mu \phi,y) \limp y \in \mu \phi)
    \]
    has $\pNat$-relativisation,
    \[
    \forall y \in \pNat ( \nattrans \phi (\nattrans{(\mu\phi)},y) \limp y \in \nattrans{(\mu\phi)})
    \]
    which is logically equivalent to the $\preaxiom$ axiom for $\nattrans {(\mu\phi)}$.
    
    Next, for an $\indaxiom_\phi$ axiom,
    \begin{equation}
    \label{eq:indaxiom-for-relativisation}
        \forall x (\phi(\psi,x) \limp \psi(x)) \limp \forall x (x \in \mu\phi \limp \psi(x))
    \end{equation}
    we derive its $\pNat$-relativisation as follows,
    \[
    \begin{array}{rll}
         &\forall x ((\pnat x \land \nattrans \phi (\nattrans \psi,x) ) \limp \nattrans \psi(x)) \limp \forall y (y \in \nattrans{(\mu\phi)} \limp \nattrans \psi (y)) & \text{by $\indaxiom$ for $\nattrans{(\mu\phi)}$}   \\
         \implies & \forall x \in \pNat ( \nattrans \phi (\nattrans \psi,x)  \limp \nattrans \psi(x)) \limp \forall y \in \pNat (y \in \nattrans{(\mu\phi)} \limp \nattrans \psi (y)) & \text{by pure logic}
    \end{array}
    \]
    where the last line is just the $\pNat$-relativisation of \eqref{eq:indaxiom-for-relativisation}.
\end{proof}

\subsection{An abstract realisability judgement}
In what follows we shall work with the untyped calculus $\clomuLJneg$ and its typed version $\typedclomuLJneg$ from Section~\ref{sec:term-calculi}.
For convenience we shall employ the following convention:
\begin{convention}
We henceforth identify terms of $\langarith$ with their corresponding definitions in $\muLJneg$. 
We shall furthermore simply identify closed terms of arithmetic with the numerals they reduce to under $\converts$.
Note that, for $s,t$ terms of $\langarith$, we indeed have $s\converts t \implies \muHAneg \proves s=t$.
\end{convention}
At least one benefit of the convention above
is that we avoid any confusion arising from metavariable clash between terms of arithmetic and terms of $\clomuLJneg$. 
This convention is motivated by the clauses of our realisability judgement for atomic formulas.

\medskip

A \emph{(realisability) candidate} is an (infix) relation $\cdot A \cdot $ relating a (untyped) term to a natural number such that $\cdot A n$ is always closed under $\converts$. 
Let us expand $\langmuarith^\pNat $ by construing each realisability candidate as a unary predicate symbol.
The idea is that each $\cdot An $ consists of the just the realisers of the sentence $A\numeral n$.

\anupam{could formulate all of below with closed terms, by defining a valuation function. Need different metavariables, e.g.\ $u,v,w$, reserving $r,s,t$ for terms of $\clomuLJneg$. Do this later if there is time.}
\begin{definition}
    [Realisability judgement]
    For each term $t\in \clomuLJneg$ and closed formula $\phi$ we define the (meta-level) judgement $t \, \realises\, \phi$ as follows:
\begin{itemize}
    \item $t\, \realises \, (\numeral m = \numeral n)$ if $\numeral m \converts t \converts \numeral n$.
    \item $t \, \realises \, \pnat {\numeral n}$ if $t\converts \numeral n$.
     \item $t\, \realises\, A\numeral n$ if $t A n$.
    \item $t\, \realises\, \phi_0 \land \phi_1$ if $\projl t \, \realises \phi_0$ and $\projr t \, \realises\, \phi_1$.
    \item $t \, \realises\,  \phi \limp \psi $ if, whenever $s\, \realises \, \phi$, we have $ts \, \realises\, \psi$.
    \item $t\, \realises\, \exists x \phi(x)$ if for some $n\in \Nat$ we have $t\, \realises \, \phi(\numeral n)$.
    \item $t\, \realises\,  \forall x \phi(x)$ if for all $n \in \Nat $ we have $t\, \realises\, \phi(\numeral n)$.
    \item $t\, \realises\, \numeral n \hspace{-.15em}\in \hspace{-.2em} \mu X \lambda x \phi(X,x)$ if $t\, \realises \, \forall x (\phi(A,x)\limp Ax) \limp A\numeral n$ for all candidates $A$.
\end{itemize}
\end{definition}

    

\begin{definition}
    [Realising type]
    The \emph{realising type} of a (possibly open) formula $\phi$ without candidate symbols, written $\type \phi$, is given by:
\begin{itemize}
    \item $\type{s=t} \dfn \N$
    \item $\type{\pnat t} \dfn \N$
    \item $\type {Xt} \dfn X$
    \item $\type{\phi\land \psi} \dfn \type \phi \times \type \psi$
    \item $\type{\phi \limp \psi} \dfn \type \phi \to \type \psi$
    \item $\type{\exists x \phi}\dfn  \type \phi$
     \item $\type{\forall x \phi} \dfn \type\phi$
    \item $\type {t \in \mufo X x \phi} \dfn \mu X  \type \phi$
\end{itemize}
\end{definition}

\subsection{Closure under $\converts$ and realising induction}
Our realisability model is compatible with our notion of conversion from $\muLJ$:

\begin{lemma}
    [Closure under $\converts$]
    \label{lem:realises-closed-under-converts}
    If $t \, \realises\, \phi$ and $t\converts t'$ then $t'\, \realises \, \phi$.
\end{lemma}
\begin{proof}
    By induction on the structure of $\phi$:
    \begin{itemize}
        \item if $\phi$ is $\numeral m=\numeral n$ then still $\numeral m \converts t' \converts \numeral n$ by symmetry and transitivity of $\converts$.
        \item if $\phi$ is $\pnat {\numeral n}$ then still $t' \converts \numeral n$ by symmetry and transitivity of $\converts$.
        \item if $\phi$ is $A\numeral n $ then still $t' A n$ by definition of candidate.
         \item if $\phi$ is $\phi_0 \land \phi_1$ then we have $\projl t \, \realises\, \phi_0$ and $\projr t\, \realises\, \phi_1$, and so by IH and context-closure of $\converts$ we have $\projl t' \, \realises\, \phi_0$ and $\projr t'\, \realises\, \phi_1$, thus indeed $t'\, \realises\, \phi$.
        \item if $\phi$ is $\phi_0 \limp \phi_1$ then for any $s\, \realises\, \phi_0$ we have $ts \, \realises \, \phi_1$, and so also $t's \, \realises \, \phi_1$ by IH and context-closure of $\converts$. 
        Since choice of $s$ was arbitrary we are done.
        \item if $\phi $ is $\exists x\phi'(x)$ then there is some $n\in \Nat $ with $ t \, \realises\, \phi'(\numeral n)$. 
        So by IH we have $ t' \, \realises\, \phi'(\numeral n)$, and so also $t'\, \realises\, \phi $.
        \item if $\phi$ is $\forall x \phi'(x)$ then for each $n\in \Nat$ we have $t\, \realises\, \phi'(\numeral n)$, and so also $t' \, \realises\, \phi'(\numeral n)$ by IH. Since choice of $n$ was arbitrary we are done.
        \item if $\phi$ is $\mu\phi'\, \numeral n$ then, for all candidates $A$ we have $t \, \realises \, \forall x (\phi'(A,x) \limp Ax) \limp A\numeral n$, and so also $t' \, \realises \, \forall x (\phi'(A,x) \limp Ax) \limp A\numeral n$ by IH.
        Since choice of $A$ was arbitrary we are done. \qedhere
    \end{itemize}
\end{proof}

Note that, despite its simplicity, the above lemma has the following consequence, by consideration of the candidate $\{(t,n) : t \, \realises\, \psi(n)\}$:
\begin{lemma}
\label{prop:mu-realiser-realises-formula-invariants}
    If $t \, \realises\, \mu \phi\, \numeral n$ then, for any formula $\psi(x)$, $t \, \realises\, \forall x (\phi(\psi,x) \limp \psi(x)) \limp \psi(\numeral n)$.
\end{lemma}

This allows us to realise all the induction axioms for fixed points:
\begin{proposition}
    [Realising Induction]
    \label{prop:iter-realises-ind}
    Let $\sigma(X) = \type {\phi(X,x)}$ and $\tau = \type{\psi(x)}$. 
    Then  $\iter{\sigma,\tau}\, \realises\, \indaxiom_{\phi,\psi}$.
\end{proposition}
\begin{proof}
    We shall omit subscripts to lighten the syntax.
    Let $u\, \realises\, \forall x (\phi(\psi,x) \limp \psi(x))$ and $v \, \realises\, \mu\phi\, \numeral n$.
    We need to show that $\iter{}\, u\, v\, \realises\, \psi(n)$:
    \[
    \begin{array}[b]{r@{\ \realises\ }ll}
         v & \forall x (\phi(\psi,x) \limp \psi(x)) \limp \psi(\numeral n) & \text{since $v\, \realises\, \mu\phi\, \numeral n$ and by \Cref{prop:mu-realiser-realises-formula-invariants}} \\
         v\, u & \psi(\numeral n) & \text{by $\realises\limp$ since $u\, \realises \, \forall x (\phi(\psi,x) \limp \psi(x))$}\\
         \iter\, u\, v & \psi(\numeral n) & \text{by $\converts$ and \Cref{lem:realises-closed-under-converts}}  \tag*{\qed}
    \end{array}
    \]
\renewcommand{\qed}{}
\end{proof}

\subsection{Functoriality and realising prefix axioms}

As expected we will realise $\preaxiom$ by $\inj$.
For this we need a `functoriality lemma' establishing a semantics for functors of $\muLJ$ within our realisability model.
This is because the reductions for $\inj$ introduce functors.

In fact, due to the fact that we admit inductive types as primitive, with native $\lr\mu$ and $\rr \mu$ rules, we shall have to establish the realisability of $\preaxiom$ and the functoriality properties \emph{simultaneously}.

\begin{lemma}
    [Functoriality and realising $\preaxiom$]
    \label{lem:funct-and-pre-realisation}
    Let $\type{\phi(\vec Y,Z,z)} = \sigma(\vec Y,Z)$.
    \todonew{G: Write simply $\mu\phi(\vec Y)$ for $\mu Z\lambda z \phi(\vec Y,Z,z)$}
    Then for all $\vec \tau = \type{\vec \psi(x)}$ we have:
    \begin{enumerate}
        \item\label{item:in-realises-pre} $\injX{\sigma(\vec \tau)} \, \realises\, \preaxiom_{\phi(\vec \psi)} $
        \item\label{item:funct-lem} Let $t \, \realises\, \forall z (\chi(z) \limp \chi'(z))$ \todonew{G: with $\type{\chi(z)} = \gamma$ and $\type{\chi'(z)} = \gamma'$}. Then:
        \todonew{G: THE RED TEXT CAN BE MOVED IN THE PROOF}
        \begin{enumerate}
            \item\label{item:funct-lem-pos} if $\phi(\vec Y,Z,z)$ is positive in $Z$ then $\sigma(\vec \tau,t)\, \realises\, \forall z (\phi(\vec \psi,\chi,z) \limp \phi (\vec\psi , \chi',z))$.
            \item\label{item:funct-lem-neg}
            if $\phi(\vec Y,Z,z)$ is negative in $Z$ then $\sigma(\vec \tau,t)\, \realises\, \forall z (\phi(\vec \psi,\chi',z) \limp \phi (\vec\psi , \chi,z))$.
        \end{enumerate}
    \end{enumerate}
\end{lemma}

Before proving this, let us make some pertinent remarks:
\begin{remark}
    [Comparison to $\F$]
    For the reader used to fixed points via second-order encodings, it is perhaps surprising that functoriality and realisability of $\preaxiom$ are mutually dependent. 
    To recall, in system $\F$, we usually set $\mu X\sigma(X) \dfn \forall X ((\sigma(X) \to X) \to X)$. 
    From here we might set the appropriate functor as,
    \[
        \mu X \sigma(X,t) \dfn \Lambda X ((\sigma(X,t)\to X)\to X)
    \]
    whence the appropriate functoriality properties, analogous to \eqref{item:funct-lem} of \Cref{lem:funct-and-pre-realisation}, are readily established directly, without appeal to the realisation of $\preaxiom$.
    However the above is not (explicitly) how $\muLJ$ defines the fixed point functors.
    
    While a second-order calculus introduces $\forall$ by way of an eigenvariable representing an \emph{arbitrary} prefixed point, our $\rr\mu$ rule works precisely because $\mu X\sigma(X)$ is the \emph{least} fixed point, and so is already sufficiently general.
    This can be realised by appropriate proof transformations, themselves relying on the $\rr\mu$ rule.
    It is for this reason that the mutual dependency above presents in our setting.
\end{remark}

\begin{proof}
[Proof of \Cref{lem:funct-and-pre-realisation}]
    We proceed by induction on the structure of $\phi(\vec Y,Z,z)$, proving both \eqref{item:in-realises-pre} and \eqref{item:funct-lem} simultaneously.
    More precisely, \eqref{item:in-realises-pre} will rely on instances of \eqref{item:funct-lem} for the same $\phi$, and \eqref{item:funct-lem} will rely on smaller instances of both \eqref{item:in-realises-pre} and \eqref{item:funct-lem}.

    We have $\preaxiom_{\phi(\vec \psi)} = \forall z (\phi(\vec \psi,\mu\phi(\vec \psi),z) \limp \mu\phi(\vec \psi)\, z)$.
    So to prove \eqref{item:in-realises-pre} let, 
    \begin{equation}
        \label{eq:in-realises-pre-lfp-unfolded-realiser}
        u\, \realises \, \phi(\vec \psi,\mu\phi(\vec \psi),\numeral n)
    \end{equation}
    and, by $\realises\forall $ and $\realises\limp$, we need to show $\injX{\sigma(\vec \tau)}\, u\, \realises\, \mu\phi(\vec \psi)\, \numeral n$.
    Now, by $\realises\mu$ and $\realises\limp$, let $C$ be a candidate and, 
    \begin{equation}
    \label{eq:in-realises-pre-candidate-prefix-realiser}
        v\, \realises \, \forall z (\phi(\vec \psi,C,z) \limp Cz)
    \end{equation} 
    so that it suffices to show $ \injX{\sigma(\vec \tau)}\, u\, v\, \realises\, C\numeral n$.
    We have:\footnote{The types for $\iter{}$ are omitted but determined by context.}
    \[
    \begin{array}{r@{\ \realises \ }ll}
         \iter{}\, v & \forall z (\mu\phi(\vec \psi)\, z \limp Cz) & \text{by \Cref{prop:iter-realises-ind}, \eqref{eq:in-realises-pre-candidate-prefix-realiser} and $\realises\limp$} \\
         \sigma (\vec \tau,\iter{} \, v ) & \phi(\vec \psi,\mu\phi(\vec \psi),\numeral n) \limp \phi(\vec \psi,C,\numeral n) & \text{by IH for \cref{item:funct-lem} and $\realises \forall$} \\
         \sigma (\vec \tau,\iter{} \, v ) \, u & \phi(\vec \psi,C,\numeral n) & \text{by \eqref{eq:in-realises-pre-lfp-unfolded-realiser} and $\realises \limp$} \\
         v (\sigma (\vec \tau,\iter{} \, v ) \, u) & C\numeral n & \text{by \eqref{eq:in-realises-pre-candidate-prefix-realiser} and $\realises \forall $ and $\realises\limp$} \\
         \injX{\sigma(\vec \tau)}\, u\, v & C\numeral n & \text{by $\converts \inj$ and \Cref{lem:realises-closed-under-converts}}
    \end{array}
    \]

    To prove \eqref{item:funct-lem} the critical case is when $\phi(\vec Y,Z,z) $ is a fixed point formula $ \numeral m \in \mu X \lambda x \phi'(\vec Y,\allowbreak Z,z,X,x)$.
    We consider only the case when $Z$ is positive, \eqref{item:funct-lem-pos}.
    As before we shall simply write, say $\mu\phi'(\vec Y,Z,z)$ for $\mu X\lambda x \phi'(\vec \psi, Z,z,X,x)$.
    To prove \eqref{item:funct-lem-pos} let $n \in \Nat$ and
    \begin{equation}
    \label{eq:funct-lem-mu-antecedent}
        u\ \realises\ \phi(\vec \psi,\chi,\numeral n)
        \qquad\qquad
        \text{(i.e.\ $u\, \realises\, \mu\phi'(\vec \psi,\chi,\numeral n)\, \numeral m$)}
    \end{equation}
    so, by $\realises\forall$ and $\realises\limp$ we need to show $\sigma(\vec \tau,t)\, u\,  \realises\, \phi(\vec \psi,\chi',\numeral n)$, i.e.\ $ \sigma(\vec \tau,t)\, u\,  \realises\,  \mu\phi'(\vec \psi,\chi',\numeral n)\, \numeral m$.
    So, by $\realises\mu$, let $A$ be a candidate and:
    \begin{equation}
    \label{eq:funct-lem-mu-succedent-input}
        v\ \realises\ \forall x (\phi'(\vec \psi,\chi',\numeral n, A,x) \limp Ax)
    \end{equation}
    We need to show that $\sigma(\vec \tau,t)\, u\, v\, \realises\, A\numeral m$.

    Now let $\sigma'(\vec Y,Z,X) = \type {\phi'(\vec Y,Z,z,X,x)}$. 
    Again we may similarly write simply $\mu\sigma'(\vec Y,Z)$ for $\mu X \sigma'(\vec Y,Z,X)$ (which is just $\sigma(\vec Y,Z)$).
    First, by IH\eqref{item:in-realises-pre} for $\sigma'(\vec Y,Z,X)$, we have:
    \begin{equation}
    \label{eq:in-realises-pre-ih}
        \injX{\sigma'(\vec \tau,\gamma')} \, \realises\, \forall x (\phi'(\vec \psi, \chi',\numeral n, \mu\phi'(\vec \psi,\chi',\numeral n),x) \limp \mu\phi'(\vec \psi,\chi',\numeral n)\, x)
    \end{equation}
    We claim that:
    \begin{equation}
    \label{eq:mur-realises-pre-with-context}
    \rr\mu\, \sigma'(\vec \tau,t,\sigma(\vec \tau,\gamma'))\ \realises\ 
    \forall x (\phi'(\vec \psi,\chi,\numeral n,\mu\phi'(\vec \psi,\chi',\numeral n),x) \limp \mu\phi'(\vec \psi,\chi',\numeral n)\, x)
    \end{equation}
    To prove this let $w\, \realises\, \phi'(\vec \psi, \chi,\numeral n,\mu\phi'(\vec \psi,\chi',\numeral n),\numeral k)$ and we show $ \rr\mu\, \sigma'(\vec \tau,t,\sigma(\vec \tau,\gamma'))\, w \, \realises\, \mu\phi'(\vec \psi,\allowbreak \chi',\numeral n)\, \numeral k$:
    \[
    \begin{array}{r@{\ \realises\ }ll}
         \sigma'(\vec \tau,t,\mu X \sigma(\vec \tau,\gamma')) \, w & \phi'(\vec \psi,\chi',\numeral n, \mu\phi'(\vec \psi, \chi',\numeral n), \numeral k) &\text{by IH\eqref{item:funct-lem} for $\phi'$, $\realises\forall $ and $\realises\limp$} \\
         \injX{\sigma'(\vec \tau,\gamma')} (\sigma'(\vec \tau,t,\mu X \sigma(\vec \tau,\gamma')) \, w) & \mu\phi'(\vec \psi,\chi',\numeral n)\, \numeral k & \text{by \eqref{eq:in-realises-pre-ih} and $\realises\forall$ and $\realises\limp$} \\
         \mu_r\, \sigma'(\vec \tau, t, \sigma(\vec \tau,\gamma'))\, w & \mu\phi'(\vec \psi, \chi', \numeral n)\, \numeral k & \text{by $\converts$}
    \end{array}
    \]
    Now by \eqref{eq:mur-realises-pre-with-context}, $\realises\mu$ and \Cref{prop:mu-realiser-realises-formula-invariants} we have,
    \begin{equation}
        u\, \realises\, \forall x (\phi'(\vec \psi,\chi,\numeral n,\mu\phi'(\vec \psi,\chi',\numeral n),x) \limp \mu\phi'(\vec \psi,\chi',\numeral n)\, x) \limp \mu\phi'(\vec \psi,\chi',\numeral n)\, \numeral m
    \end{equation}

    Now we prove $\sigma(\vec \tau,t)\, u\, v\, \realises\, A\numeral m$ as follows:
    \[
    \begin{array}{r@{\ \realises \ }ll}
          u\, (\rr\mu\, \sigma'(\vec \tau,t,\sigma(\vec \tau,\gamma'))) & \mu\phi'(\vec \psi,\chi',\numeral n)\, \numeral m & \text{by \eqref{eq:mur-realises-pre-with-context} and $\realises\limp$}\\
          \iter{}\, (\rr\mu\, \sigma'(\vec \tau,t,\sigma(\vec \tau,\gamma')))\, u & \mu\phi'(\vec \psi,\chi',\numeral n)\, \numeral m & \text{by $\converts$} \\
          \iter{}\, (\rr\mu\, \sigma'(\vec \tau,t,\sigma(\vec \tau,\gamma')))\, u\, v & A \numeral m & \text{by $\realises \mu$ and \eqref{eq:funct-lem-mu-succedent-input}} \\
          \mu \sigma'(\vec \tau,t) \, u \, v & A\numeral m & \text{by defintion of $\mu$-functor} \\
          \sigma(\vec \tau,t)\, u\, v & A\numeral m
    \end{array}
    \]

    The remaining steps for functoriality, \eqref{item:funct-lem}, are routine, and for these we shall simply suppress the parameters $\vec \psi,\vec \tau$ henceforth, writing simply $\phi(Z,z)$ instead of $\phi(\vec Y,Z,z)$ and $\sigma(Z)$ instead of $\sigma(\vec Y,Z)$.
    This is because we shall never vary the parameters $\vec \psi,\vec \tau$ in the remaining cases.
    We give only the cases for \cref{item:funct-lem-pos}, the ones for \cref{item:funct-lem-neg} being similar.
     We need to show that $\sigma(t) \, \realises\, \forall z (\phi(\chi,z) \limp \phi(\chi',z))$ so let $n\in \Nat$ and $u\, \realises\, \phi(\chi,\numeral n)$ and let us show that $\sigma(t)\, u\,  \realises\, \phi(\chi',\numeral n)$, under $\realises\forall$ and $\realises\limp$.

    \begin{itemize}
        \item If $\phi(Z,z)=Zs(z)$ then $\sigma(t) = t $. So we need $t\, \realises\, \forall z(\chi(s(z)) \limp \chi'(s(z))) $. Let $n\in \Nat$ and, since already $t\, \realises\, \forall z (\chi(z) \limp \chi'(z))$ by assumption, we have $t\, \realises \chi(\numeral{s(n)}) \limp \chi(\numeral {s(n)})$ as required.
        \item If $\phi(Z,z)$ is any other atomic formula $\theta$ then $\forall z (\phi(\chi,z)\limp \phi(\chi',z))$ is just $\forall z (\theta \limp \theta)$ and $\sigma(t)$ is just $\id\, \realises \, \forall z (\theta \limp \theta)$ by $\realises\forall$ and $\id\converts$.
        \item Suppose $\phi(Z,z)$ is $\phi_0(Z,z) \land \phi_1(Z,z)$. 
        Write $\sigma_i(Z) \dfn \type {\phi_i(Z,z)}$ for $i\in \{0,1\}$ so that $\sigma(Z) = \sigma_0(Z) \times \sigma_1(Z)$.
        We need to show that $\sigma(t)\, \realises\, \forall z (\phi(\chi,z) \limp \phi(\chi',z))$, so let $n\in \Nat$ and $u\, \realises\, \phi(\chi,\numeral n)$ and let us show that $\sigma(t)\, u \, \realises \, \phi(\chi',\numeral n)$:
        \[
        \begin{array}{r@{\ \realises\ }ll}
             \proj i u & \phi_i(\chi,\numeral n) & \text{for $i\in \{0,1\}$ by $\realises\land$} \\
             \sigma_i(t) (\proj i u) & \phi_i(\chi', \numeral n) & \text{for $i\in \{0,1\}$ by IH and $\realises\forall$ and $\realises\limp$} \\
             \proj i (\sigma(t)\, u) & \phi_i(\chi',\numeral n) &\text{for $i\in \{0,1\}$ by form of $\sigma(t)$ and $\converts$} \\
             \sigma(t)\, u  & \phi(\chi',\numeral n) & \text{by $\realises\land$}
        \end{array}
        \]
        \item Suppose $\phi(Z,z)$ is $\phi_0(Z,z) \limp \phi_1(Z,z)$ with $\phi_0(Z,z)$ negative in $Z$ and $\phi_1(Z,z)$ positive in $Z$. 
        Write $\sigma_i(Z)\dfn \type{\phi_i(Z,z)}$ for $i\in \{0,1\}$ so that $\sigma(Z) = \sigma_0(Z) \limp \sigma_1(Z)$.
        We need to show that $\sigma(t)\, \realises\, \forall z (\phi(\chi,z) \limp \phi(\chi',z))$ so let $n\in \Nat $ and $u\, \realises\, \phi(\chi,\numeral n)$ and let us show that $\sigma(t)\, u \, \realises\, \phi(\chi',\numeral n)$.
        By the form of $\phi(Z,z)$, and by $\realises\limp$, let $v\, \realises\, \phi_0(\chi',\numeral n)$ so that it suffices to show $\sigma(t)\, u\, v\, \realises\, \phi_1(\chi', \numeral n)$:
        \[
        \begin{array}{r@{\ \realises\ }ll}
             \sigma_0(t)\, v & \phi_0 (\chi,\numeral n) & \text{by IH and $\realises\forall $ and $\realises\limp$} \\
             u\, (\sigma_0(t)\, v) & \phi_1(\chi,\numeral n) & \text{by $\realises\limp$} \\
             \sigma_1(t)\, (u\, (\sigma_0(t)\, v)) & \phi_1(\chi',\numeral n) & \text{by IH and $\realises\forall$ and $\realises\limp$} \\
             \sigma(t)\, u\, v & \phi_1(\chi',\numeral n) & \text{by form of $\sigma(t)$ and $\converts$}
        \end{array}
        \]
        
        \item Suppose $\phi(Z,z) $ is $\exists y \phi'(Z,z,y)$. Write $\sigma'(Z)\dfn \type{\phi'(Z,z,y)}$ so that $\sigma(Z) = \sigma'(Z)$.
        We need to show that $\sigma(t) \, \realises\, \forall z (\phi(\chi,z) \limp \phi(\chi',z))$ so let $n\in \Nat$ and $u\, \realises\, \phi(\chi,\numeral n)$ and let us show that $\sigma(t)\, u\,  \realises\, \phi(\chi',\numeral n)$, under $\realises\forall$ and $\realises\limp$. 
        \[
        \begin{array}{r@{\ \realises\ }ll}
             u & \exists y \phi'(\chi,\numeral n,y) & \text{by assumption and form of $\phi(Z,z)$} \\
             u & \phi' (\chi,\numeral n, \numeral k) & \text{for some $k\in \Nat$ by $\realises\exists$} \\
             \sigma'(t)\, u & \phi'(\chi',\numeral n, \numeral k) & \text{for some $k\in \Nat$ by IH and $\realises\forall$ and $\realises\limp$} \\
             \sigma'(t)\, u & \exists y \phi'(\chi', \numeral n, y) & \text{by $\realises\exists$} \\
             \sigma(t)\, u & \phi(\chi',\numeral n) & \text{by forms of $\sigma(Z)$ and $\phi(Z,z)$}             
        \end{array}
        \]
        \item Suppose $\phi(Z,z) $ is $\forall y \phi'(Z,z,y)$. Write $\sigma'(Z) \dfn \type{\phi'(Z,z,y)}$ so that $\sigma(Z) = \sigma'(Z)$.
        We need to show that $\sigma(t)\, \realises\, \forall z (\phi(\chi,z) \limp \phi(\chi',z))$ so let $n\in \Nat $ and $u\, \realises\, \phi(\chi,\numeral n)$ and let us show that $\sigma(t)\, u\, \realises\, \phi(\chi',\numeral n)$ under $\realises\forall$ and $\realises\limp$.
        \[
        \begin{array}[b]{r@{\ \realises\ }ll}
             u & \forall y \phi'(\chi,\numeral n,y) & \text{by assumption and form of $\phi(Z,z)$} \\
             u & \phi'(\chi,\numeral n, \numeral k) & \text{for all $k\in \Nat$ by $\realises \forall$} \\
             \sigma'(t)\, u & \phi'(\chi',\numeral n, \numeral k) & \text{for all $k\in \Nat$ by IH and $\realises\forall$ and $\realises\limp$} \\
             \sigma'(t)\, u & \forall y \phi'(\chi',\numeral n, y) & \text{by $\realises\forall$} \\
             \sigma(t)\, u & \phi(\chi', \numeral n) & \text{by forms of $\sigma(Z)$ and $\phi(Z,z)$} \tag*{\qed}
        \end{array}
        \]
    \end{itemize}
    \renewcommand{\qed}{}
\end{proof}

\subsection{Putting it all together}

\begin{theorem}
    [Soundness]
    \label{thm:realisability-soundness}
    If $\muHAneg \proves \phi$ then there is a typed $\typedclomuLJneg$ term $t:\type \phi$ such that $t\, \realises\, \phi$.
\end{theorem}
\begin{proof}
The axioms and rules of intuitionistic predicate logic are realised as usual, not using the fixed point rules. 
The basic arithmetical axioms of $\muHAneg$ are just universal closures of true equations between primitive recursive terms, which are realised by the corresponding derivations of those terms in $\muLJneg$.\footnote{Recall that $\muLJneg$ represents at least all the (even higher-order) primitive recursive functions.}
The axiom schemes $\preaxiom$ and $\indaxiom$ are realised by $\inj$ and $\iter{}$ of appropriate types, by \Cref{lem:funct-and-pre-realisation} and \Cref{prop:iter-realises-ind} respectively.
    It remains to consider the $\pNat$-specific axioms,  which are realised by the analogous $\nat$-specific rules of $\muLJneg$:
    \begin{itemize}
        \item $\numeral 0 \, \realises\, \nprezero$, as this is just $\numeral 0\, \realises\, \pnat {\numeral 0}$, which follows by definition of $\realises\pNat$.
        \item $\succ \, \realises\, \npresucc$. Let $n\in \Nat $ and $u\, \realises\, \pnat {\numeral n}$, i.e.\ $u\converts \numeral n$, so that, by $\realises\forall $ and $\realises\limp$, it suffices to show that $\succ u\, \realises\, \pnat{\succ \numeral n}$.
        We have that $\succ u \converts \numeral{n+1} = \succ \numeral n $, and so indeed $\succ u\, \realises\, \pnat{\succ \numeral n}$ as required.
        \item $\Niter\, \realises \, \nindaxiom$.
        Let $u\, \realises\, \phi(\numeral 0)$ and $v\, \realises\, \forall x (\phi(x)\limp \phi(\succ x ))$ so that it suffices, by $\realises\limp$, to show that $\Niter u\, v\, \realises \, \forall x^\pNat \phi(x)$.
        For this let $n\in \Nat $ and $w\, \realises\, \pnat {\numeral n}$, and we show $\Niter u\, v\, w\, \realises\, \phi(\numeral n)$ by induction on $n$:
        \begin{itemize}
            \item If $n=0$ then $w\converts \numeral 0$ so $\Niter u\, v\, w\, \converts\, \Niter u\, v\, 0\, \converts u\, \realises \, \phi(\numeral 0)$ by assumption.
            \item If $n = n'+1$ then we have:
            \[
            \begin{array}[b]{r@{\ \realises\ }ll}
                 \Niter u\, v\, \numeral n' & \phi(\numeral n') & \text{by IH} \\
                 v\, (\Niter u\, v\, \numeral n') & \phi(\succ{\numeral n'}) & \text{by $\realises\forall$ and $\realises\limp$} \\
                 \Niter u\, v\, w & \phi(\numeral n) & \text{by $\converts$} \tag*{\qed}
            \end{array}
            \]
        \end{itemize}
    \end{itemize}
    \renewcommand{\qed}{}
\end{proof}

\begin{corollary}
    Any provably total recursive function of $\muPA$ is representable in $\muLJneg$.
\end{corollary}
\begin{proof}
    Assume $\muPA \proves \forall \vec x \exists ! y \, s(\vec x, y) = t(\vec x, y)$, without loss of generality.\footnote{Recall we have included each primitive recursive function definition as a symbol of $\langarith$ and defining equations among the axioms of $\PA$, $\HA$ and extensions.} 
    So we have,
    \[
    \begin{array}{rll}
          & \muPA \proves \forall \vec x \exists  y \, s(\vec x,y) = t(\vec x,y) \\
          \implies & \muHA \proves \forall \vec x \exists  y \, s(\vec x,y) = t(\vec x,y) & \text{by \Cref{thm:muPA-pi02-cons-muHA}} \\
           \implies  & \muHAneg \proves \forall \vec x^\pNat \exists y^\pNat s(\vec x,y) = t(\vec x,y) & \text{by \Cref{prop:muHA-to-muHAneg}} \\
          \implies  &  u\, \realises\,  \forall \vec x^\pNat \exists y^\pNat s(\vec x,y) = t(\vec x,y) 
        \end{array}
          \]
          for some $u$ of $ \typedclomuLJneg$ (of appropriate type) by \Cref{thm:realisability-soundness}.
                    Thus:
          \[
              \begin{array}{rll}
          & \forall \vec m \in \Nat\ u\, \vec {\numeral m}\, \realises \, \exists y^\pNat s(\vec {\numeral m},y) = t(\vec {\numeral m},y) & \text{by $\realises \forall$, $\realises \limp$ and $\realises\pNat$} \\
          \implies & \forall \vec{ m} \in \Nat \exists n \in \Nat\, (\projl (u\, \vec{\numeral m})\, \realises\, \pnat {\numeral n} \text{ and } \projr (u\, \vec {\numeral m}) \, \realises \, s(\vec {\numeral m}, \numeral n) = t(\vec {\numeral m}, \numeral n) ) & \text{by $\realises \exists$, $\realises\land$ and $\realises\pNat$} \\
          \implies & \forall \vec{ m} \in \Nat \exists n \in \Nat\, (\projl (u\, \vec {\numeral m}) \converts \numeral n \text{ and } s(\vec {\numeral m},\numeral n) \converts t(\vec {\numeral m}, \numeral n)) &\text{by $\realises\pNat$ and $\realises=$}
    \end{array}
    \]
    Thus, as the graph of $s(\vec x,y) = t(\vec x,y)$, in the standard model $\nmod$ is functional, by assumption, it is computed by $ \lambda \vec x^{\vec \nat} (\projl (u\, \vec x))$.
By \Cref{thm:terms-to-derivations} this is equivalent to some $\muLJneg$ derivation too.
\todonew{G: WE COULD REPLACE $t(\vec x, y)$ with $0$.}
\end{proof}
 \section{Conclusions}
\label{sec:concs}

In this work we investigated the computational expressivity of fixed point logics.
Our main contribution is a characterisation of the functions representable in the systems $\muLJ$ and $\cmuLJ$ as the functions provably recursive in the second-order theory $\PSCA$.
This extends the tradition of such correspondences between arithmetic theories an type systems, in particular including between $\PA$ and system $\T$ due to G\"odel \cite{goedel1958bisher}, and between $\SOPA$ and $\F$ due to Girard \cite{girard1972systemF}.
In a sense ours is a somewhat intuitive result in light of M\"ollerfeld's work \cite{Moellerfeld02:phd-thesis}, identifying the latter's arithmetical theorems with the extension of $\ACA$ by least and greatest general fixed points.

Our characterisation also applies to aforementioned (circular) systems of linear logic, namely $\muMALL$ and its circular counterpart, from \cite{BaeldeMiller07,Baelde12,BaeldeDS16Infinitaryprooftheory,EhrJaf21:cat-models-muLL,EhrJafSau21:totality-CmuMALL,BDKS22:bouncing-threads,DeSau19:infinets1,DePelSau:infinets2,DeJafSau22:phase-semantics}.
This result was given the preliminary conference version of this paper \cite{CD23:muLJ-lics}, but will be expanded upon in a separate full paper.

Referring to Rathjen's ordinal notation system in \cite{Rathjen95:recent-advances}, this means that all these systems represent just the functions computable by recursion on ordinals of $\mathfrak T[\theta = \omega]$.
To this end we used a range of techniques from proof theory, reverse mathematics, higher-order computability and metamathematics.
This contribution settles the question of computational expressivity of (circular) systems with fixed points, cf.~Question~\ref{question:main}.

In future work it would be interesting to investigate the computational expressivity of systems with only \emph{strictly} positive fixed points, where $\mu,\nu$ may only bind variables that are \emph{never} under the left of $\arrow$.
We suspect that such systems are computationally weaker than those we have considered here.
On the arithmetical side, it would be interesting to investigate the \emph{higher-order} reverse mathematics of the Knaster-Tarski theorem, cf.~\cite{SatoYamazaki17:rev-math-fixed-point-thms-with-kt} and the present work.

\comment{
\begin{itemize}
    \item what about strict positivity? conjecture weaker.
    \item we could first translate into a well-founded system, but this will not (immediately) give us effective ordinal bounds for formalisation. (Actually this is an interesting idea)
\end{itemize}
}

\section*{Acknowledgements}

The authors would like to thank the anonymous referees for their diligent work in reviewing this paper, which has surely improved its presentation.
The authors would like to thank Igor Walukiewicz, Alexis Saurin, Graham Leigh, Pierre Clairambault, Colin Riba, Ulrich Berger and Paul Levy for several insightful conversations around this work.

\bibliographystyle{alphaurl}
\bibliography{main}

\begin{thebibliography}{KMPS19}

\bibitem[AF98]{avigad1998godel}
Jeremy Avigad and Solomon Feferman.
\newblock G{\"o}del’s functional (“dialectica”) interpretation.
\newblock {\em Handbook of proof theory}, 137:337--405, 1998.

\bibitem[Bae12]{Baelde12}
David Baelde.
\newblock Least and greatest fixed points in linear logic.
\newblock {\em {ACM} Trans. Comput. Log.}, 13(1):2:1--2:44, 2012.
\newblock \href {https://doi.org/10.1145/2071368.2071370} {\path{doi:10.1145/2071368.2071370}}.

\bibitem[BDKS22]{BDKS22:bouncing-threads}
David Baelde, Amina Doumane, Denis Kuperberg, and Alexis Saurin.
\newblock Bouncing threads for circular and non-wellfounded proofs: Towards compositionality with circular proofs.
\newblock In Christel Baier and Dana Fisman, editors, {\em {LICS} '22: 37th Annual {ACM/IEEE} Symposium on Logic in Computer Science, Haifa, Israel, August 2 - 5, 2022}, pages 63:1--63:13. {ACM}, 2022.
\newblock \href {https://doi.org/10.1145/3531130.3533375} {\path{doi:10.1145/3531130.3533375}}.

\bibitem[BDS16]{BaeldeDS16Infinitaryprooftheory}
David Baelde, Amina Doumane, and Alexis Saurin.
\newblock Infinitary proof theory: the multiplicative additive case.
\newblock In Jean{-}Marc Talbot and Laurent Regnier, editors, {\em 25th {EACSL} Annual Conference on Computer Science Logic, {CSL} 2016, August 29 - September 1, 2016, Marseille, France}, volume~62 of {\em LIPIcs}, pages 42:1--42:17. Schloss Dagstuhl - Leibniz-Zentrum f{\"{u}}r Informatik, 2016.
\newblock \href {https://doi.org/10.4230/LIPIcs.CSL.2016.42} {\path{doi:10.4230/LIPIcs.CSL.2016.42}}.

\bibitem[BM07]{BaeldeMiller07}
David Baelde and Dale Miller.
\newblock Least and greatest fixed points in linear logic.
\newblock In Nachum Dershowitz and Andrei Voronkov, editors, {\em Logic for Programming, Artificial Intelligence, and Reasoning, 14th International Conference, {LPAR} 2007, Yerevan, Armenia, October 15-19, 2007, Proceedings}, volume 4790 of {\em Lecture Notes in Computer Science}, pages 92--106. Springer, 2007.
\newblock \href {https://doi.org/10.1007/978-3-540-75560-9\_9} {\path{doi:10.1007/978-3-540-75560-9\_9}}.

\bibitem[BT21]{BerTsu21:ifp}
Ulrich Berger and Hideki Tsuiki.
\newblock Intuitionistic fixed point logic.
\newblock {\em Annals of Pure and Applied Logic}, 172(3):102903, 2021.
\newblock \href {https://doi.org/10.1016/j.apal.2020.102903} {\path{doi:10.1016/j.apal.2020.102903}}.

\bibitem[CD23]{CD23:muLJ-lics}
Gianluca Curzi and Anupam Das.
\newblock Computational expressivity of (circular) proofs with fixed points.
\newblock In {\em {LICS}}, pages 1--13, 2023.
\newblock \href {https://doi.org/10.1109/LICS56636.2023.10175772} {\path{doi:10.1109/LICS56636.2023.10175772}}.

\bibitem[Cla09]{Clairambault09}
Pierre Clairambault.
\newblock Least and greatest fixpoints in game semantics.
\newblock In Ralph Matthes and Tarmo Uustalu, editors, {\em 6th Workshop on Fixed Points in Computer Science, {FICS} 2009, Coimbra, Portugal, September 12-13, 2009}, pages 39--45. Institute of Cybernetics, 2009.
\newblock URL: \url{http://cs.ioc.ee/fics09/proceedings/contrib5.pdf}.

\bibitem[Cla13]{Clairambault13interleaving}
Pierre Clairambault.
\newblock Strong functors and interleaving fixpoints in game semantics.
\newblock {\em {RAIRO} Theor. Informatics Appl.}, 47(1):25--68, 2013.
\newblock \href {https://doi.org/10.1051/ita/2012028} {\path{doi:10.1051/ita/2012028}}.

\bibitem[Das20a]{Das21:CT-preprint}
Anupam Das.
\newblock A circular version of {G}{\"{o}}del's {T} and its abstraction complexity.
\newblock {\em CoRR}, abs/2012.14421, 2020.
\newblock URL: \url{https://arxiv.org/abs/2012.14421}, \href {https://arxiv.org/abs/2012.14421} {\path{arXiv:2012.14421}}.

\bibitem[Das20b]{Das20:ca-log-comp}
Anupam Das.
\newblock {On the logical complexity of cyclic arithmetic}.
\newblock {\em {Logical Methods in Computer Science}}, {Volume 16, Issue 1}, January 2020.
\newblock \href {https://doi.org/10.23638/LMCS-16(1:1)2020} {\path{doi:10.23638/LMCS-16(1:1)2020}}.

\bibitem[Das21]{Das21:CT-fscd}
Anupam Das.
\newblock On the logical strength of confluence and normalisation for cyclic proofs.
\newblock In Naoki Kobayashi, editor, {\em 6th International Conference on Formal Structures for Computation and Deduction, {FSCD} 2021, July 17-24, 2021, Buenos Aires, Argentina (Virtual Conference)}, volume 195 of {\em LIPIcs}, pages 29:1--29:23. Schloss Dagstuhl - Leibniz-Zentrum f{\"{u}}r Informatik, 2021.
\newblock \href {https://doi.org/10.4230/LIPIcs.FSCD.2021.29} {\path{doi:10.4230/LIPIcs.FSCD.2021.29}}.

\bibitem[DJS22]{DeJafSau22:phase-semantics}
Abhishek De, Farzad Jafar{-}Rahmani, and Alexis Saurin.
\newblock Phase semantics for linear logic with least and greatest fixed points.
\newblock In Anuj Dawar and Venkatesan Guruswami, editors, {\em 42nd {IARCS} Annual Conference on Foundations of Software Technology and Theoretical Computer Science, {FSTTCS} 2022, December 18-20, 2022, {IIT} Madras, Chennai, India}, volume 250 of {\em LIPIcs}, pages 35:1--35:23. Schloss Dagstuhl - Leibniz-Zentrum f{\"{u}}r Informatik, 2022.
\newblock \href {https://doi.org/10.4230/LIPIcs.FSTTCS.2022.35} {\path{doi:10.4230/LIPIcs.FSTTCS.2022.35}}.

\bibitem[Dou17]{Doumane17thesis}
Amina Doumane.
\newblock {\em On the infinitary proof theory of logics with fixed points. (Th{\'{e}}orie de la d{\'{e}}monstration infinitaire pour les logiques {\`{a}} points fixes)}.
\newblock PhD thesis, Paris Diderot University, France, 2017.
\newblock URL: \url{https://tel.archives-ouvertes.fr/tel-01676953}.

\bibitem[DPS21]{DePelSau:infinets2}
Abhishek De, Luc Pellissier, and Alexis Saurin.
\newblock Canonical proof-objects for coinductive programming: infinets with infinitely many cuts.
\newblock In Niccol{\`{o}} Veltri, Nick Benton, and Silvia Ghilezan, editors, {\em {PPDP} 2021: 23rd International Symposium on Principles and Practice of Declarative Programming, Tallinn, Estonia, September 6-8, 2021}, pages 7:1--7:15. {ACM}, 2021.
\newblock \href {https://doi.org/10.1145/3479394.3479402} {\path{doi:10.1145/3479394.3479402}}.

\bibitem[DS19]{DeSau19:infinets1}
Abhishek De and Alexis Saurin.
\newblock Infinets: The parallel syntax for non-wellfounded proof-theory.
\newblock In Serenella Cerrito and Andrei Popescu, editors, {\em Automated Reasoning with Analytic Tableaux and Related Methods - 28th International Conference, {TABLEAUX} 2019, London, UK, September 3-5, 2019, Proceedings}, volume 11714 of {\em Lecture Notes in Computer Science}, pages 297--316. Springer, 2019.
\newblock \href {https://doi.org/10.1007/978-3-030-29026-9\_17} {\path{doi:10.1007/978-3-030-29026-9\_17}}.

\bibitem[EJ21]{EhrJaf21:cat-models-muLL}
Thomas Ehrhard and Farzad Jafarrahmani.
\newblock Categorical models of linear logic with fixed points of formulas.
\newblock In {\em Proceedings of the 36th Annual ACM/IEEE Symposium on Logic in Computer Science}, LICS '21, New York, NY, USA, 2021. Association for Computing Machinery.
\newblock \href {https://doi.org/10.1109/LICS52264.2021.9470664} {\path{doi:10.1109/LICS52264.2021.9470664}}.

\bibitem[EJS21]{EhrJafSau21:totality-CmuMALL}
Thomas Ehrhard, Farzad Jafarrahmani, and Alexis Saurin.
\newblock {On relation between totality semantic and syntactic validity}.
\newblock In {\em {5th International Workshop on Trends in Linear Logic and Applications (TLLA 2021)}}, Rome (virtual), Italy, June 2021.
\newblock URL: \url{https://hal-lirmm.ccsd.cnrs.fr/lirmm-03271408}.

\bibitem[Gir72]{girard1972systemF}
Jean-Yves Girard.
\newblock {\em Interpr{\'e}tation fonctionnelle et {\'e}limination des coupures de l'arithm{\'e}tique d'ordre sup{\'e}rieur}.
\newblock PhD thesis, {\'E}diteur inconnu, 1972.

\bibitem[Gö58]{goedel1958bisher}
Kurt Gödel.
\newblock {Über eine bisher noch nicht benützte Erweiterung des finiten Standpunktes}.
\newblock {\em Dialectica}, 12(3-4):280--287, 1958.

\bibitem[Hir05]{hirst05:ord-arith-rm-survey}
Jeffry~L. Hirst.
\newblock A survey of the reverse mathematics of ordinal arithmetic.
\newblock In Stephen~G.Editor Simpson, editor, {\em Reverse Mathematics 2001}, Lecture Notes in Logic, page 222–234. Cambridge University Press, 2005.
\newblock \href {https://doi.org/10.1017/9781316755846.014} {\path{doi:10.1017/9781316755846.014}}.

\bibitem[HS08]{Hindley}
J.~Roger Hindley and Jonathan~P. Seldin.
\newblock {\em Lambda-Calculus and Combinators: An Introduction}.
\newblock Cambridge University Press, USA, 2 edition, 2008.

\bibitem[KMPS19]{KMPS19:buchi-rev-math}
Leszek~Aleksander Kolodziejczyk, Henryk Michalewski, Pierre Pradic, and Michal Skrzypczak.
\newblock The logical strength of {B}{\"{u}}chi's decidability theorem.
\newblock {\em Log. Methods Comput. Sci.}, 15(2), 2019.
\newblock \href {https://doi.org/10.23638/LMCS-15(2:16)2019} {\path{doi:10.23638/LMCS-15(2:16)2019}}.

\bibitem[KMV22]{KupMarVen22:graph-reps-mu-forms}
Clemens Kupke, Johannes Marti, and Yde Venema.
\newblock {Succinct Graph Representations of $\mu$-Calculus Formulas}.
\newblock In Florin Manea and Alex Simpson, editors, {\em 30th EACSL Annual Conference on Computer Science Logic (CSL 2022)}, volume 216 of {\em Leibniz International Proceedings in Informatics (LIPIcs)}, pages 29:1--29:18, Dagstuhl, Germany, 2022. Schloss Dagstuhl -- Leibniz-Zentrum f{\"u}r Informatik.
\newblock \href {https://doi.org/10.4230/LIPIcs.CSL.2022.29} {\path{doi:10.4230/LIPIcs.CSL.2022.29}}.

\bibitem[Koz83]{Koz83:results-on-mu}
Dexter Kozen.
\newblock Results on the propositional $\mu$-calculus.
\newblock {\em Theoretical Computer Science}, 27(3):333--354, 1983.
\newblock Special Issue Ninth International Colloquium on Automata, Languages and Programming (ICALP) Aarhus, Summer 1982.
\newblock \href {https://doi.org/10.1016/0304-3975(82)90125-6} {\path{doi:10.1016/0304-3975(82)90125-6}}.

\bibitem[KPP21]{KuperbergPP21systemT}
Denis Kuperberg, Laureline Pinault, and Damien Pous.
\newblock Cyclic proofs, system {T}, and the power of contraction.
\newblock {\em Proc. {ACM} Program. Lang.}, 5({POPL}):1--28, 2021.
\newblock \href {https://doi.org/10.1145/3434282} {\path{doi:10.1145/3434282}}.

\bibitem[Lub93]{Lubarsky93}
Robert~S. Lubarsky.
\newblock $\mu$-definable sets of integers.
\newblock {\em The Journal of Symbolic Logic}, 58(1):291–313, 1993.
\newblock \href {https://doi.org/10.2307/2275338} {\path{doi:10.2307/2275338}}.

\bibitem[Men87]{Mendler87:recursive-types}
Nax~Paul Mendler.
\newblock Recursive types and type constraints in second-order lambda calculus.
\newblock In {\em Logic in Computer Science}, 1987.

\bibitem[Men91]{Mendler91:inductive-types}
Nax~Paul Mendler.
\newblock Inductive types and type constraints in the second-order lambda calculus.
\newblock {\em Annals of Pure and Applied Logic}, 51(1):159--172, 1991.
\newblock \href {https://doi.org/10.1016/0168-0072(91)90069-X} {\path{doi:10.1016/0168-0072(91)90069-X}}.

\bibitem[Mö02]{Moellerfeld02:phd-thesis}
Michael Möllerfeld.
\newblock {\em Generalized inductive definitions. The {\(\mu\)}-calculus and {\(\Pi^1_2\)}-comprehension}.
\newblock PhD thesis, University of M\"{u}nster, 2002.
\newblock University of M\"unster, \url{https://nbn-resolving.de/urn:nbn:de:hbz:6-85659549572}.

\bibitem[NW96]{NiwWal96:games-for-mu-calc}
Damian Niwinski and Igor Walukiewicz.
\newblock Games for the {\(\mu\)}-calculus.
\newblock {\em Theor. Comput. Sci.}, 163(1{\&}2):99--116, 1996.
\newblock \href {https://doi.org/10.1016/0304-3975(95)00136-0} {\path{doi:10.1016/0304-3975(95)00136-0}}.

\bibitem[PY17]{PengYamazaki17:rev-math-fixed-point-thms-no-kt}
Weiguang Peng and Takeshi Yamazaki.
\newblock Two kinds of fixed point theorems and reverse mathematics.
\newblock {\em Mathematical Logic Quarterly}, 63(5):454--461, 2017.
\newblock \href {https://arxiv.org/abs/https://onlinelibrary.wiley.com/doi/pdf/10.1002/malq.201600096} {\path{arXiv:https://onlinelibrary.wiley.com/doi/pdf/10.1002/malq.201600096}}, \href {https://doi.org/10.1002/malq.201600096} {\path{doi:10.1002/malq.201600096}}.

\bibitem[Rat95]{Rathjen95:recent-advances}
Michael Rathjen.
\newblock Recent advances in ordinal analysis: {$\Pi^1_2$-CA} and related systems.
\newblock {\em Bulletin of Symbolic Logic}, 1(4):468–485, 1995.
\newblock \href {https://doi.org/10.2307/421132} {\path{doi:10.2307/421132}}.

\bibitem[Rey74]{Reynolds74systemF}
John~C. Reynolds.
\newblock Towards a theory of type structure.
\newblock In Bernard~J. Robinet, editor, {\em Programming Symposium, Proceedings Colloque sur la Programmation, Paris, France, April 9-11, 1974}, volume~19 of {\em Lecture Notes in Computer Science}, pages 408--423. Springer, 1974.
\newblock \href {https://doi.org/10.1007/3-540-06859-7\_148} {\path{doi:10.1007/3-540-06859-7\_148}}.

\bibitem[RS22]{sep-proof-theory}
Michael Rathjen and Wilfried Sieg.
\newblock {Proof Theory}.
\newblock In Edward~N. Zalta and Uri Nodelman, editors, {\em The {Stanford} Encyclopedia of Philosophy}. Metaphysics Research Lab, Stanford University, {W}inter 2022 edition, 2022.

\bibitem[Sim99]{Simpson99:monograph}
Stephen~G. Simpson.
\newblock {\em Subsystems of second order arithmetic}.
\newblock Perspectives in mathematical logic. Springer, 1999.

\bibitem[Sim17]{Simpson17:cyc-arith}
Alex Simpson.
\newblock Cyclic arithmetic is equivalent to {P}eano arithmetic.
\newblock In Javier Esparza and Andrzej~S. Murawski, editors, {\em Foundations of Software Science and Computation Structures - 20th International Conference, {FOSSACS} 2017, Held as Part of the European Joint Conferences on Theory and Practice of Software, {ETAPS} 2017, Uppsala, Sweden, April 22-29, 2017, Proceedings}, volume 10203 of {\em Lecture Notes in Computer Science}, pages 283--300, 2017.
\newblock \href {https://doi.org/10.1007/978-3-662-54458-7\_17} {\path{doi:10.1007/978-3-662-54458-7\_17}}.

\bibitem[Stu08]{Studer08:mu-calc}
Thomas Studer.
\newblock On the proof theory of the modal {$\mu$}-calculus.
\newblock {\em Studia Logica: An International Journal for Symbolic Logic}, 89(3):343--363, 2008.
\newblock URL: \url{http://www.jstor.org/stable/40268983}.

\bibitem[SY17]{SatoYamazaki17:rev-math-fixed-point-thms-with-kt}
Takashi Sato and Takeshi Yamazaki.
\newblock Reverse mathematics and order theoretic fixed point theorems.
\newblock {\em Arch. Math. Log.}, 56(3-4):385--396, 2017.
\newblock \href {https://doi.org/10.1007/s00153-017-0526-y} {\path{doi:10.1007/s00153-017-0526-y}}.

\bibitem[Tro98]{Troelstra98:realizability-handbook}
A.S. Troelstra.
\newblock Chapter vi - realizability.
\newblock In Samuel~R. Buss, editor, {\em Handbook of Proof Theory}, volume 137 of {\em Studies in Logic and the Foundations of Mathematics}, pages 407--473. Elsevier, 1998.
\newblock \href {https://doi.org/10.1016/S0049-237X(98)80021-9} {\path{doi:10.1016/S0049-237X(98)80021-9}}.

\bibitem[Tup04]{Tupailo04doubleneg}
Sergei Tupailo.
\newblock On the intuitionistic strength of monotone inductive definitions.
\newblock {\em J. Symb. Log.}, 69(3):790--798, 2004.
\newblock \href {https://doi.org/10.2178/jsl/1096901767} {\path{doi:10.2178/jsl/1096901767}}.

\bibitem[Ven08]{venema2008lectures}
Yde Venema.
\newblock Lectures on the modal $\mu$-calculus.
\newblock {\em Renmin University in Beijing (China)}, 2008.

\end{thebibliography}

\clearpage
\appendix

\section{Weak and strong (co)iteration rules}\label{app:simple-types-with-fixed-points}







Our presentation of the iteration rule (i.e., $\lr\mu$)  and the coiteration rule (i.e., $\rr \nu$) for $\muLJ$ is the most permissive one, in that contexts are allowed to appear in both premises:
 \begin{equation}\label{eqn:strong-2-premises}
      \small
    \vliinf{\lr \mu}{}{\Gamma,\Delta, \mu X \sigma(X)  \seqar \tau}{\Gamma, \sigma(\rho)\seqar \rho}{\Delta, \rho \seqar \tau}
  \qquad
   \vliinf{\rr \nu }{}{ \Delta, \Gamma  \seqar \nu X\sigma(X)}{\Delta \seqar  \tau }{  \Gamma, \tau \seqar \sigma(\tau)}
 \end{equation}

 Their cut-reduction rules are as in~\Cref{fig:cut-reduction-mu-nu-muLJ}. The reader should notice that the context $\Gamma$ in both the rules of~\eqref{eqn:strong-2-premises} is contracted during cut-reduction. For this reason, alternative presentations of $\lr\mu $ and $\rr \nu$ without  the context $\Gamma$ have been studied in the literature, especially in the setting of linear logic, where contraction cannot be freely used (see, e.g.,~\cite{BaeldeMiller07})
 
  \begin{equation}\label{eqn:weak-2-premises}
      \small
    \vliinf{\lr \mu}{}{\Delta, \mu X \sigma(X)  \seqar \tau}{ \sigma(\rho)\seqar \rho}{\Delta, \rho \seqar \tau}
  \qquad
   \vliinf{\rr \nu }{}{ \Delta  \seqar \nu X\sigma(X)}{\Delta \seqar  \tau }{   \tau \seqar \sigma(\tau)}
 \end{equation}

Single premise formulations of $\lr \mu$ and $\rr \nu$ for both~\eqref{eqn:strong-2-premises} and~\eqref{eqn:weak-2-premises} have been investigated in the literature as well (see, e.g.,~\cite{Clairambault09}):
\begin{equation}\label{eqn:strong-1-premises}
      \small
    \vlinf{\lr \mu}{}{\Gamma, \mu X \sigma(X)  \seqar \tau}{\Gamma, \sigma(\tau)\seqar \tau}
  \qquad
   \vlinf{\rr \nu }{}{  \Gamma  \seqar \nu X\sigma(X)}{  \Gamma \seqar \sigma(\tau)}
 \end{equation}
 \begin{equation}\label{eqn:weak-1-premises}
      \small
    \vlinf{\lr \mu}{}{ \mu X \sigma(X)  \seqar \tau}{ \sigma(\tau)\seqar \tau}
  \qquad
   \vlinf{\rr \nu }{}{ \tau  \seqar \nu X\sigma(X)}{   \tau \seqar \sigma(\tau)}
 \end{equation}

 The corresponding cut-reduction rules for~\eqref{eqn:weak-2-premises},~\eqref{eqn:strong-1-premises}, and~\eqref{eqn:weak-1-premises} can be easily extracted from the ones in~\Cref{fig:cut-reduction-mu-nu-muLJ}. 

 It is worth mentioning, however, that cut-elimination fails in presence of the rules \eqref{eqn:strong-1-premises} and~\eqref{eqn:weak-1-premises}. By contrast, being essentially endowed with a built-in cut, the rules  \eqref{eqn:strong-2-premises} and~\eqref{eqn:weak-2-premises} allow cut-elimination results. 

Perhaps surprisingly, all  formulations of $\lr \mu$ and $\rr \nu$ are equivalent, meaning they can derive each other  crucially using  the cut rule. 

\begin{proposition}\label{prop:weak-derives-strong}
    In $\LJ$ endowed with the rules $\rr \mu$ and $\lr \nu$ (\Cref{fig:unfolding-rules}), the  rules~\eqref{eqn:strong-2-premises}, \eqref{eqn:weak-2-premises},~\eqref{eqn:strong-1-premises}, and~\eqref{eqn:weak-1-premises} and their cut-reduction are inter-derivable.
\end{proposition}
\begin{proof}
It suffices to show that the rules in~\eqref{eqn:weak-1-premises} can derive the rules in~\eqref{eqn:strong-1-premises}. We only show the case of $\lr \mu$, as the case of $\rr \nu$ is symmetric. W.l.o.g., we will treat  $\Gamma$ in~\eqref{eqn:strong-1-premises} as a single formula. The case where $\Gamma$ is an arbitrary context can be recovered using $\lr \times$ and $\rr \times$. The derivation  is the following:
\[
\vlderivation{
  \vliin{\cut}{}{\Gamma, \mu X. \sigma \seqar \tau}
      {
      \vltr{\der}{\Gamma, \mu X. \sigma \seqar \mu X. (\Gamma \times \sigma)}{\vlhy{\ }}{\vlhy{\ }}{\vlhy{\ }}
      }{
      \vlin{\lr \mu}{}{\mu X. (\Gamma \times \sigma)\seqar \tau}
        {
        \vlin{\lr \times}{}{\Gamma \times \sigma (\tau)\seqar \tau}
          {
             \vlhy{\Gamma, \sigma (\tau)\seqar \tau}
          }
        }
      }
}
\]
where $\der$ is the derivation in~\Cref{fig:deriving-strong-from-weak}. It is tedious but routine to show that the above derivation of $\lr \mu$ allows us to derive the cut-reduction rule for~\eqref{eqn:weak-1-premises}.
\end{proof}

\begin{figure*}
    \centering
   \[
\footnotesize
\vlderivation{
\vlin{\rr \mu}{}{\Gamma, \mu X. \sigma \seqar \mu X. (\Gamma \times \sigma)}
   {
   \vliq{\contr}{}{\Gamma, \mu X. \sigma \seqar \Gamma \times \sigma (\mu X. (\Gamma \times \sigma))}{
   \vliin{\rr \times}{}{\Gamma, \Gamma, \mu X. \sigma \seqar \Gamma \times \sigma (\mu X. (\Gamma \times \sigma))  }
     {
     \vlin{\id}{}{\Gamma \seqar \Gamma}{\vlhy{}}
     }{
     \vliin{\cut}{}{\Gamma, \mu X. \sigma \seqar  \sigma (\mu X. (\Gamma \times \sigma)) }
       {
      \vlin{\lr \mu}{}{\mu X. \sigma \seqar \Gamma \to \sigma (\mu X. (\Gamma \times \sigma))}
        {
        \vlin{\rr \to }{}{\sigma ( \Gamma \to \sigma (\mu X. (\Gamma \times \sigma))) \seqar \Gamma \to \sigma (\mu X. (\Gamma \times \sigma))}
          {
          \vliq{\sigma}{}{\Gamma,\sigma ( \Gamma \to \sigma (\mu X. (\Gamma \times \sigma))) \seqar  \sigma (\mu X. (\Gamma \times \sigma))}
            {
            \vlin{\rr \mu}{}{\Gamma, \Gamma \to \sigma (\mu X. (\Gamma \times \sigma)) \seqar \mu X. (\Gamma \times \sigma)}
              {
              \vliq{\contr}{}{\Gamma, \Gamma \to \sigma (\mu X. (\Gamma \times \sigma)) \seqar \Gamma \times \sigma(\mu X. (\Gamma \times \sigma))}
              {
              \vliin{\rr \times }{}{\Gamma ,\Gamma, \Gamma \to \sigma (\mu X. (\Gamma \times \sigma)) \seqar \Gamma \times \sigma(\mu X. (\Gamma \times \sigma))}
                {
                \vlin{\id}{}{\Gamma \seqar \Gamma}{\vlhy{}}
                }{
                \vliin{\lr \to}{}{\Gamma, \Gamma \to \sigma (\mu X. (\Gamma \times \sigma)) \seqar  \sigma(\mu X. (\Gamma \times \sigma))}
                  {
                   \vlin{\id}{}{\Gamma \seqar \Gamma}{\vlhy{}}
                  }{
                  \vlhy{\sigma (\mu X. (\Gamma \times \sigma)) \seqar  \sigma(\mu X. (\Gamma \times \sigma))}
                  }
                }
              }
              }
            }
          }
        }
       }{
       \vliin{\lr \to}{}{\Gamma \to  \sigma (\mu X. (\Gamma \times \sigma)), \Gamma \seqar \sigma (\mu X. (\Gamma \times \sigma)) }{ \vlin{\id}{}{\Gamma \seqar \Gamma}{\vlhy{}}}{\vlhy{\sigma (\mu X. (\Gamma \times \sigma)) \seqar \sigma (\mu X. (\Gamma \times \sigma))}}
       }
     }
    } 
   }
}
\]
    \caption{Deriving~\eqref{eqn:strong-1-premises} from~\eqref{eqn:weak-1-premises}.}
    \label{fig:deriving-strong-from-weak}
\end{figure*}

 \section{Proof of~\Cref{prop:cmulj-into-cmuljminus}}\label{app:proofs-of-sec-circular-mulj}
%
\label{subsec:proofs-of-cmulj-into-cmuljminus}

We give a bespoke combination of a negative translation and a Friedman-Dragalin `A translation' suitable for our purposes.
Writing $\N:= \mu X (1+X)$ and $\nega \N \sigma := \sigma \to \N$.
At the risk of ambiguity but to reduce syntax, we shall simply write $\neg \sigma$ for $\nega \N \sigma$ henceforth.\footnote{
Any confusion will be minimal as we shall not parametrise negation by any other A-translation in this work.} Also, we shall include the following derivable rules of $\cmuLJ$:
    \begin{equation}\label{eqn:negation-rules}
         \vlinf{\rr{\neg} }{}{\Gamma \seqar  \neg \sigma}{\Gamma, \sigma \seqar \N} 
 \qquad 
 \vlinf{\lr{\neg}}{}{\Gamma, \neg \sigma \seqar \N}{\Gamma \seqar \sigma}
    \end{equation}
\[
     \vlinf{\rr{\neg \neg} }{}{\Gamma \seqar \neg \neg \sigma}{\Gamma \seqar \sigma}   
     \qquad 
     \vlinf{\lr {\neg \neg }}{}{\Gamma, \neg\neg \sigma \seqar \N}{\Gamma, \sigma \seqar \N }
\]

We define the translations $\negtrans \cdot$ and $\unnegtrans \cdot$ from arbitrary types to types over $\{\N,\times, \to , \mu\}$ as follows:
\[
\begin{array}{r@{\ := \ }l}
\negtrans \sigma & \neg \unnegtrans \sigma\\ 
\unnegtrans X & \negn X \\
    \unnegtrans 1 & \N \\
    \unnegtrans{(\sigma \times \tau)} & \negn (\negtrans \sigma \times \negtrans \tau) \\
    \unnegtrans{(\sigma \to \tau)} & \neg( \negtrans\sigma \to \negtrans \tau) \\
    \unnegtrans{(\sigma+\tau)} & \neg   \negtrans \sigma \times \neg \negtrans \tau \\
    \unnegtrans{(\nu X \sigma)} & \neg \neg \mu X \negn \negtrans\sigma[\negn X / X] \\
    \unnegtrans{(\mu X \sigma)} &  \neg \mu X \negtrans{\sigma}
\end{array}
\]

\begin{proposition}
    [Substitution]{\ }
    \todonew{what are the invariants? also mixed notations below, be consistent.}
    \begin{itemize}
    \item $\negtrans{(\sigma (\mu X. \sigma  ))}=\negtrans{\sigma}(\mu X. \negtrans{\sigma}) $
    \item $\negtrans{(\sigma (\nu X. \sigma))} = \negtrans{\sigma}(\neg \mu X. \neg \negtrans{\sigma}[\neg X/X])$
     \item $\unnegtrans{(\sigma (\nu X. \sigma))} = \unnegtrans{\sigma}(\neg \mu X. \neg \negtrans{\sigma}[\neg X/X])$.
    \end{itemize}
\end{proposition}

We shall extend the notations $\negtrans \cdot$ and $\unnegtrans \cdot$ to cedents, e.g.\ writing $\negtrans {\Sigma}$ and $\unnegtrans {\Sigma}$, by distributing the translation over the list.
For a sequent $\Sigma \seqar \tau$, we also write $\negtrans{(\Sigma \seqar \tau)}$ for $\negtrans{\Sigma} , \unnegtrans \tau \seqar \N$

\begin{definition}
    [$N$-translation of steps]
    For each inference step \[\vliiinf{\rrule}{}{\Sigma \seqar \tau}{\Sigma_1 \seqar \tau_1}{\cdots }{\Sigma_n \seqar \tau_n}\] we define a gadget,
    \[
    \toks0={0.5}
    \vlderivation{
    \vltrf{\negtrans \rrule}{\negtrans{(\Sigma \seqar \tau)}}{\vlhy{\negtrans{(\Sigma_1 \seqar \tau_1)}}}{\vlhy{\cdots}}{\vlhy{\negtrans{(\Sigma_n \seqar \tau_n)}}}{\the\toks0}
    }
    \]
    as in~\Cref{fig:translation-inference-rules}.
       We lift this to a translation on coderivations $\der \mapsto \negtrans \der$ in the obvious way.
\end{definition}

\begin{figure*}
\centering
\[\small
    \begin{array}{rclrcl}
       \negtrans \ax  &\dfn & 
         \vlderivation{
         \vlin{\lr \neg }{}{\negtrans \sigma, \unnegtrans \sigma \seqar \N}{
         \axiom{\unnegtrans \sigma}{\unnegtrans \sigma}
           }
         }  
         &\qquad
         \negtrans \cut &\dfn &\vlderivation{
    \vliin{\cut}{}{\negtrans \Gamma, \negtrans \Delta, \unnegtrans \tau \seqar \N}
    {
    \vlin{\rr \neg }{}{\negtrans \Gamma \seqar \negtrans \sigma}{\haxiom{\negtrans \Gamma, \unnegtrans \sigma}{\N}}
    }
    {
    \haxiom{\negtrans \Delta, \negtrans \sigma, \unnegtrans \tau}{\N}
    }
    }
    \end{array}
    \]
    \bigskip
      \[\small
    \begin{array}{rclrcl}
    \negtrans {\rr \unit }  &  \dfn    & 
    \vlderivation{
    \vlin{\id}{}{\unnegtrans \unit \seqar \N}{\vlhy{}}
    }& \qquad 
\negtrans{\lr \unit}    &\dfn & 
\vlderivation{
\vlin{\wk}{}{\negtrans \Gamma, \negtrans \unit, \unnegtrans \sigma \seqar \N}
  {
\haxiom{\negtrans \Gamma, \unnegtrans \sigma}{\N}
  }
}
\end{array}
    \]
  \bigskip
      \[\small
    \begin{array}{rclrcl}
 \negtrans{\rr \to} & \dfn  &
    \vlderivation{
    \vlin{\lr \neg}{}{\negtrans \Gamma, \unnegtrans{(\sigma \to \tau)}\seqar \N}
     {
     \vlin{\rr \to}{}{\negtrans \Gamma \seqar \negtrans \sigma \to \negtrans \tau}
       {
       \vlin{\rr \neg}{}{\negtrans \Gamma, \negtrans \sigma \seqar \negtrans \tau}{\haxiom{\negtrans \Gamma, \negtrans \sigma, \unnegtrans \tau}{\N}}
       }
     }
    }
    &\qquad
\negtrans {\lr \to}    &\dfn & \vlderivation{
\vlin{\lr {\neg \neg}}{}{\negtrans \Gamma, \negtrans \Delta, \negtrans{(\sigma \to \tau)}, \unnegtrans \gamma \seqar \N}
    {
    \vliin{\lr \to}{}{\negtrans \Gamma, \negtrans \Delta, \negtrans \sigma \to \negtrans \tau , \unnegtrans \gamma \seqar \N}
   {
   \vlin{\rr \neg}{}{\negtrans \Gamma \seqar \negtrans \sigma}{\haxiom{\negtrans \Gamma, \unnegtrans \sigma}{\N}}
   }{
   \haxiom{\negtrans \Delta, \negtrans \tau, \unnegtrans \gamma}{\N}
   }
    }

}
\end{array}
    \]  
    
      \[\small
    \begin{array}{rclrcl}
\negtrans{\rr \times}&\dfn & \vlderivation{
\vlin{\lr \neg}{}{\negtrans \Gamma, \unnegtrans{(\sigma \times  \tau)}\seqar \N}
  {
  \vliin{\rr \times}{}{\negtrans \Gamma, \negtrans \Delta \seqar \negtrans \sigma \times \negtrans \tau}
    {
    \vlin{\rr \neg}{}{\negtrans \Gamma \seqar \negtrans \sigma}{\haxiom{\negtrans \Gamma, \unnegtrans \sigma}{\N}}
    }
    {
    \vlin{\rr \neg}{}{\negtrans \Delta \seqar \negtrans \tau}{\haxiom{\negtrans \Delta, \unnegtrans \tau}{\N}}
    }
  }
} &\qquad
\negtrans {\lr \times}&\dfn&\vlderivation{
\vlin{\lr {\neg \neg}}{}{\negtrans \Gamma, \negtrans{(\sigma \times \tau)}, \unnegtrans \gamma \seqar \N} 
    {
    \vlin{\lr \times }{}{\negtrans{\Gamma}, \negtrans \sigma \times \negtrans \tau , \unnegtrans \gamma \seqar \N} 
    {
    \haxiom{\negtrans \Gamma, \negtrans \sigma, \negtrans \tau, \unnegtrans \gamma}{\N}
    }
    }
  }
  \end{array}
    \]
      \bigskip
      \[\small
    \begin{array}{rclrcl}
\negtrans{\cev{\rr +}}&\dfn & \vlderivation{
\vlin{\lr \times }{}{\negtrans \Gamma, \unnegtrans{(\sigma + \tau)} \seqar \N}
  {
  \vlin{\wk}{}{\negtrans \Gamma, \neg  \negtrans \sigma  ,  \neg \negtrans \tau \seqar \N}
    {
    \vlin{\lr {\neg \neg}}{}{\negtrans \Gamma, \neg  \negtrans \sigma  \seqar \N}{\haxiom{\negtrans \Gamma, \unnegtrans \sigma}{\N}}
    }
  }
}
&
\qquad 
\negtrans{{+^1_r}}&\dfn & \vlderivation{
\vlin{\lr \times }{}{\negtrans \Gamma, \unnegtrans{(\sigma + \tau)} \seqar \N}
  {
  \vlin{\wk}{}{\negtrans \Gamma, \neg \negtrans \sigma  ,  \neg  \negtrans \tau \seqar \N}
    {
    \vlin{\lr {\neg \neg}}{}{\negtrans \Gamma, \neg  \negtrans \tau  \seqar \N}{\haxiom{\negtrans \Gamma, \unnegtrans \tau}{\N}}
    }
  }
}
\end{array}
\]
  \bigskip
\[
\negtrans {\lr +} \dfn\vlderivation{
\vlin{\lr {\neg  }}{}{\negtrans \Gamma, \negtrans{(\sigma + \tau)}, \unnegtrans \gamma \seqar \N} 
    {
    \vliq{\contr}{}{\negtrans{\Gamma}, \unnegtrans \gamma \seqar \neg  \negtrans \sigma \times \neg  \negtrans \tau}{
    \vliin{\lr \times }{}{\negtrans{\Gamma},\negtrans{\Gamma}, \unnegtrans \gamma \seqar \neg  \negtrans \sigma \times \neg  \negtrans \tau} 
    {
    \vlin{\lr \neg}{}{\negtrans \Gamma, \unnegtrans \gamma\seqar \neg \negtrans \sigma}{
    \haxiom{\negtrans \Gamma, \negtrans \sigma, \unnegtrans \gamma}{\N}
    }
    }{
    \vlin{\lr \neg}{}{\negtrans \Gamma, \unnegtrans \gamma\seqar \neg \negtrans \tau}{
    \haxiom{\negtrans \Gamma, \negtrans \tau, \unnegtrans \gamma}{\N}
    }
    }
    }
    }
}
    \]
      \bigskip
     \[\small
    \begin{array}{rclrcl}
    \negtrans {\rr \mu} & \dfn & 
    \vlderivation{
    \vlin{\lr {\neg }}{}{\negtrans \Gamma, \unnegtrans{(\mu X. \sigma)} \seqar \N}
    {
    \vlin{\rr \mu}{}{\negtrans \Gamma\seqar \mu X. \negtrans{\sigma}} 
    {
    \vlid{=}{}{\negtrans \Gamma\seqar  \negtrans{\sigma}(\mu X. \negtrans{\sigma})} 
      {
      \vlin{\rr {\neg }}{}{\negtrans \Gamma \seqar \negtrans{(\sigma (\mu X. \sigma))}}{
      \haxiom{\negtrans \Gamma, \unnegtrans{(\sigma (\mu X. \sigma))}}{\N}
      }
      }
    }
    }
    }
      & \qquad 
      \negtrans{\lr \mu}& \dfn & \vlderivation{
   \vlin{\lr {\neg \neg}}{}{\negtrans \Gamma, \negtrans{(\mu X. \sigma)} , \unnegtrans \gamma \seqar \N} 
       {
       \vlin{\rr \mu}{}{\negtrans \Gamma, \mu X. \negtrans{\sigma} , \unnegtrans \gamma \seqar \N}
         {
         \vlid{=}{}{\negtrans \Gamma,  \negtrans{\sigma}(\mu X. \negtrans{\sigma}) ,\unnegtrans \gamma \seqar \N}{
         \haxiom {\negtrans \Gamma,  \negtrans{(\sigma(\mu X.\sigma))} ,\unnegtrans \gamma}{  \N}
         }
         }
       }
   }
    \end{array}
    \]
       \bigskip
     \[\small
    \begin{array}{rclrcl}
    \negtrans{\rr \nu} & \dfn & 
    \vlderivation{
    \vlin{\lr {\neg \neg}}{}{\negtrans \Gamma, \unnegtrans{(\nu X. \sigma)} \seqar \N}{
    \vlin{\lr {\mu }}{}{\negtrans \Gamma,  \mu X. \neg {\negtrans \sigma(\neg X)}\seqar \N}
    {
    \vlin{\lr {\neg\neg }}{}{\negtrans \Gamma, \neg {\negtrans \sigma(\neg \mu X. \neg {\negtrans \sigma(\neg X)})} \seqar \N} 
    {
    \vlid{=}{}{\negtrans \Gamma,  {\unnegtrans \sigma(\neg \mu X. \neg {\negtrans \sigma(\neg X)})} \seqar \N} 
      {
      \haxiom{\negtrans \Gamma, \unnegtrans{(\sigma(\nu X. \sigma))}}{\N}
      }
    }
    }
    }
    }& \qquad
   \negtrans{\lr \nu}&\dfn &  \vlderivation{
   \vlin{\lr {\neg \neg }}{}{\negtrans \Gamma, \negtrans{(\nu X. \sigma )}, \unnegtrans \gamma \seqar \N} 
   {
   \vlin{\lr {\neg}}{}{\negtrans \Gamma,  \neg \mu X. \neg {\negtrans \sigma(\neg X)}, \unnegtrans \gamma \seqar \N} 
   {
   \vlin{\lr {\mu}}{}{\negtrans \Gamma, \unnegtrans \gamma \seqar  \mu X. \neg {\negtrans \sigma(\neg X)}} 
       {
       \vlin{\rr {\neg } }{}{\negtrans \Gamma, \unnegtrans \gamma \seqar  \neg {\negtrans \sigma(\neg \mu X. \neg {\negtrans \sigma(\neg X)})}}
         {
         \vlid{=}{}{\negtrans \Gamma,  \negtrans \sigma(\neg \mu X. \neg {\negtrans \sigma(\neg X)}) ,\unnegtrans \gamma \seqar \N}{
         \haxiom{\negtrans \Gamma,  \negtrans{(\sigma (\nu X. \sigma ))} ,\unnegtrans \gamma}{  \N}
         }
         }
       }
       }
       }
   }
    \end{array}
    \]
    \caption{Translation $\negtrans{(\_)}$ on inference rules.}
    \label{fig:translation-inference-rules}
\end{figure*}

\begin{proposition}\label{prop:translation-preserves-progressiveness}
        If $\der$ is progressing and regular then so is $\negtrans \der$.
\end{proposition}
\begin{proof}
    As for regularity, we observe that the translation maps any inference rule $\rrule$ to a gadget $\negtrans \rrule$ with constantly many rules. Concerning progressiveness, it is easy to check that:
    \begin{itemize}
        \item  for any infinite branch $(\Delta^j \seqar \tau^j)_j$ of $\der'$ there is an infinite branch $(\Sigma^i \seqar \sigma^i)_i$ of $\der$ such that $\Delta^{j_i} \seqar \tau^{j_i}= \negtrans{(\Sigma^i \seqar \sigma^i)}$ for some sequence  $j_0<j_1<\ldots<j_n<\ldots$
        \item  for any thread  in $\der$ connecting a formula $\gamma$ in $\Sigma^i \seqar \tau^i$ with a formula $\gamma'$ in $\Sigma^{i+1} \seqar \tau^{i+1}$ there is a thread  in $\der'$ connecting $\negtrans \gamma$ in $\negtrans{(\Sigma^i \seqar \tau^i)}$ with a formula $\negtrans{\gamma'}$ in $\negtrans{(\Sigma^{i+1} \seqar \tau^{i+1})}$; moreover, if the former thread has progressing points then the latter has. \qedhere
    \end{itemize}
\end{proof}

\begin{proposition}
\label{lem:computational-soundness}
If $\der_1 \cutrednorec \der_2$ then $\negtrans{\der}_1\cutrednorec^* \negtrans{\der}_2$ 
\end{proposition}
\begin{proof}
The proof is routine. We just consider the cut-reduction step eliminating a cut $\rrule$ of the form $\rr\nu$ vs $\lr \nu$. W.l.o.g., we assume $\der$ ends with this cut. Then, $\der$ has the following shape:
\[
\vlderivation{
\vliin{\cut}{}{\Gamma, \Delta \seqar \tau}{\vlin{\rr \nu}{}{\Gamma \seqar \nu X. \sigma}{\vltr{\der_1}{\Gamma \seqar \sigma (\nu X. \sigma)}{\vlhy{\ \ }}{\vlhy{\ \ }}{\vlhy{\ \ }}}}{\vlin{\lr \nu}{}{\Delta, \nu X. \sigma  \seqar \tau}{\vltr{\der_2}{\Delta, \sigma (\nu X. \sigma)\seqar \tau}{\vlhy{\ \ }}{\vlhy{\ \ }}{\vlhy{\ \ }}}}
}
\]
Which is translated into a cut between the following two coderivations:
\[
\small
\vlderivation{
\vlin{\rr \neg}{}{\negtrans{\Gamma}\seqar \negtrans{(\nu X. \sigma)} }{
    \vlin{\lr {\neg \neg}}{}{\negtrans \Gamma, \unnegtrans{(\nu X. \sigma)} \seqar \N}{
    \vlin{\lr {\mu }}{}{\negtrans \Gamma,  \mu X. \neg {\negtrans \sigma(\neg X)}\seqar \N}
    {
    \vlin{\lr {\neg\neg }}{}{\negtrans \Gamma, \neg {\negtrans \sigma(\neg \mu X. \neg {\negtrans \sigma(\neg X)})} \seqar \N} 
    {
    \vlid{=}{}{\negtrans \Gamma,  {\unnegtrans \sigma(\neg \mu X. \neg {\negtrans \sigma(\neg X)})} \seqar \N} 
      {
      \vltr{\negtrans{\der_1}}{\negtrans \Gamma, \unnegtrans{(\sigma(\nu X. \sigma))}\seqar \N}{\vlhy{\ \ }}{\vlhy{\ \ }}{\vlhy{\ \ }}
      }
    }
    }
    }
    }
    }
     \]
     \[
     \vlderivation{
     \small
   \vlin{\lr {\neg \neg }}{}{\negtrans \Delta , \negtrans{(\nu X. \sigma )}, \unnegtrans \tau \seqar \N} 
   {
   \vlin{\lr {\neg}}{}{\negtrans \Delta,  \neg \mu X. \neg {\negtrans \sigma(\neg X)}, \unnegtrans \tau \seqar \N} 
   {
   \vlin{\lr {\mu}}{}{\negtrans \Delta, \unnegtrans \tau \seqar  \mu X. \neg {\negtrans \sigma(\neg X)}} 
       {
       \vlin{\rr {\neg } }{}{\negtrans \Delta, \unnegtrans \tau \seqar  \neg {\negtrans \sigma(\neg \mu X. \neg {\negtrans \sigma(\neg X)})}}
         {
         \vlid{=}{}{\negtrans \Delta,  \negtrans \sigma(\neg \mu X. \neg {\negtrans \sigma(\neg X)}) ,\unnegtrans \tau \seqar \N}{
         \vltr{\negtrans{\der_2}}{\negtrans \Delta,  \negtrans{(\sigma (\nu X. \sigma ))} ,\unnegtrans \tau\seqar  \N}{\vlhy{\ \ }}{\vlhy{\ \ }}{\vlhy{\ \ }}
         }
         }
       }
       }
       }
}
\]
  by applying a few cut-reduction  steps we obtain the following:
  \[
\vlderivation{
\vliin{\cut}{}{\negtrans \Gamma, \negtrans \Delta, \unnegtrans \tau \seqar \N}{
\vlin{\rr \neg}{}{\negtrans\Gamma\seqar  \negtrans{(\nu X. \sigma)} }{
\vltr{\negtrans{\der_1}}{\negtrans \Gamma, \unnegtrans{(\sigma (\nu X. \sigma))}\seqar \N}{\vlhy{\ \ }}{\vlhy{\ \ }}{\vlhy{\ \ }}
}
}
{
\vltr{\negtrans{\der_2}}{\negtrans \Delta , \negtrans{(\sigma (\nu X. \sigma))}, \unnegtrans \tau \seqar \N }{\vlhy{\ \ }}{\vlhy{\ \ }}{\vlhy{\ \ }}}
}
\]
Which is the translation of the following:
\[
\vlderivation{
\vliin{\cut}{}{\Gamma, \Delta \seqar \tau}
{\vltr{\der_1}{\Gamma \seqar \sigma (\nu X. \sigma)}{\vlhy{\ \ }}{\vlhy{\ \ }}{\vlhy{\ \ }}}
{\vltr{\der_2}{\Delta, \sigma (\nu X. \sigma)\seqar \tau}{\vlhy{\ \ }}{\vlhy{\ \ }}{\vlhy{\ \ }}}
} \qedhere
\]
\end{proof}

\begin{proposition}[(De)coding for $\N$]\label{prop:coding-and-decoding}
  There are  coderivations $\dercod: \N \seqar \negtrans \N$ and $\derdec: \negtrans \N \seqar \N$  over $\{\N,\to, \times, \mu\}$ in  $\cmuLJ$  such that, for any $n \in \Nat$:
  \[\small
\vlderivation{
\vliin{\cut}{}{\seqar \negtrans \N}{\vldr{\cod{n}}{\seqar \N}}{\vldr{\dercod}{\N \seqar \negtrans \N} }
}
 \cutrednorec^*
\vlderivation{\vldr{\negtrans{\cod n}}{\seqar \negtrans \N}}
\]
\[
\small
 \vlderivation{
\vliin{\cut}{}{\seqar \N}{\vldr{\negtrans{\cod n}}{\seqar \negtrans \N}}{\vldr{\derdec}{\negtrans \N \seqar  \N} }
}
 \cutrednorec^* 
\vlderivation{\vldr{\cod n}{\seqar \N}}
\]
\end{proposition}
\begin{proof}
$\dercod$ and $\derdec$ are defined, respectively, as follows:
    \[
\small
\vlderivation{
\vlin{\rr{\neg \neg}}{\bullet}{\N \seqar \negtrans{\N}}
  {
  \vlin{\mu r}{}{\N \seqar \mu X.(\neg  (\neg  \negtrans{\unit} \times  \neg \negtrans{ X}) )}
    {
    \vlin{\rr \neg}{}{\N \seqar \neg  (\neg  \negtrans{\unit} \times  \neg \negtrans{ \N}) ) }
      {
      \vlin{\lr \times}{}{\N ,  \neg  \negtrans{\unit} \times  \neg \negtrans{ \N} \seqar \N}{
      \vlin{\lr \mu}{}{\N ,  \neg  \negtrans{\unit} ,  \neg \negtrans{ \N} \seqar \N}
        {
        \vliin{\lr +}{}{ 1+     \N ,  \neg  \negtrans{\unit} ,  \neg \negtrans{ \N} \seqar \N}
          {
          \vliq{\wk}{}{ 1,  \neg  \negtrans{\unit} ,  \neg \negtrans{ \N} \seqar \N}
            {
            \vlin{\lr{\neg }}{}{   \neg \negtrans{1} \seqar \N}{\vlin{\rr \neg}{}{\seqar \negtrans 1}{\axiom{\N}{\N}}}
            }
          }
          {
          \vlin{\wk}{}{ \N ,  \neg  \negtrans{\unit} ,  \neg \negtrans{ \N} \seqar \N }
            {
            \vlin{\lr \neg}{}{\N, \neg \negtrans{\N}\seqar \N }{
            \vlin{\rr{\neg \neg}}{\bullet}{\N \seqar \negtrans \N}{\vlhy{\vdots}}
            }
            }
          }
        }
        }
      }
    }
  }
}
\]
\[
\small
\vlderivation{
\vlin{\lr {\neg \neg}}{\bullet}{\negtrans{\N} \seqar \N}
  {
  \vlin{\lr \mu}{}{ \mu X.(\neg  (\neg  \negtrans{\unit} \times  \neg \negtrans{ X}) )\seqar \N}
   {
   \vlin{\lr{\neg }}{}{\neg  (\neg  \negtrans{\unit} \times  \neg \neg \negtrans{\N}) \seqar \N}
      {
      \vliin{\rr \times}{}{\seqar \neg  \negtrans{\unit} \times  \neg \neg \negtrans{\N}}
        {
        \vlin{\rr {\neg}}{}{\seqar \neg \negtrans{1}}
         {
         \vlin{\wk}{}{\negtrans{1} \seqar \N}{\vlin{\rr N}{}{\seqar \N}{\vlhy{}}}
         }
        }{
        \vlin{\rr {\neg}}{}{\seqar \neg \negtrans{\N}}
         {
         \vliin{\cut}{}{ \negtrans{\N}\seqar \N}
         {
         \vlin{\lr{\neg \neg }}{\bullet}{\negtrans \N \seqar \N}{\vlhy{\vdots}}
         }
         {
         \vlin{\Nsucc}{}{\N \seqar \N}{\vlin{\id}{}{\N \seqar \N}{\vlhy{}}}
         }
         }
        }
      }
   }
  }
}
\]
where $0$ and $\succ$ are derivations as in~\Cref{fig:native-n-rules-mulj}.
\end{proof}

\begin{theorem}\label{thm:cmullminus-simulates-cmulj}
    If $\der: \vec \N \seqar \N$ in $\cmuLJ$ then there is  coderivation  $\der':\vec \N \seqar \N$ over $\{\N,\to, \times, \mu\}$ in $\cmuLJ$ such that $\der$ and $\der'$ represent the same functions on natural numbers.
\end{theorem}
\begin{proof}
Let $f$ be a numerical function of $\cmuLJ$, w.l.o.g.~we can assume $f$ is unary. This means that there is a derivation $\der$ of  $\N \seqar  \N$ such that 
\[\vlderivation{
\vliin{\cut}{}{\seqar \N}{\vldr{\cod{n}}{\seqar \N}}{\vldr{\der}{\N \seqar  \N} }
}
 \qquad  \cutrednorec^* \qquad 
\vlderivation{\vldr{\cod{f(n)}}{\seqar \N}}
\]
By Proposition~\ref{lem:computational-soundness} we have that, for any $n \in \Nat$:
 \[
 \begin{array}{rcl}
 \vlderivation{
\vliin{\cut}{}{\seqar \negtrans \N}{\vldr{\negtrans {\cod{n}}}{\seqar \negtrans \N}}{
\vlin{\rr \neg}{}{\negtrans\N \seqar   \negtrans\N}{
\vldr{ \negtrans \der}{ \negtrans\N, \unnegtrans \N \seqar   \N}
}
}
}
 &  \cutrednorec^* & 
\vlderivation{\vldr{ \negtrans{\cod{f(n)}}}{\seqar  \negtrans\N}}
\end{array}
\]
Then, $f$ can be represented by the following coderivation $\der'$ over $\{\N,\to, \times, \mu\}$:
\[
\begin{array}{rcl}
    \vlderivation{
    \vliin{\cut}{}{\N \seqar \N}
    {
    \vliin{\cut}{}{\N \seqar \negtrans \N}{\vldr{\dercod}{\N \seqar \negtrans\N}}
    {
  \vlin{\rr \neg}{}{\negtrans\N \seqar   \negtrans\N}{
\vldr{ \negtrans \der}{ \negtrans\N, \unnegtrans \N \seqar   \N}
}
    }
    }
    {\vldr{\derdec}{\negtrans \N \seqar \N}}
    } \tag*{\qed}
\end{array}
\]
\renewcommand{\qed}{}
\end{proof}
 \section{Proof of \Cref{thm:muPA-pi02-cons-muHA}}

\label{sec:conservativity-via-double-negation}

\newcommand{\godelgentz}[1]{{#1}^{g}}
\newcommand{\friedman}[1]{{#1}^\rho}

In what follows we prove~\Cref{thm:muPA-pi02-cons-muHA}.

%
%
We start with a definition a G\"{o}del-Gentzen-style translation of formulas of $\muPA$ into formulas of $\muHA$:
\[
\def\arraystretch{1.2}
\begin{array}{rcl}
    \godelgentz {(t=u)}  & \dfn & t=u \\
  \godelgentz {(t<u)}  & \dfn & \neg \neg t<u\\
  \godelgentz {(t \in X)} & \dfn & \neg \neg t \in X\\
  \godelgentz {(t \in \mu X\lambda x\phi)} & \dfn & \neg \neg t \in \mu X\lambda x\godelgentz{\phi} \\
\godelgentz {(\neg \phi)} & \dfn & \neg \godelgentz{\phi}\\
 \godelgentz {(\phi \wedge \psi)} & \dfn &  \godelgentz{\phi}\wedge \godelgentz{\psi}\\
  \godelgentz {(\phi \vee \psi)} & \dfn & \neg(\neg \godelgentz{\phi}\wedge \neg \godelgentz{\psi})\\
\godelgentz{(\forall x\phi)}  & \dfn & \forall x\godelgentz{\phi}\\
 \godelgentz{(\exists x\phi)} & \dfn & \neg (\forall x\neg \godelgentz{\phi})
\end{array}
\]

Some straightforward properties of the translation:

\begin{lemma}{\ }\label{lem:godel}
\begin{enumerate}
 \item \label{lem:godel1}$\godelgentz{(\phi(\psi))}= \godelgentz{\phi}(\godelgentz{\psi})$
    \item\label{lem:godel2} $\muHA \vdash \neg \neg \godelgentz {\phi} \to \godelgentz {\phi}$
\end{enumerate}
\end{lemma}

\begin{lemma}\label{lem:muPA-in-muHA}
    If $\muPA \vdash \phi$ then $\muHA \vdash \godelgentz{\phi}$.
\end{lemma}
\begin{proof}
We show that any axiom of $\muPA$ is derivable in $\muHA$, where the translation on equations is justified by the fact that $\muHA \vdash \neg \neg t=s \to t=s$. We just check the axioms for  the least fixed point operator $\mu$.  On the one hand, we have $ \godelgentz{\phi}(\mufo X x \godelgentz{\phi}, y) \limp y \in \mufo X x \godelgentz{\phi}$ as an instance of (pre),  and $y \in \mufo X x \godelgentz{\phi} \to \neg \neg y \in \mufo X x \godelgentz{\phi}$ by logic. We conclude by applying transitivity of implication and the generalisation rule introducing $\forall$. On the other hand, let us assume $\forall x(\godelgentz{\phi}(\godelgentz{\psi}, x) \limp \godelgentz{\psi}(x))$. By (ind) we obtain $y \in \mufo X x \godelgentz{\phi}  \limp \godelgentz{\psi}(y)$. By logic we have 
 $\neg \neg y \in \mufo X x \godelgentz{\phi}  \limp \neg \neg \godelgentz{\psi}(y)$ and $ \neg \neg \godelgentz{\psi}(y)\to \godelgentz{\psi}(y) $ (\Cref{lem:godel}.\ref{lem:godel2}). We conclude by applying transitivity of implication and deduction theorem, using~\Cref{lem:godel}.\ref{lem:godel1}. 
\end{proof}

We now use a standard trick called \emph{Friedman's translation} to infer our conservativity result from the above lemma. Let $\rho$ be a fixed formula of $\muHA$. For any formula $\phi$ of $\muHA$ we define $\friedman{\phi}$ as follows:
\[
\begin{array}{rcll}
    \friedman{(t=u)}  &  \dfn  & (t=u) \vee \rho\\
    \friedman{(t<u)}  &  \dfn  & (t<u) \vee \rho\\
    \friedman{(t \in X)}  &  \dfn  & (t \in X) \vee \rho\\
    \friedman{(t \in \mu X \lambda x\phi)}  &  \dfn  & (t \in \mu X\lambda x\friedman{\phi}) \vee \rho\\
   \friedman{(\neg \phi)}  & \dfn & \neg \friedman{\phi}\\
    \friedman{( \phi \circ \psi)}  & \dfn & ( \friedman{\phi}\circ \friedman{\psi})& \qquad \circ \in \{\wedge, \vee\}\\
    \friedman{\sigma x\phi} &\dfn & (\sigma x. \friedman{\phi}) &\qquad \sigma \in \{\forall, \exists\}
 \end{array}
\]
 The replacement must be done without variable capture, so whenever we consider, e.g., $\friedman{(\forall x \phi)}$, we must assume that $x$ is not free in $\rho$. 

Some straightforward properties of $\friedman{(\_)}$:
 \begin{lemma}{\ }\label{lem:friedman}
 \begin{enumerate}
     \item \label{lem:friedman1} $ \friedman{(\phi(\psi))}=\friedman{\phi}(\friedman{\psi})$
     \item \label{lem:friedman2}$\rho \vdash \friedman{\phi}$
 \end{enumerate}
 \end{lemma}

\begin{lemma}[Friedman's translation, adapted]
If $\muHA \vdash \phi$ then $\muHA \vdash \friedman{\phi}$.
\end{lemma}
\begin{proof}
       It is easy to check that $\Gamma \vdash \phi$ implies $\friedman{\Gamma}\vdash \friedman{\phi}$, where $\friedman{\Gamma}\dfn\{\friedman{\gamma} \ \vert \ \gamma \in \Gamma\}$, observing that $\rho \vdash \friedman{\phi}$ by Lemma~\ref{lem:friedman}.\ref{lem:friedman2}. So, Friedman's translation preserves derivability. Therefore, to conclude, it suffices to show that $\muHA$ proves Friedman's translation of $\muHA$'s axioms. We only consider the cases of the axioms for the least fixed point operator $\mu$.  On the one hand, by instantiating (pre), we have $\friedman{\phi}(\mufo X x \friedman{\phi}, y) \limp y \in \mufo X x \friedman{\phi}$, with $\friedman{\phi}(\mufo X x \friedman{\phi}, y)= \friedman{(\phi(\mufo X x\phi, y))}$ (Lemma~\ref{lem:friedman}.\ref{lem:friedman1}), so that we can conclude  $\friedman{\forall y (\phi(\mufo X x \phi, y) \limp y \in \mufo X x \phi)}$. On the other hand, we have to prove that $\forall x (\friedman{\phi}(\friedman{\psi}, x) \limp \friedman{\psi}(x)) $ implies $y \in \mufo X x \friedman{\phi} \vee \rho  \limp \friedman{\psi}(y)$.  By instantiating (ind) we have that the former implies $y \in \mufo X x \friedman{\phi}   \limp \friedman{\psi}(y)$. Also, from Lemma~\ref{lem:friedman}.\ref{lem:friedman2} we infer $\rho \to \friedman{\psi}(y)$, so $\forall x (\friedman{\phi}(\friedman{\psi}, x) \limp \friedman{\psi}(x))$ implies $(y \in \mufo X x \friedman{\phi}) \vee \rho  \limp \friedman{\psi}(y)$. This concludes the proof.
\end{proof}

We can finally prove Proposition~\ref{thm:muPA-pi02-cons-muHA}:

\begin{proof}[Proof of Proposition~\ref{thm:muPA-pi02-cons-muHA}]
Assume $\muPA \vdash \forall x\exists yA(x,y)$ for $A$ quantifier-free.  There is a primitive recursive function symbol $f$ such that $\vdash A(x,y)\leftrightarrow f(x,y)=0$ in both $\muPA$ and $\muHA$, so it suffices to show $\muHA \vdash \forall x\exists yf(x,y)=0$. We set $\rho\dfn \exists y .f(x,y)=0$.  Of course, $\muPA \vdash \rho$, so $\muHA \vdash \neg \neg \rho$ by Lemma~\ref{lem:muPA-in-muHA}. By Lemma~\ref{lem:friedman} we have  $\muHA \vdash (\friedman{\rho}\to \rho) \to \rho$, where $\friedman{\rho}= \exists y(f(x,y)=0 \vee \exists y.(x,y)=0)$, which is clearly equivalent to $\rho$. Thus, we have $\muHA \vdash (\rho\to \rho)\to \rho$, and so $\muHA\vdash \rho$. We conclude by generalising over $x$.
\end{proof}

\end{document}



